%% file: cphe.tex
\documentclass[a4paper,11pt,unpublished]{quantumarticle}
\pdfoutput=1

\makeatletter
\providecommand{\@afterenddocumenthook}{}
\def\@hangfroms@section#1#2{\noindent#1#2}
\makeatother

\usepackage[utf8]{inputenc}
\usepackage[english]{babel}
\usepackage[T1]{fontenc}

\usepackage{amsmath,mathtools,amsthm}
\usepackage{amssymb}
\usepackage{dsfont}

\usepackage{xurl}
\usepackage[
    colorlinks,
    linkcolor={black!30!blue},
    citecolor={black!30!blue},
    urlcolor={black!30!blue}
]{hyperref} 
\usepackage{hyperxmp} 

\usepackage{tikz}
\usepackage{lipsum}

\usepackage{aliascnt} 
\usepackage{cleveref}
\crefname{algocf}{algorithm}{algorithms}
\Crefname{algocf}{Algorithm}{Algorithms}
\crefname{appendix}{appendix}{appendices}
\Crefname{appendix}{Appendix}{Appendices}
\crefname{theorem}{theorem}{theorems}
\Crefname{theorem}{Theorem}{Theorems}
\crefname{lemma}{lemma}{lemmas}
\Crefname{lemma}{Lemma}{Lemmas}
\crefname{conjecture}{conjecture}{conjectures}
\Crefname{conjecture}{Conjecture}{Conjectures}
\crefname{corollary}{corollary}{corollaries}
\Crefname{corollary}{Corollary}{Corollaries}
\crefname{problem}{problem}{problems}
\Crefname{problem}{Problem}{Problems}

\usepackage{siunitx}
\usepackage{booktabs}

\usepackage[numbers, sort&compress]{natbib}

\usepackage{enumitem}
\usepackage[nolist]{acronym}
\newacro{LP}{linear program}
\newacro{SDP}{semidefinite program}

\usepackage[ruled]{algorithm2e}
\usepackage{calc}

\newtheorem{theorem}{Theorem}
\newaliascnt{lemma}{theorem}
\newtheorem{lemma}[lemma]{Lemma}
\aliascntresetthe{lemma}
\newaliascnt{conjecture}{theorem}

\aliascntresetthe{conjecture}
\newaliascnt{corollary}{theorem}
\newtheorem{corollary}[corollary]{Corollary}
\aliascntresetthe{corollary}
\newaliascnt{proposition}{theorem}
\newtheorem{proposition}[proposition]{Proposition}
\aliascntresetthe{proposition}
\newaliascnt{problem}{theorem}
\newtheoremstyle{problemstyle}
  {\topsep}{\topsep}
  {\itshape}{0pt}
  {\bfseries}{.}
  { }{}
\theoremstyle{problemstyle}
\newtheorem{problemenv}[problem]{Problem}

\newlist{probspec}{description}{1}
\setlist[probspec]{
  nosep, topsep=0pt, partopsep=0pt,
  align=left, labelsep=0.5em,
  labelwidth=\widthof{\normalfont Question}, leftmargin=\widthof{\normalfont Question}+0.5em,
  font=\normalfont
}
\newcommand{\instance}{\item[Instance]}
\newcommand{\question}{\item[Question]}
\newcommand{\task}{\item[Task]}

\NewDocumentEnvironment{problem}{o}
  {\IfValueTF{#1}{\begin{problemenv}[#1]}{\begin{problemenv}}%
    \leavevmode\par\nobreak
    \begin{probspec}}
  {\end{probspec}\end{problemenv}}
\aliascntresetthe{problem}

\input{definitions}
\AtBeginDocument{\normalfont}

\newcommand{\tuhh}{Institute for Quantum Inspired and Quantum Optimization, Hamburg University of Technology, Germany}

\begin{document}

    \title{Fast Hamiltonian engineering from cut polytope geometry}

    \author{Thomas Joachim Friese}
    \email{thomas.friese@tuhh.de}
    \affiliation{\tuhh}
    \orcid{0009-0009-0495-2635}

    \author{Özgün Kum}
    \affiliation{\tuhh}
    \orcid{0000-0003-1466-4900}

    \author{Aram W.\ Harrow}
    \affiliation{Center for Theoretical Physics -- a Leinweber Institute, Massachusetts Institute of Technology, USA}
    \orcid{0000-0003-3220-7682}

    \author{Martin Kliesch}
    \email{martin.kliesch@tuhh.de}
    \affiliation{\tuhh}
    \orcid{0000-0002-8009-0549}

    \maketitle

    \hypersetup{
        pdftitle = {Fast Hamiltonian engineering from cut polytope geometry},
        pdfauthor = {Thomas Joachim Friese, Özgün Kum, Aram W. Harrow, Martin Kliesch},
        pdfsubject = {Quantum simulation},
        pdfkeywords = {Hamiltonian engineering, quantum simulation, analog quantum simulator, local control, pulse sequences,
            cut polytope, complex cut polytope, elliptope, semidefinite relaxation, Krivine rounding, randomized rounding,
            linear programming, NP-hardness, approximation algorithm, quantum evolution time, time-optimal gates,
            qubits, qudits, Pauli conjugation, clock model, chiral clock model, fermions, Harper--Hofstadter model,
            artificial gauge fields, Fermi--Hubbard model, Wendel's theorem, average Hamiltonian theory}
    }
    \begin{abstract}
        Hamiltonian engineering (HE) simulates quantum dynamics under a target Hamiltonian using a native entangling system Hamiltonian and restricted control, with applications from quantum gate design to analog quantum simulation.
        Observing that, for many systems, formulating HE as a linear program only requires specific commutation relations between the control pulses and the system Hamiltonian, we provide a unified framework for broad classes of $2$-local qubit, qudit, and fermionic systems.
        We reformulate time-optimal HE as a complex $k$-cut polytope problem, with $k$ the number of distinct phases in these commutation relations, and prove it NP-complete for every finite $k$.
        We therefore relax the polytope to the elliptope and apply a Krivine-type rounding, yielding pulses informed by the system and target Hamiltonians, which we then use in a linear program.
        We derive how many informed pulses are necessary and sufficient for the linear program to be reliably feasible.
        Additionally mixing in uniformly sampled pulses guarantees a solution close to the relaxation value from $\mathrm{O}(m)$ pulses, with $m$ the number of interaction terms.
        Together with upper and lower bounds on the optimal quantum run time, this yields an $\mathrm{O}(\sqrt{m})$ approximation ratio.
        In benchmarks on a fully connected Ising model, a chiral clock model for qudits, and the fermionic Harper--Hofstadter model, our approach attains near-optimal quantum run times wherever the optimum is computable, reaches run times that saturate with the lattice size in the fermionic case, and outperforms state-of-the-art methods with $\mathrm{O}(m)$ pulses.
        This establishes a unified approach with provable guarantees to the automatic programming of analog quantum simulators.
    \end{abstract}

    \section{Introduction}
    Simulating quantum many-body systems is broadly considered the most promising arena for quantum devices to deliver a practical advantage over classical computation~\citep{georgescu2014,daley2022}.
    Analog quantum simulators~\citep{bloch2012,altman2021,flannigan2022,trivedi2024} have emerged as powerful devices enabling the simulation of native qubit and qudit Hamiltonians, as well as fermionic and bosonic ones.
    However, the class of interactions which can be natively implemented on such devices are  restricted by experimental constraints. 
    Hamiltonian engineering allows simulating broad classes of interactive dynamics by using natively available control operations. 
    An important crux is limiting the amount of noise acting on the system. 
    Therefore, compiling control operations to minimize the total operation time is a crucial challenge for accurate quantum simulation, since long operation times expose quantum devices to decoherence. 
    
    Hamiltonian engineering by sequences of local pulses has a long history in nuclear magnetic resonance, where average Hamiltonian theory~\citep{haeberlen1968} and dynamical decoupling~\citep{viola1999} describe how interleaving the native evolution with fast control operations shapes the effective interaction.
    This approach has been extended to programmable simulators of spin models with trapped ions~\citep{hayes2014}, nitrogen-vacancy ensembles~\citep{choi2017,choi2020,zhou2024}, and Rydberg atoms~\citep{geier2021,scholl2022}, and, in the language of quantum information, to the question of which Hamiltonians can be simulated from a given one with local unitaries~\citep{lloyd1996,dodd2002,nielsen2002,cubitt2018}.
    The time-optimal version of this question is classical for two qubits~\citep{dur2001,khaneja2001,bennett2002}, and constructions based on Hadamard matrices~\citep{leung2002} and Walsh sequences~\citep{hayes2014,votto2024} realize Ising-type targets on many qubits with sequences whose lengths scale with the system size.
    For $2$-local qubit Hamiltonians conjugated by layers of Pauli gates, an explicit construction~\citep{garciadeandoin2026} and a bound on the total analog time~\citep{garciadeandoin2026a} are available.
    Finally, an automatically programmable framework has been developed, first in the case of Ising Hamiltonians \cite{bassler2023,bassler2024} and Pauli-X pulse layers and the extended to an efficient and robust framework for arbitrary qubit Hamiltonians and pulse layers given by local Clifford gates \cite{bassler2025}. 
    This framework has been extended to fermionic systems recently \cite{kum2026}. 
    
    The Hamiltonian engineering approaches~\citep{bassler2023,bassler2024,bassler2025,kum2026} conjugate time evolution under a global native Hamiltonian with layers of local operations, which are called \emph{pulses}, to implement an effective time evolutions under a given target Hamiltonian. 
    One can efficiently find a suboptimal set of pulses (single-qubit gate layers \citep{bassler2025} or layers of single-mode fermionic unitaries \citep{kum2026}) by solving a ``relaxed'' \ac{LP}, which otherwise scales exponentially with the number of qubits or fermionic modes.
    This \emph{efficient relaxation} is achieved via \emph{uninformed sampling}, i.e.\ uniformly and randomly sampling from the set of all possible pulses, and it is observed that simply sampling more layers reduces the total quantum run time while increasing the time required for the classical precomputation.
    This method is the state of the art against which we compare throughout, and we refer to it as the \emph{uninformed \ac{LP} relaxation}.

    In this work, we develop a unified framework for Hamiltonian engineering on broad classes of $2$-local qubit, qudit, and fermionic systems\footnote{For qubits and qudits, we require the symplectic binary representation of all system terms to have Hamming weight $2$. For fermions, the engineered interactions must not contain number operators.}. 
    It applies whenever a pulse acts on each interaction term of the system Hamiltonian only by a relative phase between the two involved sites, with the phases restricted to $k$ equally spaced values, where we explicitly also allow for $k=\infty$, i.e., arbitrary phases. 
    This phase commutation relation is the common structure of multi-qubit gates, Pauli conjugations, and fermionic systems, which were previously treated separately, and it reduces all of them to a specific geometric problem. 
    The optimal quantum run time is determined by the point at which a ray, fixed by the system and target couplings, leaves the complex $k$-cut polytope.
    In this formulation, established tools from combinatorial optimization become directly applicable.

    We prove that deciding whether a given point on this ray lies in the polytope is \NP-hard for every $k$ and \NP-complete for finite $k$, extending the known case $k = 2$~\citep{bassler2024}.
    We therefore relax the polytope to the elliptope and apply Krivine-type rounding, inverting the systematic distortion of the randomized rounding before sampling.
    This yields pulses that are \emph{informed} of the system and target Hamiltonians, in contrast to the uniformly sampled pulses of the uninformed \ac{LP} relaxation.

    From these pulses, we construct two algorithms.
    The \emph{informed \ac{LP} algorithm} restricts the \ac{LP} to the informed pulses.
    We derive a necessary number of pulses for its robust feasibility, as well as a sufficient one, which depends on the instance and is in general not linear in the number $m$ of interaction terms.
    On typical instances, we observe numerically that the probability of finding feasible, i.e., sufficiently expressible pulses sets in sharply at the necessary number of pulses. 
    To obtain a guarantee on every instance, the \emph{mixed \ac{LP} algorithm} additionally samples uniform pulses and returns with high probability a feasible solution close to the relaxation value from $\LandauO(m)$ pulses.
    Moreover, it allows finite pulse time errors \cite{bassler2025} to be suppressed.
    For both algorithms, we derive an upper bound on the quantum run time, which on the qubit Hamiltonians captured by our framework has the same form as the best known bound on the optimal solution~\citep{garciadeandoin2026a}.
    We complement it with a lower bound on the optimal quantum run time, which generalizes a qubit result of \citet{bassler2024}.
    To our knowledge, both are the first such bounds for qudits and fermions, and together they yield an approximation ratio of $\LandauO(\sqrt{m})$.

    We benchmark the informed \ac{LP} algorithm on qubit, qudit, and fermionic systems against the uninformed \ac{LP} relaxation and, for qubits, against the explicit construction of \citet{garciadeandoin2026}.
    Using a number of pulses linear in $m$, it performs far better than the worst-case guarantees suggest and attains near-optimal quantum run times wherever the optimum is computable.
    On the fermionic Hofstadter model, its quantum run time does not grow with the lattice size, whereas that of the uninformed \ac{LP} relaxation grows linearly in the linear lattice size.

    The rest of this paper is organized as follows.
    \Cref{sec:notation} introduces the notation used throughout this work.
    \Cref{sec:engineering} formulates the Hamiltonian engineering task together with the phase commutation framework and illustrates it with motivating examples.
    In \cref{sec:approx_algo}, we recast the task as an optimization over the complex $k$-cut polytope, establish its hardness, derive the semidefinite relaxation and rounding, introduce the informed and mixed \ac{LP} algorithms, and bound their quantum run times.
    \Cref{sec:applications} applies the framework to qubits, qudits, and fermions and presents the numerical benchmarks.
    We conclude and discuss open questions in \cref{sec:conclusion}.
    The appendices contain extended proofs (\cref{sec:hardness,sec:expectation,sec:feasibility_proofs,sec:mixed_proofs}), a numerical study of the feasibility transition (\cref{sec:feasibility_numerics}), and details on the numerical methods (\cref{sec:numerics}).

    \section{Notation}
    \label{sec:notation}
    We denote the sets of natural, real, and complex numbers by $\NN$, $\RR$, and $\CC$, respectively.
    For any natural number $k$, $\NN_{\geq k}$ denotes the set of natural numbers greater or equal to $k$.
    The set of integers from $1$ to $k$ is defined to be $[k] \coloneqq \Set{1, \dots, k}$.
    We define the complex conjugate of $z \in \CC$ to be $\conj{z}$, and the principal $k$-th root of unity to be
    \begin{equation}
        \uroot_k \coloneqq \exp \left( \i \frac{2 \pi}{k} \right).
    \end{equation}
    The set of all $k$-th roots of unity is represented by
    \begin{equation}
        \urootset_k \coloneqq \Set{\uroot_k^j \given j \in [k]}.
    \end{equation}
    Moreover, for ease of notation we denote the complex unit circle by 
    \begin{equation}
        \urootset_{\infty} \coloneqq \Set{z \in \CC \given \abs{z} = 1},
    \end{equation}
    and the closed unit disk by
    \begin{equation}
        \DD \coloneqq \Set{z \in \CC \given \abs{z} \leq 1}.
    \end{equation}

    Vectors are denoted by bold letters, e.g.\ $\vec{x}$, their entries by the corresponding non-bold subscripted letter, e.g.\ $x_i$, and matrices by uppercase letters, e.g.\ $A$.
    The set of indices with nonzero entries in a vector $\vec{x}$, or more generally the support of a function $f$, is defined to be $\nz(\vec{x})$ and $\nz(f)$, respectively.
    For $p \in [1, \infty]$, we denote by $\lpnorm{\vec{x}}$ the $\ell_p$-norm of a vector $\vec{x}$, and by $\lpnorm{f} \coloneqq \big( \sum_{c \in \nz(f)} \abs{f(c)}^p \big)^{1/p}$ that of a function $f$ with finite support, where $\lpnorm[\infty]{f} \coloneqq \max_{c \in \nz(f)} \abs{f(c)}$.
    For a matrix $A$, $\lpnorm{A}$ denotes the entrywise $\ell_p$-norm, i.e.\ the $\ell_p$-norm of $A$ regarded as a vector, and $\pnorm{A}$ denotes the Schatten $p$-norm, i.e.\ the $\ell_p$-norm of the vector of singular values of $A$.
    In particular, we use the spectral norm, the Frobenius norm, and the entrywise maximum norm, given by
    \begin{equation}
        \snorm{A} = \max_{\lpnorm[2]{\vec{x}} = 1} \lpnorm[2]{A \vec{x}},
    \end{equation}
    \begin{equation}
        \fnorm{A} = \sqrt{\Tr(A^\dagger A)} = \lpnorm[2]{A},
    \end{equation}
    \begin{equation}
        \lpnorm[\infty]{A} = \max_{i,j} \abs{A_{ij}}.
    \end{equation}
    If $A$ is positive semidefinite, we indicate this property by $A \succeq 0$. 
    We write $\Herm_n$ for the Hermitian $n \times n$ matrices, $\Herm_n^{\vec{0}}$ and $\Herm_n^{\vec{1}}$ for the Hermitian matrices whose diagonal is zero (\emph{hollow}) and all-ones, respectively, $\Herm_n(\KK)$ for those with entries in a subfield $\KK \subseteq \CC$, $\Sym_n$ for real symmetric matrices, $\innerp{A}{B} \coloneqq \Tr(A^\dagger B)$ for the Hilbert--Schmidt inner product, and $\arg(z) \in [0, 2\pi)$ for the argument of $z \in \CC \setminus \Set{0}$.
    For a function $f$ on a set $\mathcal{C}$, $f^{\elementwise}(\vec{x})$ and $f^{\elementwise}(A)$ denote the elementwise application of this function on a vector $\vec{x} \in \mathcal{C}^{n}$ and a matrix $A \in \mathcal{C}^{n \times m}$, respectively.

    Finally, $\1, X, Y, Z$ will usually refer to the single-qubit Pauli matrices, unless stated otherwise.
    The fermionic creation and annihilation operators are given by $c^\dagger$ and $c$, and the fermionic number operator is defined to be $c^\dagger c$. 
    Whenever a single-site operator is supposed to act on a specific site of a larger system, we denote this by subscripting with the site's index, e.g.\ $X_i$ is Pauli $X$ acting on site $i$. 

    \section{Hamiltonian engineering} 
    \label{sec:engineering}
    We start with the general formulation of our Hamiltonian engineering framework and then explain it further with examples. 
    \subsection{General formulation}
    We consider the task of simulating the time evolution of a target Hamiltonian $H_T$ using a native system Hamiltonian $H_S$ and access to a set of unitary pulses $\Set{U_{\vec{\theta}}}_{\vec{\theta} \in \Theta}$ with parameter set $\Theta$.
    To achieve this, we seek a decomposition of the target Hamiltonian $H_T$ as a finite conical combination of conjugations of the system Hamiltonian by the pulses,
    \begin{equation}\label{eq:decomposition}
        H_T = \sum_{\vec{\theta} \in \nz(\vec{\lambda})} \lambda(\vec{\theta}) U_{\vec{\theta}}^\dagger H_S U_{\vec{\theta}},
    \end{equation}
    with $\vec{\lambda}: \Theta \to \RR_{\geq 0}$ having finite support, and
    \begin{equation}
        \sum_{\vec{\theta} \in \nz(\vec{\lambda})} \lambda(\vec{\theta}) = \lpnorm[1]{\vec{\lambda}}.
    \end{equation}
    Then the Trotter--Suzuki formula~\citep{trotter1959a,suzuki1991} yields a protocol to synthesize the time evolution of $H_T$ as
    \begin{equation}\label{eq:trotter}
        \begin{split}
            &\e^{-\i t H_T} \approx \prod_{\vec{\theta}} U_{\vec{\theta}}^\dagger \e^{-\i \lambda(\vec{\theta}) t H_S} U_{\vec{\theta}}.
        \end{split}
    \end{equation}
    The product in \cref{eq:trotter} reproduces $\e^{-\i t H_T}$ exactly only if the conjugated summands in \cref{eq:decomposition} commute pairwise, and otherwise incurs a Trotter error, which is controlled in the standard way by repeating the pulse sequence~\citep{bassler2025,kum2026}.
    Splitting the evolution into $N_{\mathrm{Tro}}$ cycles of duration $t / N_{\mathrm{Tro}}$, the first-order formula in \cref{eq:trotter} approximates the target evolution up to a spectral-norm error $\LandauO \left( (\lpnorm[1]{\vec{\lambda}} \snorm{H_S} t)^2 / N_{\mathrm{Tro}} \right)$~\citep{childs2021}, because the summands are bounded by 
    \begin{equation}
        \sum_{\vec{\theta} \in \nz(\vec{\lambda})} \snorm{\lambda(\vec{\theta}) U_{\vec{\theta}}^\dagger H_S U_{\vec{\theta}}} = \lpnorm[1]{\vec{\lambda}} \snorm{H_S}.
    \end{equation}
    Higher-order Trotter--Suzuki formulas suppress this error further, with the recursively constructed formula of order $2p$~\citep{suzuki1991} attaining an error in $\LandauO \left( (\lpnorm[1]{\vec{\lambda}} \snorm{H_S} t)^{2p+1} / N_{\mathrm{Tro}}^{2p} \right)$~\citep{childs2021}, so that any accuracy $\epsilon > 0$ is reached with 
    \begin{equation}
        N_{\mathrm{Tro}} \in \LandauO \left( \lpnorm[1]{\vec{\lambda}} \snorm{H_S} t \left( \lpnorm[1]{\vec{\lambda}} \snorm{H_S} t / \epsilon \right)^{1 / 2p} \right)
    \end{equation}
    cycles.
    For the first- and second-order formulas, the accumulated evolution time under $H_S$ equals $\lpnorm[1]{\vec{\lambda}} t$ irrespective of $N_{\mathrm{Tro}}$, so that the Trotter error is suppressed at the expense of the number of pulses, which grows as $\LandauO ( 5^{p} \abs{\nz(\vec{\lambda})} N_{\mathrm{Tro}} )$, and not at the expense of the quantum run time.
    Beyond the second order, the formulas necessarily contain negative time steps~\citep{suzuki1991}, whose realization requires access to the reversed evolution under $H_S$ and inflates the quantum run time by a constant factor depending on $p$.
    In all cases, the error is governed by the same quantity $\lpnorm[1]{\vec{\lambda}}$ that determines the quantum run time, so that minimizing the latter simultaneously tightens the Trotter approximation.
    If the conjugated summands commute, no Trotter error arises for any ordering of the factors.
    If additionally the noise during the evolution under $H_S$ is described by a Lindbladian that commutes with the evolution under $H_S$ and is invariant under conjugation with the pulses, then the generators of all noisy factors in \cref{eq:trotter} commute, and the protocol implements the target dynamics subject to the same noise at a rate amplified by $\lpnorm[1]{\vec{\lambda}}$.
    If the conjugated summands do not commute, pulse sequences with equal quantum run times may differ in their Trotter errors and have to be assessed beyond $\lpnorm[1]{\vec{\lambda}}$, for instance by numerical simulation as in~\citep{bassler2025,kum2026}.
    Our primary goal is to determine a feasible coefficient assignment $\vec{\lambda}$ that minimizes the overhead in quantum run time, given by $\lpnorm[1]{\vec{\lambda}}$.

    We observe that for a wide range of physical systems, the Hamiltonians can be decomposed into ``simple'' operators $O_{ij}=O_{ji}^\dagger$ as
    \begin{equation}
        H_S = \sum_{i \not = j \in [n]} A_{ij} O_{ij},
        \quad
        H_T = \sum_{i \not = j \in [n]} B_{ij} O_{ij},
    \end{equation}
    with Hermitian matrices $A, B \in \CC^{n \times n}$ for some $n \in \NN$, where what ``simple'' means depends on the type of system.
    For $i_1 < j_1$, $i_2 < j_2$, we additionally assume linear independence of $O_{i_1j_1}$ and $O_{i_2j_2}$.
    For instance, we may have $n$-site lattices, with the operators $O_{ij}$ corresponding to $2$-local interactions.
    
    Our solution framework is then defined for any system where the operators $O_{ij}$ commute up to equally spaced phases with the applied unitaries.
    Specifically, we require that the unitaries can be parameterized by vectors $\vec{\theta} \in \Theta_k^n$ for some constant $k \in \NN_{\geq 2} \cup \Set{\infty}$, and fulfill
    \begin{equation}
        \label{eq:phase_commutation}
        O_{ij} U_{\vec{\theta}} = \e^{\i (\theta_i - \theta_j)} U_{\vec{\theta}} O_{ij}.
    \end{equation}
    The phase set $\Theta_k$ is defined as
    \begin{equation}
        \Theta_k \coloneqq 
        \begin{cases}
            \Set*{ \frac{j}{k} 2\pi \given j \in \Set{0, \dots, k-1} }, & \textnormal{if } k \in \NN_{\geq 2}, \\
            [0, 2 \pi), & \textnormal{if } k = \infty, 
        \end{cases}
    \end{equation}
    corresponding to the unit circle subset $\urootset_k$. 
    Such a unitary $U_{\vec{\theta}}$ can be thought of as a layer of local operations on the lattice, where each entry $\theta_i$ describes the phase acquired by the $i$-th lattice site.

    Whenever this relation is satisfied, the sought-after decomposition \eqref{eq:decomposition} is fulfilled if
    \begin{equation}\label{eq:BinTermsOfA}
        B_{ij} = \sum_{\vec{\theta} \in \nz(\vec{\lambda})} A_{ij} \lambda(\vec{\theta}) \e^{\i (\theta_i - \theta_j)}
    \end{equation}
    for all $i \not = j \in [n]$.
    Clearly, ${A_{ij} = 0 \implies B_{ij} = 0}$ is a necessary condition for the existence of a feasible solution.
    Wherever $A_{ij} \not = 0$, we define $M$ as the elementwise quotient of $B$ and $A$, i.e.\
    \begin{equation}
        M_{ij} \coloneqq \frac{B_{ij}}{A_{ij}}
    \end{equation}
    for all $(i, j) \in \nz(A)$.
    The remaining constraints on a valid decomposition \eqref{eq:BinTermsOfA} are then
    \begin{equation}
        \label{eq:decomposition_constraints}
        \sum_{\vec{\theta} \in \nz(\vec{\lambda})} \lambda(\vec{\theta}) e^{\i (\theta_i - \theta_j)} = M_{ij}
    \end{equation}
    for all $(i, j) \in \nz(A)$.
    Combining this with the minimization of $\lpnorm[1]{\vec{\lambda}}$ leads to a \acf{LP}:
    \begin{equation*}
        \tag{GeneralLP}
        \label{eq:general_lp}
        \begin{alignedat}{3}
            &\textnormal{min} \enspace &&\lpnorm[1]{\vec{\lambda}} \\
            &\textnormal{s.t.} \enspace &&\sum_{\mathclap{\vec{\theta} \in \nz(\vec{\lambda})}} \lambda(\vec{\theta}) \e^{\i (\theta_i - \theta_j)} = M_{ij}, \enspace && \forall (i,j) \in \nz(A), \\
            &&& \lambda(\vec{\theta}) \geq 0, && \forall \vec{\theta} \in \Theta_k.
        \end{alignedat}
    \end{equation*}
    For special cases of systems and unitaries, this \ac{LP} was studied extensively by \citet{bassler2023,bassler2024,bassler2025,kum2026}.
    Because it has an exponential number of variables in $n$ for finite $k$, and an infinite number of variables if $k = \infty$, they \cite{bassler2025,kum2026} rely on uniformly sampling a small subset of pulses $B \subseteq \Theta_k$ and restrict the \ac{LP} to these.
    Numerical results show that for $m \coloneqq \abs{\nz(A)} / 2$ system terms, $\LandauO(m)$ sampled pulses suffice to find an exact solution to the \ac{LP}.
    However, because the sampling is uninformed of the specific system and target Hamiltonians, this relaxation yields excessive quantum run times.

    We take an alternative approach.
    By expressing the constraint as a ray over a polytope, we efficiently approximate the minimization of $\lpnorm[1]{\vec{\lambda}}$.
    Our algorithms are detailed in \cref{sec:approx_algo}.

    The introduced quantities have direct physical readings.
    The matrices $A$ and $B$ collect the coupling strengths of the system and the target Hamiltonian, so the entries of $M$ state by which factor each interaction available in the system has to be reweighted.
    Since conjugating $O_{ij}$ by $U_{\vec{\theta}}$ multiplies it by the phase $\e^{\i (\theta_i - \theta_j)} \in \urootset_k$, the attainable reweightings are exactly the conical combinations of such phases, and $\lpnorm[1]{\vec{\lambda}}$ is the total weight spent on them.
    The parameter $k$ counts the equally spaced phases that the local operations can imprint, ranging from coarse, discrete pulse sets for small $k$ to continuously tunable single-site rotations for $k = \infty$.
    Complex entries $M_{ij}$ therefore require $k > 2$, and we assume throughout that $M$ is real whenever $k = 2$.
    We further exclude the trivial case $M = 0$, for which $\vec{\lambda} = 0$ is optimal.

    So far, we have treated the pulses $U_{\vec{\theta}}$ as instantaneous.
    On actual hardware, every pulse has a finite duration $t_{\mathrm{p}} > 0$, during which the native Hamiltonian $H_S$ is always on and keeps acting on the system.
    The conjugations in \cref{eq:decomposition} are therefore not realized exactly.
    The average Hamiltonian of a pulse block acquires an additional \emph{finite pulse time error} term of first order in $t_{\mathrm{p}} \snorm{H_S}$, which accumulates over the pulses of the sequence~\citep{votto2024,bassler2025}.
    After the Trotter error, this is the dominant error source of the protocols considered here.
    Moreover, suppressing the Trotter error by increasing the number of cycles $N_{\mathrm{Tro}}$, or by passing to a higher-order Trotter--Suzuki formula, multiplies the number of applied pulses and thereby inflates the accumulated finite pulse time error.
    We return to this error, and to its interplay with our sampling scheme, at the end of \cref{sec:mixed}.

    \subsection{Motivating examples}
    \label{sec:motivating_examples}
    The simplest example of a pairing of Hamiltonians and unitaries satisfying the outlined properties is a $Z$-type Ising model with $X$-type conjugations.
    That is,
    \begin{equation}
        O_{ij} = Z_i Z_j
    \end{equation}
    on $n$ qubits, as well as
    \begin{equation}
        U_{\vec{\theta}} = \prod_{i=1}^{n} \e^{-\i \theta_i (\1 - X_i) / 2}
        = \prod_{i=1}^{n} \begin{cases}
            X_i, & \textnormal{if } \theta_i = \pi, \\
            \1, & \textnormal{if } \theta_i = 0,
        \end{cases}
    \end{equation}
    for $\vec{\theta} \in \Theta_2^n$.
    It is easy to check that \cref{eq:phase_commutation} is satisfied here because of the standard Pauli anticommutations.
    This setting for Hamiltonian engineering was extensively investigated in~\citep{bassler2023,bassler2024}.

    Using the symplectic representation of the Pauli group, we can naturally extend this formulation to also allow for other Pauli terms in the Hamiltonians.
    For $n = 2q + 1$ and $i \in [n]$, denote by
    \begin{equation}
        \label{eq:examples_single_pauli}
        P_i \coloneqq \begin{cases}
            \1, & \textnormal{if } i = 2q+1, \\
            X_i, & \textnormal{if } i \leq q, \\
            Z_{i - q}, & \textnormal{if } q+1 \leq i \leq 2q,
        \end{cases}
    \end{equation}
    a single $X$ or $Z$ operator on $q$ qubits.
    Then, we may choose 
    \begin{equation}
        \label{eq:examples_single_interaction}
        O_{ij} =
        \begin{cases}
            P_i P_j, & \textnormal{if } [P_i, P_j] = 0, \\
            \i P_i P_j, & \textnormal{if } \{P_i, P_j\} = 0 \textnormal{ and } i < j, \\
            -\i P_i P_j, & \textnormal{if } \{P_i, P_j\} = 0 \textnormal{ and } i > j.
        \end{cases}
    \end{equation}
    Notice that thus $O_{ij}$ can be either any $2$-local interaction of $X$ and $Z$ operators or a single $X$, $Y$, or $Z$ operator.
    Pairs with a $Y$ and some other nontrivial operator cannot be encoded in this way, so we cannot capture all $2$-local qubit Hamiltonians.
    The valid system terms are exactly those, that have at most Hamming weight $2$ in their binary symplectic representation.
    The unitaries are given accordingly by Pauli strings
    \begin{equation}
        \begin{aligned}
            U_{\vec{\theta}} 
            &= \prod_{i=1}^{q} \e^{-\i \theta_{i} (\1 - Z_i) / 2} \e^{-\i \theta_{i + q} (\1 - X_i) / 2}\\
            &= \prod_{i=1}^{q} \begin{cases}
                \1, & \textnormal{if } \theta_i = 0, \theta_{i+q} = 0, \\
                X_i, & \textnormal{if } \theta_i = 0, \theta_{i+q} = \pi, \\
                Z_i, & \textnormal{if } \theta_i = \pi, \theta_{i+q} = 0, \\
                \i Y_i, & \textnormal{if } \theta_i = \pi, \theta_{i+q} = \pi,
            \end{cases}
        \end{aligned}
    \end{equation}
    for $\vec{\theta} \in \Theta_2^{2q} \times \Set{0}$.
    Here, we fixed the phase on the dummy index $2q+1$ to $0$, which enables us to engineer also the single-qubit terms.
    We can do this without violating our framework, because \cref{eq:phase_commutation} is global phase invariant, so we may always shift all phases such that $\theta_{2q+1} = 0$.
    This generalized setting of Hamiltonian engineering for Pauli Hamiltonians with Pauli conjugations was, without our locality constraints, investigated in~\citep{bassler2025}.
    We discuss this and the former application of our method to qubit systems in greater detail in \cref{sec:qubit_application}.

    By replacing the notion of Pauli operators with the generalized shift and clock operators, we find the corresponding formulation for $k$-level qudits.
    That is, we define
    \begin{equation}
        X = \sum_{j=0}^{k-1} \ketbra{j+1 \mod k}{j},
        \quad
        Z = \sum_{j=0}^{k-1} \uroot_k^{j} \ketbra{j}{j}
    \end{equation}
    as acting on a single qudit.
    $P_i$ is then defined as in \cref{eq:examples_single_pauli}, using the redefined $X$ and $Z$ operators.
    The interaction terms are given on $q$ qudits with $n = 2q+1$ by 
    \begin{equation}
        O_{ij} = P_i P_j^\dagger.
    \end{equation}
    We can thus again encode all two-term interactions of $X$ and $Z$ operators, but not the full Weyl--Heisenberg group.
    For generating the unitaries, consider the qudit number and phase operators
    \begin{equation}
        N = \sum_{j=0}^{k-1} j \ketbra{j}{j},
        \quad
        \Phi = \sum_{\ell=0}^{k-1} \ell \ketbra{\phi_{\ell}}{\phi_{\ell}},
    \end{equation}
    with the Fourier basis
    \begin{equation}
        \ket{\phi_{\ell}} = \frac{1}{\sqrt{k}} \sum_{j=0}^{k-1} \uroot_k^{-\ell j} \ket{j}.
    \end{equation}
    Then,
    \begin{equation}
     \e^{\i \frac{2\pi}{k} \Phi} = X,
        \quad
     \e^{\i \frac{2\pi}{k} N} = Z,
    \end{equation}
    are the shift and clock operators, so we take
    \begin{equation}
        U_{\vec{\theta}} = \prod_{i=1}^{q} \e^{-\i \theta_{i} N_i} \e^{-\i \theta_{i+q} \Phi_i},
    \end{equation}
    for $\vec{\theta} \in \Theta_k^{2q} \times \Set{0}$.
    Because of the relation
    \begin{equation}
        Z X = \uroot_k X Z,
    \end{equation}
    it is easy to see that \cref{eq:phase_commutation} is satisfied.
    More generally, for qudits with $d \geq k$ levels, 
    restricting $\vec{\theta}$ to $\Theta_k^{2q} \times \Set{0}$ satisfies \cref{eq:phase_commutation} for every $k$ dividing $d$, since $\Theta_k \subseteq \Theta_d$.
    Conjugations of this kind, with the goal to engineer qudit systems, were studied in~\citep{alvarez-ahedo2025}.
    We provide our full discussion in \cref{sec:qudit_application}.

    As a final motivating example, consider a fermionic system with $n$ different modes.
    Then we take
    \begin{equation}
        O_{ij} = c_i^\dagger c_j,
    \end{equation}
    with the fermionic creation and annihilation operators $c^\dagger$ and $c$. 
    Taking inspiration from the qudit example, we may consider the exponentiation of the fermionic number operator, which fulfills
    \begin{equation}
        c^\dagger \e^{-\i \theta c^\dagger c} = \e^{\i \theta} \e^{-\i \theta c^\dagger c} c^\dagger,
        \quad
        c \e^{-\i \theta c^\dagger c} = \e^{-i \theta} \e^{-\i \theta c^\dagger c} c.
    \end{equation}
    Thus, we take
    \begin{equation}
        U_{\vec{\theta}} = \prod_{i=1}^{n} \e^{-\i \theta_i c_i^\dagger c_i}
    \end{equation}
    for $\vec{\theta} \in \Theta_k^n$, and \cref{eq:phase_commutation} is satisfied for all $k \in \NN_{\geq 2} \cup \Set{\infty}$.
    Hamiltonian engineering with this type of conjugations on fermionic systems was studied in~\citep{kum2026}.
    We provide a more detailed discussion in \cref{sec:fermion_application}.
    
    The previous examples were only concerned with \num{2}-local systems.
    This is expected, as the $O_{ij}$ naturally encode interactions on two sites $i$ and $j$.
    However, by encoding a superlattice over the system, higher-order terms may be engineered.
    This comes at the cost of losing control over some lower-order terms.
    As an example, consider the \num{4}-qubit Hamiltonian
    \begin{equation}
        H_{S} = Z_1 Z_2 + Z_3 Z_4 + Z_1 Z_2 Z_3 Z_4.
    \end{equation}
    Then, taking $n = 2$ with combined sites $(1, 2)$ and $(3, 4)$, together with $O_{12} = Z_1 Z_2 Z_3 Z_4$, and $U_{\vec{\theta}} = (X_1 X_2)^{\theta_1 / \pi} (X_3 X_4)^{\theta_2 / \pi}$ for $\vec{\theta} \in \Theta_2^2$ gives control over the \num{4}-local term, while keeping the \num{2}-local terms invariant.

    \section{Hardness and approximation algorithms}
    \label{sec:approx_algo}
    In this section, we develop the theoretical core of this work. 
    We first reformulate the \ac{LP} \eqref{eq:general_lp} geometrically: 
    Its constraints describe a ray in the space of Hermitian matrices, and the optimal quantum run time is determined by the point at which this ray leaves the complex $k$-cut polytope.
    Deciding whether a given point of the ray lies in this polytope turns out to be \NP-hard for every $k$, so that an exact solution is out of reach in general.
    In \cref{sec:relax}, we therefore relax the polytope to the elliptope and remove the systematic distortion of the standard randomized rounding by bending the ray.
    This relaxation yields pulses that are informed of the system and target Hamiltonians.
    \Cref{sec:informed,sec:mixed} turn these pulses into the informed and the mixed \ac{LP} algorithm and analyze how many pulses they require, and \cref{sec:bounds} bounds the quantum run times of both algorithms against the optimum.

    \subsection{Cut polytope formulation}
    We identify $\vec{\theta} \in \Theta_k^n$ with $\vec{x} \in \urootset_{k}^n$ by
    \begin{equation}
        x_i \coloneqq \e^{\i \theta_i}
    \end{equation}
    for all $i \in [n]$.
    Then the commutation factors correspond to the entries of the outer product matrix $\vec{x} \vec{x}^\dagger$, i.e.\
    \begin{equation}
        \e^{\i (\theta_i - \theta_j)} = \left( \vec{x} \vec{x}^\dagger \right)_{ij}.
    \end{equation}

    To illustrate the geometric motivation of our algorithm, assume that $A$ is nonzero for all off-diagonal entries $A_{ij}$, $i \not = j$.
    Then the constraint~\eqref{eq:decomposition_constraints} corresponds to the matrix equality
    \begin{equation}
        \sum_{\vec{x} \in \nz(\vec{\lambda})} \lambda(\vec{x}) \vec{x} \vec{x}^\dagger = \lpnorm[1]{\vec{\lambda}} \1 + M.
    \end{equation}
    Here, the identity term has to be added, because the diagonal of $M$ consists of only zeros, whereas each $\vec{x} \vec{x}^\dagger$ has only ones on the diagonal.
    Dividing the equation by $\lpnorm[1]{\vec{\lambda}}$ decouples the left- and right-hand side into independently parameterized matrices:
    \begin{equation}
        \label{eq:poly_ray}
        \sum_{\vec{x} \in \nz(\vec{\lambda})} \frac{\lambda(\vec{x})}{\lpnorm[1]{\vec{\lambda}}} \vec{x} \vec{x}^\dagger = \1 + \frac{1}{\lpnorm[1]{\vec{\lambda}}} M.
    \end{equation}
    On the left-hand side, any point in the convex hull
    \begin{equation}
        \conv \Set*{\vec{x} \vec{x}^\dagger \given \vec{x} \in \urootset_k^n} =: \CUT_k^n
    \end{equation}
    can be generated.
    We call $\CUT_k^n$ the complex $k$-cut polytope, following \citet{sinjorgo2024}, who introduced it for general $k$ and study its facets and semidefinite liftings.
    It is related to the MAX-$k$-CUT problem~\citep{frieze1997}, to complex quadratic optimization~\citep{zhang2006,so2007}, and intensively studied for the case of $k=2$, see e.g.~\citep{barahona1986,deza1997}.
    With the substitution $\gamma \coloneqq \frac{1}{\lpnorm[1]{\vec{\lambda}}} > 0$, the right-hand side $\1 + \gamma M$ of \cref{eq:poly_ray} describes a ray originating from $\1$ and propagating in the direction $M$. 
    We refer to this ray as the \emph{solution ray}.
    All feasible solutions lie in the intersection of the solution ray with $\CUT_k^n$, with the optimum being where the solution ray meets the boundary of the cut polytope, so where the total evolution time $\lpnorm[1]{\vec\lambda}$ is smallest. 
    This geometric formulation of the problem is visualized in \cref{fig:poly_ray}.

    \begin{figure}
        \centering
        \includegraphics{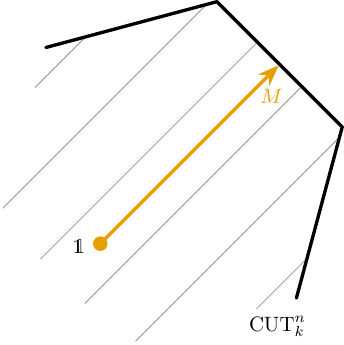}
        \caption[Visualization of the solution ray and the cut polytope.]{Visualization of the solution ray and the cut polytope.
                All solutions lie on the intersection of the ray originating from $\1$, going in the direction $M$, and $\CUT_k^n$.
                Better solutions lie further along the ray.
                The optimum is where the ray meets the boundary of $\CUT_k^n$.}
        \label{fig:poly_ray}
    \end{figure}

    If there are off-diagonal zero entries in $A$, we can extend $M$ to a full matrix by setting
    \begin{equation}
        \label{eq:effective_matrix}
        M_{ij} \coloneqq \begin{cases}
            B_{ij} / A_{ij}, & \textnormal{if } A_{ij} \not = 0, \\
            0, & \textnormal{if } A_{ij} = 0.
        \end{cases}
    \end{equation}
    Then, \cref{eq:poly_ray} no longer contains all feasible solutions, due to the added constraints where $A_{ij} = 0$.
    It however still forms the basis of our approximation algorithm.

    \Cref{eq:poly_ray} motivates a solution approach:
    First optimize the ray over $\CUT_k^n$, then convexly decompose the found point into rank-$1$ matrices.
    However, already maximizing a ray over $\CUT_k^n$ is \NP-hard.
    This was shown by \citet{bassler2024} for $k=2$, and we show the general case in the following.
    In \cref{sec:relax}, we instead introduce a semidefinite relaxation of the ray maximization.
    
    We first define the relevant families of decision problems.
    For both, we take $k$ to be an instance-independent constant, i.e.\ not part of the input.
    In the following, use $\uroot_{\infty} \coloneqq 1$, s.t.\ $\QQ(\uroot_{\infty}, \i) = \QQ(\i)$ are the Gaussian rationals. 
    The other cyclotomic fields are denoted by $\QQ(\uroot_k, \i)$. 

    \begin{problem}[$\CUT_k^n$ ray]
        \label{prob:cut_ray}
        \instance A matrix ${M \in \Herm_{n}^{\vec{0}}(\QQ(\uroot_k, \i))}$ and a distance $\gamma \in \QQ_{\geq 0}$.
        \question Is $\1 + \gamma M \in \CUT_k^n$?
    \end{problem}
    \noindent
    We will see in \cref{th:ray_membership_equiv} that \cref{prob:cut_ray} is equivalent to the following membership problem.

    \begin{problem}[$\CUT_k^n$ membership]
        \label{prob:cut_membership}
        \instance A matrix ${X \in \Herm_{n}^{\vec{1}}(\QQ(\uroot_k, \i))}$.
        \question Is $X \in \CUT_k^n$?
    \end{problem}

    It is well known that the $\CUT_2^n$ membership problem is \NP-complete:
    By \citet{pitowsky1991} the membership problem of the correlation polytope is \NP-complete, and via \citet{desimone1990} the correlation polytope on $n$ vertices is isomorphic to $\CUT_2^{n+1}$ via a polynomial-time computable map.
    See also \citep{caprara2026} for details on the combined statement.
    For $k > 2$, \citet{sinjorgo2024} state \NP-hardness of linear optimization over $\CUT_k^n$, but, to the best of our knowledge, no such statement regarding the membership problem of $\CUT_k^n$ exists currently in the literature.
    We close this gap.

    \begin{proposition}
    \label{th:membership_hardness}
    For every $k \in \NN_{\geq 2} \cup \Set{\infty}$, \cref{prob:cut_membership} is \NP-hard.
    Moreover, for finite $k$ it is \NP-complete.
    \end{proposition}
    \noindent
    First, a many-one reduction from the $\CUT_2^n$ to the $\CUT_k^n$ membership problem for each $k \in \NN_{\geq 3}$ asserts hardness for finite $k$.
    Then, we show that for finite $k$, \cref{prob:cut_membership} lies in \NP, proving \NP-completeness.
    Finally, a Turing reduction from the optimization problem over $\CUT_{\infty}^{n}$, which is known to be \NP-hard \citep{zhang2006}, to the $\CUT_{\infty}^n$ membership problem establishes hardness for $k = \infty$.
    The full proofs of \cref{th:membership_hardness} are provided in \cref{sec:hardness}.

    \begin{lemma}
        \label{th:ray_membership_equiv}
        Over $\CUT_k^n$, the ray maximization (\cref{prob:cut_ray}) and the membership problem (\cref{prob:cut_membership}) are polynomial-time equivalent. 
    \end{lemma}
    \begin{proof}
        Let $M \in \Herm_n^{\vec{0}}(\QQ(\uroot_k, \i))$, $\gamma \in \QQ_{\geq 0}$.
        Then $(M, \gamma)$ is a yes-instance of \cref{prob:cut_ray} if and only if $X \coloneqq \1 + \gamma M \in \Herm_n^{\vec{1}}(\QQ(\uroot_k, \i))$ is a yes-instance of \cref{prob:cut_membership}.

        Conversely, let $X \in \Herm_n^{\vec{1}}(\QQ(\uroot_k, \i))$.
        Take $M \coloneqq X - \1 \in \Herm_n^{\vec{0}}(\QQ(\uroot_k, \i))$ and $\gamma \coloneqq 1$.
        Then $X$ is a yes-instance of \cref{prob:cut_membership} if and only if $(M, \gamma)$ is a yes-instance of \cref{prob:cut_ray}.

        Both maps are polynomial-time computable many-one reductions.
    \end{proof}

    \begin{corollary}
        \label{th:ray_hardness}
        For every $k\in \NN_{\geq 2} \cup \Set{\infty}$, the $\CUT_k^n$ ray problem (\cref{prob:cut_ray}) is \NP-hard. Moreover for finite $k$, it is \NP-complete.
    \end{corollary}
    \begin{proof}
        By \cref{th:membership_hardness} the membership problem (\cref{prob:cut_membership}) is \NP-hard for every $k \in \NN_{\geq 2} \cup \Set{\infty}$, and by \cref{th:ray_membership_equiv} it reduces to the ray problem (\cref{prob:cut_ray}) in polynomial time. 
        Hence, the latter is \NP-hard as well.

        For finite $k$, \cref{th:membership_hardness} moreover places \cref{prob:cut_membership} in \NP.
        Since \cref{th:ray_membership_equiv} also reduces \cref{prob:cut_ray} to \cref{prob:cut_membership} by a many-one reduction, \cref{prob:cut_ray} lies in \NP\ and is therefore \NP-complete.
    \end{proof}

    \subsection{Semidefinite relaxation and rounding}
    \label{sec:relax}
    We defined the complex $k$-cut polytope as the convex hull of rank-$1$ positive-semidefinite matrices with entries in $\urootset_k$.
    By dropping the nonconvex rank-$1$ and $\urootset_k$ constraints, we relax this set to the elliptope~\citep{laurent1995}:
    \begin{equation}
        \elliptope_n \coloneqq \Set*{ X \in \CC^{n \times n} \given X \succeq 0, \ X_{ii} = 1 \ \forall i \in [n] }.
    \end{equation}
    Notice that membership in $\elliptope_n$ can be efficiently verified.
    
    The elliptope corresponds to the set of covariance matrices, where the individual variances are all one.
    Thus, to round a relaxed matrix $X \in \elliptope_n$ to an extremal point $\vec{x} \vec{x}^\dagger$ of $\CUT_k^n$, it is natural to generate randomized samples with the covariance matrix $X$.
    For $k = 2$, the matrix $X$ is real, and we sample a real Gaussian vector $\vec{\xi} \sim \mathcal{N}(0, X)$.
    For $k \geq 3$, the Hermitian matrix $X$ only dictates the covariance $\EE[\vec{\xi} \vec{\xi}^\dagger]$ of a complex random vector $\vec{\xi}$, but not its pseudo-covariance $\EE[\vec{\xi} \vec{\xi}^T]$.
    Requiring the latter to vanish yields the \emph{circularly symmetric complex Gaussian} vectors $\vec{\xi} \sim \mathcal{CN}(0, X)$, whose real and imaginary parts are jointly Gaussian with zero mean and whose distribution is invariant under global phase rotations $\vec{\xi} \mapsto \e^{\i \varphi} \vec{\xi}$.
    Such a vector can be generated as $\vec{\xi} = X^{1/2} (\vec{g} + \i \vec{h}) / \sqrt{2}$, where $\vec{g}$ and $\vec{h}$ are independent standard real Gaussian vectors.
    For finite $k$, the entries of $\vec{\xi}$ are projected onto $\urootset_k$ by applying a rounding function $\sigma_k$.
    One possible choice is
    \begin{equation}
        \sigma_k(\xi_i) \coloneqq \begin{cases}
            1, & \textnormal{if } \arg(\xi_i) \in [ 0, \frac{1}{k} 2 \pi ), \\
            \uroot_k, & \textnormal{if } \arg(\xi_i) \in [ \frac{1}{k} 2 \pi, \frac{2}{k} 2 \pi ), \\
            \vdots \\
            \uroot_k^{j}, & \textnormal{if } \arg(\xi_i) \in [ \frac{j}{k} 2 \pi, \frac{j+1}{k} 2 \pi ), \\
            \vdots \\
            \uroot_k^{k-1}, & \textnormal{if } \arg(\xi_i) \in [ \frac{k-1}{k} 2 \pi, 2 \pi ),
        \end{cases}
    \end{equation}
    which is depicted in \cref{fig:projection}.
    If $k = \infty$, the entries of $\vec{\xi}$ only have to be normalized, such that each has modulus one.
    We set
    \begin{equation}
        \sigma_{\infty}(\xi_i) \coloneqq \frac{\xi_i}{\abs{\xi_i}}.
    \end{equation}
    \begin{figure}
        \centering
        \includegraphics{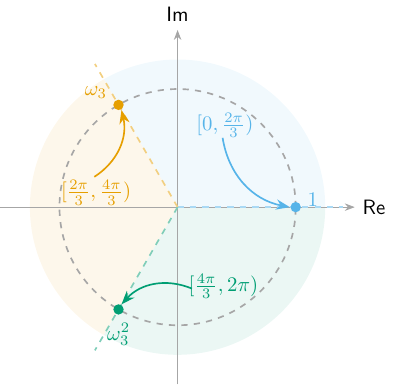}
        \caption[A possible projection onto the $k$-th roots of unity.]{A possible projection onto the $k$-th roots of unity.
                The complex plane is cut from the origin into $k$ equiangular sections.
                In counterclockwise direction, each section is assigned a $k$-th root of unity.
                Rotating the sections around the origin yields rounding functions with equivalent distributions.
        }
        \label{fig:projection}
    \end{figure}
    The resulting vector is given by $\vec{x} = \sigma_k^{\elementwise}(\vec{\xi})$.
    This randomized rounding method is typically applied to complex semidefinite programs.
    It was first introduced for the real case with $k=2$ by \citet{goemans1995} and \citet{bertsimas1998}, extended to complex values and $k=3$ by \citet{goemans2004}, and finally formalized in this general way by \citet{zhang2006}, see also \citet{so2007,huang2010}.
    Moreover, \citet{zhang2006} show that the expectation of the outer product matrix $\vec{x} \vec{x}^\dagger$ can be calculated elementwise from $X$ by
    \begin{equation}
        \EE[\vec{x} \vec{x}^\dagger] = \RDF_k^{\elementwise}(X),
    \end{equation}
    with the continuous \emph{rounding distortion function} 
    \begin{equation}
        \RDF_k : \DD \to \DD
    \end{equation}
    given by
    \begin{equation}
        \RDF_k(z) = C_k \sum_{j=0}^{k-1} \uroot_k^j \arccos^2 \left( - \Re \left( \uroot_k^{-j} z \right) \right),
    \end{equation}
    with the $k$-dependent constant
    \begin{equation}
        C_k \coloneqq \frac{k (2 - \uroot_k^{-1} - \uroot_k)}{8 \pi^2} \in \RR_{> 0},
    \end{equation}
    for finite $k$, and
    \begin{equation}
        \begin{split}
            &\RDF_{\infty}(z) \\
            &\quad = \frac{1}{4 \pi} \int_{0}^{2 \pi} \e^{\i \vartheta} \arccos^2 (- r \cos(\vartheta - \phi)) \, \dd \vartheta
        \end{split}
    \end{equation}
    with $z = r \e^{\i \phi}$ for $k=\infty$. 
    In the following lemma we collect important properties of the distortion function, that will be used extensively in the remainder of this section and \cref{sec:bounds}.
    The full proofs are provided in \cref{sec:expectation}.

    \begin{lemma}     
        \label{th:expectation}   
        For all $k \in \NN_{\geq 3} \cup \Set{\infty}$, $\RDF_k$ fulfills the following properties:
        \begin{enumerate}
            \item $\RDF_k(0) = 0$, and $\RDF_k(\e^{\i \theta}) = \e^{\i \theta}$ for all $\theta \in \Theta_k$. In particular, $\RDF_k(1) = 1$.
            \item $\RDF_k(\conj{z}) = \conj{\RDF_k(z)}$ for all $z \in \DD$.
            \item $\RDF_k$ is injective on $\DD$.
            \item $\abs{\RDF_k(z)} \geq L_k \abs{z}$ for all $z \in \DD$. 
            The $k$-dependent constant $L_k$ is given by
            \begin{equation}
                L_k \coloneqq C_k k \frac{\pi}{2} \begin{cases}
                    1, & \textnormal{if } k > 3, \\
                    1 - \frac{\pi}{4}, & \textnormal{if } k = 3,
                \end{cases}
            \end{equation}
            for finite $k$, and
            \begin{equation}
                L_{\infty} \coloneqq \frac{\pi}{4}.
            \end{equation}
            \item The image of $\RDF_k$ contains the closed disk of radius $L_k$:
            \begin{equation}
                \Set{w \in \CC \given \abs{w} \leq L_k} \subseteq \RDF_k(\DD).
            \end{equation}
            In particular, $\RDF_k^{\elementwise -1}(X)$ is well-defined for every matrix $X$ with $\lpnorm[\infty]{X} \leq L_k$.
        \end{enumerate}
        Moreover, $\RDF_2$ fulfills the following weaker properties on real values:
        \begin{enumerate}[resume]
            \item $\RDF_2(-1) = -1$, $\RDF_2(0) = 0$, and $\RDF_2(1) = 1$.
            \item $\RDF_2$ is injective on $[-1, 1]$.
            \item For all $x \in [-1, 1]$:
            \begin{equation}
                \abs{\RDF_2(x)} \geq \frac{2}{\pi} \abs{x} =: L_2 \abs{x}.
            \end{equation}
        \end{enumerate}
    \end{lemma}
        
    Now we can state a relaxation scheme for maximizing the ray over $\CUT_k^n$.
    Because membership in $\elliptope_n$ can be efficiently verified, we may optimize the relaxed ray
    \begin{equation}
        \1 + \gamma M \in \elliptope_n
    \end{equation}
    efficiently, for example via a binary search.
    Applying the randomized rounding would then yield pulses $\vec{x}$ with expected outer product matrix 
    \begin{equation}
        \label{eq:systematic_deviation}
        \EE[\vec{x} \vec{x}^\dagger] = \RDF_k^{\elementwise}(\1 + \gamma M) = \1 + \RDF_k^{\elementwise}(\gamma M).
    \end{equation}
    Notice that \cref{th:expectation} allowed us to pull the identity out of the distortion function, as $\gamma M$ is zero on the diagonal.
    This approach is depicted in \cref{fig:elliptope_ray}.

    \begin{figure}
        \centering
        \includegraphics[width=0.4\textwidth]{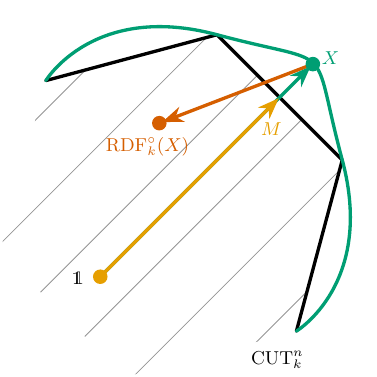}
        \caption{Intuitively, one might try to solve the elliptope relaxation of the solution ray directly.
                Then the rounding distortion would, however, introduce a systematic deviation away from the ray, leading to an ineffective selection of pulses.}
        \label{fig:elliptope_ray}
    \end{figure}

    Because the distortion function $\RDF_k$ is a nonlinear function, \cref{eq:systematic_deviation} will deviate systematically from the solution ray in general.
    To prevent this distortion, we will apply the inverse of $\RDF_k$, which we know to exist by \cref{th:expectation} on the image $\RDF_k^{\elementwise}(\elliptope_n)$, before rounding.
    This is also known as \emph{Krivine rounding}, as it was first applied by \citet{krivine1979} for improved bounds on the Grothendieck constant.
    Such Krivine schemes are the standard tool for rounding in Grothendieck-type inequalities~\citep{alon2006,briet2014,naor2014}.
    We thus optimize the ray \emph{bent} by the inverse distortion:
    \begin{equation}
        \label{eq:bended_ray}
        \RDF_k^{\elementwise -1}(\1 + \gamma M) = \1 + \RDF_k^{\elementwise -1}(\gamma M) \in \elliptope_n.
    \end{equation}
    Then, rounding to vectors $\vec{x}$ is guaranteed to yield outer product matrices with expected value on the ray
    \begin{equation}
        \EE[\vec{x} \vec{x}^\dagger] = \RDF_k^{\elementwise}\left( \1 + \RDF_k^{\elementwise -1} ( \gamma M ) \right) = \1 + \gamma M.
    \end{equation}
    We depict this method in \cref{fig:bent_ray}.

    \begin{figure}
        \centering
        \includegraphics[width=0.4\textwidth]{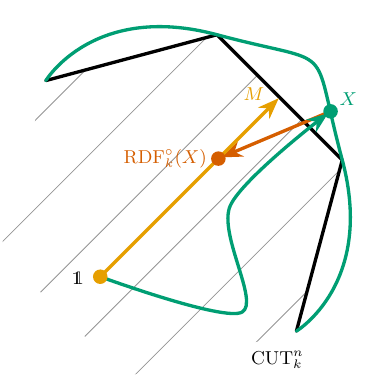}
        \caption{Instead of optimizing the relaxed solution ray directly, we can bend it by the inverse of the distortion function $\RDF_k^{\elementwise -1}$.
                Then, in the rounding stage the sampled pulses will average out back onto the solution ray.
                This fully prevents the systematic deviation observed before.}
        \label{fig:bent_ray}
    \end{figure}

    Note that we can still employ a binary search to optimize \cref{eq:bended_ray}, because we can check the existence of $\RDF_k^{\elementwise -1}(\gamma M)$, as well as verify the inclusion in $\elliptope_n$ efficiently.
    The numerical inversion of $\RDF_k$ is detailed in \cref{sec:numerics}.
    In \cref{sec:bounds}, we show that setting $\gamma$ to $\gamma_{\mathrm{lo}} \coloneqq \frac{L_k}{\fnorm{M}} \sqrt{\frac{n}{n-1}}$ satisfies both.
    Moreover, since $\vec{x} \vec{x}^\dagger$ is a random extreme point of the compact convex set $\CUT_k^n$ and its expectation therefore lies in $\CUT_k^n$,
    \begin{equation}
        \RDF_k^{\elementwise} : \elliptope_n \to \CUT_k^n,
    \end{equation}
    so a necessary condition for 
    \begin{equation}
        \RDF_k^{\elementwise -1} (\1 + \gamma M) \succeq 0
    \end{equation}
    is 
    \begin{equation}
        \1 + \gamma M \succeq 0,
    \end{equation}
    which is violated for $\gamma > - 1 / \mu_{\min}(M)$.
    Hence, the search can be constrained to the range 
    \begin{equation}
        \gamma \in \left[\frac{L_k}{\fnorm{M}} \sqrt{\frac{n}{n-1}}, - \frac{1}{\mu_{\min}(M)}\right].
    \end{equation}
    Each iteration step of the binary search then costs $\LandauO(n^3)$, yielding $\LandauO(n^3 \log(1 / \tau))$ overall where $\tau$ is the tolerance.

    We are, however, no longer guaranteed to find the maximum $\gamma$.
    The binary search might terminate, when the bent ray leaves the elliptope for the first time, say at $\hat{\gamma}$.
    However, the bent ray is no longer linear, so it could re-enter $\elliptope_n$ for a larger $\gamma > \hat{\gamma}$.
    Nonetheless, we employ the ray binary search as outlined in \cref{alg:ray_binary_search}.
    In \cref{sec:bounds}, we derive lower bounds on the solution $\hat{\gamma}$ found by this binary search.    
    One could consider alternative ray-tracing techniques, which we will not do in this work.

    \begin{algorithm}
        \caption{Ray binary search}\label{alg:ray_binary_search}
        \DontPrintSemicolon
        \SetKwInOut{Input}{Input}\SetKwInOut{Output}{Output}
        \SetKwFunction{OnRay}{OnRay}
        \SetKwProg{Fn}{Function}{:}{}
        \BlankLine
        \Input{Hollow Hermitian matrix $M$; tolerance $\tau > 0$.}
        \Output{$\hat{\gamma} > 0$ with $\1 + \RDF_k^{\elementwise -1}(\hat{\gamma} M) \in \elliptope_n$, and some $\gamma \in (\hat{\gamma}, \hat{\gamma} + \tau]$ violating this.}
        \BlankLine
        \Fn{\OnRay{$\gamma$}}{
        \lIf{$\gamma M_{ij} \notin \RDF_k(\DD)$ for some $i \neq j$}{\Return \textnormal{false}}
        \Return $\1 + \RDF_k^{\elementwise -1}(\gamma M) \succeq 0$\;
        }
        \BlankLine
        $\gamma_{\mathrm{lo}} \gets \frac{L_k}{\fnorm{M}}\sqrt{\frac{n}{n-1}}$\; 
        $\gamma_{\mathrm{hi}} \gets -1 / \mu_{\min}(M)$\;
        \lIf{\OnRay{$\gamma_{\mathrm{hi}}$}}{\Return $\gamma_{\mathrm{hi}}$}
        \While{$\gamma_{\mathrm{hi}} - \gamma_{\mathrm{lo}} > \tau$}{
            $\gamma \gets (\gamma_{\mathrm{lo}} + \gamma_{\mathrm{hi}}) / 2$\;
            \lIf{\OnRay{$\gamma$}}{$\gamma_{\mathrm{lo}} \gets \gamma$}
            \lElse{$\gamma_{\mathrm{hi}} \gets \gamma$}
        }
        \Return $\gamma_{\mathrm{lo}}$\;
    \end{algorithm}

    \subsection{Informed LP algorithm}
    \label{sec:informed}
    So far we have introduced a scheme to sample vectors $\vec{x}$, each corresponding to a conjugation of the system by some unitary, that yield expected outer product matrices on the solution ray. 
    We will now discuss how to derive solutions to the original Hamiltonian engineering problem, that is evolution times $\lambda(\vec{x})$, from the generated samples~$\vec x$. 

    An exact solution can be found similar as for the uninformed relaxation in~\citep{bassler2025,kum2026}.
    By restricting the set of pulses to the set of sampled vectors $S$, a \ac{LP} can be solved efficiently:
    \begin{equation*}
        \tag{RestrictedLP}
        \label{eq:restricted_lp}
        \begin{alignedat}{3}
            &\textnormal{min} \enspace &&\lpnorm[1]{\vec{\lambda}} \\
            &\textnormal{s.t.} \enspace &&\sum_{\vec{x} \in S} \lambda(\vec{x}) (\vec{x} \vec{x}^{\dagger})_{ij} = M_{ij}, && \enspace \forall (i,j) \in \nz(A), \\
            &&& \lambda(\vec{x}) \geq 0, && \enspace \forall \vec{x} \in S.
        \end{alignedat}
    \end{equation*}
    We call this approach the informed \ac{LP} algorithm and outline it in \cref{alg:informed_lp}.
    
    \begin{algorithm}
        \caption{Informed \ac{LP} algorithm}\label{alg:informed_lp}
        \DontPrintSemicolon
        \SetKwInOut{Input}{Input}\SetKwInOut{Output}{Output}
        \BlankLine
        \Input{Hollow Hermitian matrices $A$ and $B$; number $s > 0$ of pulses to sample.}
        \Output{Evolution times $\vec{\lambda}$ fulfilling \cref{eq:decomposition_constraints}.}
        \BlankLine
        $M \gets$ effective matrix from $A, B$\;
        $\hat{\gamma} \gets$ ray binary search (\cref{alg:ray_binary_search}) on $M$\;
        $X \gets \1 + \RDF_k^{\elementwise -1}(\hat{\gamma} M)$\;
        $S \gets s$ pulses sampled and rounded on $X$\;
        $\vec{\lambda} \gets$ solve \eqref{eq:restricted_lp} over $S$\;
        \Return $\vec{\lambda}$\;
    \end{algorithm}

    When applying the informed \ac{LP} algorithm (\cref{alg:informed_lp}), there is no a priori guarantee that the \ac{LP} is feasible.
    Numerically, we observe on generic instances, that the probability that the \ac{LP} contains a solution with quantum run time at most $\frac{1}{\hat{\gamma}}$ rises sharply from close to zero to close to one around $s \approx 4m$ ($s \approx 2m$ on $k = 2$ with real $M$).
    \Cref{sec:feasibility_numerics} examines this transition systematically.
    In the following, we will qualify these observations analytically.

    Let the ray binary search (\cref{alg:ray_binary_search}) have returned with
    \begin{equation}
        X = \1 + \RDF_k^{\elementwise -1} (\hat{\gamma} M) \in \partial \elliptope_n
    \end{equation}
    on the boundary of the elliptope, and let the informed LP algorithm (\cref{alg:informed_lp}) have sampled and rounded the set of samples $S$ on covariance matrix $X$.
    We define the modified \ac{LP}
    \begin{equation}
        \label{eq:conjecture_lp}
        \begin{alignedat}{3}
            &\textnormal{find} \enspace &&\tilde{\vec{\lambda}} : S \to \RR_{\geq 0}, \\
            &\textnormal{s.t.} \enspace &&\sum_{\vec{x} \in S} \tilde{\lambda}(\vec{x}) (\vec{x} \vec{x}^{\dagger})_{ij} = \hat{\gamma} M_{ij}, \enspace && \forall (i, j) \in \nz(A), \\
            &&& \lpnorm[1]{\tilde{\vec{\lambda}}} = 1,
        \end{alignedat}
    \end{equation}
    which is feasible if and only if the \ac{LP} \eqref{eq:restricted_lp} has a feasible solution with quantum run time $\frac{1}{\hat{\gamma}}$. 
    Hence, we may argue via the feasibility of \cref{eq:conjecture_lp}.
    Thanks to the hermiticity of the involved matrices, we may reduce the problem statement to one in $\RR^D$, where
    \begin{equation}
        D \coloneqq \begin{cases}
            m, & \textnormal{if } k=2, M \textnormal{ real}, \\
            2m, & \textnormal{if } k \in \NN_{\geq 3} \cup \Set{\infty},
        \end{cases}
    \end{equation}
    by defining the real-valued target vector $\vec{m} = \vec{m}^{\Re} \oplus \vec{m}^{\Im} \in \RR^D$ corresponding to $M$, and the correlation vectors $\vec{\chi}(\vec{x}) = \vec{\chi}^{\Re}(\vec{x}) \oplus \vec{\chi}^{\Im}(\vec{x}) \in \RR^D$ for each $\vec{x} \in S$ via
    \begin{equation}
        \label{eq:correlation_var}
        \begin{aligned}
            m^{\Re}_{\Set{i, j}} &\coloneqq \Re(M_{ij}), \\ 
            m^{\Im}_{\Set{i, j}} &\coloneqq \Im(M_{ij}), \\
            \chi^{\Re}_{\Set{i, j}}(\vec{x}) &\coloneqq \Re\left( (\vec{x} \vec{x}^\dagger)_{ij} \right), \\
            \chi^{\Im}_{\Set{i, j}}(\vec{x}) &\coloneqq \Im\left( (\vec{x} \vec{x}^\dagger)_{ij} \right),
        \end{aligned}
    \end{equation}
    for all $(i, j) \in \nz(A)$, $i < j$, where the direct sum is only carried out for $k > 2$.
    We denote by $\vec{\chi}(S)$ the equivalent of $S$, where each element is replaced by its correlation vector
    \begin{equation}
        \vec{\chi}(S) \coloneqq \Set*{\vec{\chi}(\vec{x}) \in \RR^{D} \given \vec{x} \in S}.
    \end{equation}
    Then, the \ac{LP} \eqref{eq:conjecture_lp} is feasible if and only if
    \begin{equation}
        \hat{\gamma} \vec{m} \in \conv \left(\vec{\chi}(S)\right).
    \end{equation}
    For the random correlation vector $\vec{\chi}(\vec{x})$, we define its covariance matrix ${\Sigma \coloneqq \Cov(\vec{\chi}(\vec{x}))}$ with range ${V \coloneqq \ran \Sigma}$ and rank ${D_{\mathrm{eff}} \coloneqq \rank \Sigma}$.
    Notice that for an arbitrary direction $\vec{v} \in \RR^D$, its overlap with the effect of a random pulse $\vec{x}$ is given by $\innerp{\vec{v}}{\vec{\chi}(\vec{x})}$, with variance
    \begin{equation}
        \Var \innerp{\vec{v}}{\vec{\chi}(\vec{x})} = \vec{v}^T \Sigma \vec{v}.
    \end{equation}
    Thus, for $\vec{v} \in \ker \Sigma$, this variance vanishes $\Var \innerp{\vec{v}}{\vec{\chi}(\vec{x})} = 0$, so in particular
    \begin{equation}
        \label{eq:constant_overlap}
        \innerp{\vec{v}}{\vec{\chi}(\vec{x})} = \innerp{\vec{v}}{\EE[\vec{\chi}(\vec{x})]} = \innerp{\vec{v}}{\hat{\gamma} \vec{m}}
    \end{equation}
    almost surely.
    Hence, the effect of sampleable pulses in these directions is fixed, and the relevant problem space is $(\ker \Sigma)^{\bot} = \ran \Sigma = V$ with effective dimension $D_{\mathrm{eff}}$.

    Minimizing the probability of a pulse overshooting the target over the attainable directions then yields the \emph{positivity index}
    \begin{equation}
        \beta \coloneqq \inf_{\substack{\tilde{\vec{v}} \in V \\ \lpnorm[2]{\tilde{\vec{v}}} = 1}} \PP\left[\innerp{\tilde{\vec{v}}}{\vec{\chi}(\vec{x})} > \innerp{\tilde{\vec{v}}}{\hat{\gamma} \vec{m}}\right].
    \end{equation}
    We define $\beta \coloneqq 1$ in the trivial case of $D_{\mathrm{eff}} = 0$.

    We call the \ac{LP} \eqref{eq:conjecture_lp} \emph{robustly} feasible, if
    \begin{equation}
        \hat{\gamma} \vec{m} \in \interior_{V} \left( \conv \left( \vec{\chi}(S) \right) \right),
    \end{equation}
    where $\interior_V$ denotes the interior in the topology of the affine subspace $\hat{\gamma} \vec{m} + V$, that is the \ac{LP} \eqref{eq:conjecture_lp} stays feasible under any sufficiently small perturbation in an \emph{attainable} direction.
    By \cref{eq:constant_overlap}, perturbations outside $V$ cannot be matched.
    Intuitively, a small $\beta$ means that many samples are necessary for the \ac{LP} \eqref{eq:conjecture_lp} to become robustly feasible.

    The following theorem captures the transition window of robust feasibility.
    We provide its proof in \cref{sec:feasibility_proofs}.

    \begin{theorem}
        \label{th:feasibility_window}
        Denote the event that the \ac{LP} \eqref{eq:conjecture_lp} is robustly feasible by $A$. 
        \begin{enumerate}[label=\textup{(\roman*)}]
            \item For every $\eta \in (0, 1)$, if $\abs{S} \leq (2 - \eta) D_{\mathrm{eff}}$, then the probability that the \ac{LP} \eqref{eq:conjecture_lp} is robustly feasible is upper bounded by
            \begin{equation}
                \PP[A] \leq \exp(- \eta^2 D_{\mathrm{eff}} / 4).
            \end{equation}
            \item The positivity index satisfies $\beta > 0$ for every instance.
            Moreover, for every $\delta \in (0, 1)$, if
            \begin{equation}
                \abs{S} \geq 6 \beta^{-1} (D_{\mathrm{eff}} \ln(12 / \beta) + \ln(2 / \delta)),
            \end{equation}
            then the \ac{LP} \eqref{eq:conjecture_lp} is robustly feasible with probability at least
            \begin{equation}
                \PP[A] \geq 1 - \delta.
            \end{equation}
        \end{enumerate}
    \end{theorem}
        
    Part (i) provides a necessary number of pulses and closely matches our observations.
    Part (ii) is a sufficiency statement and gives pulse numbers much larger than the observed $s = 2D$.
    In fact, for instance families where $\beta$ is not bounded from below by a positive constant, the sufficient $s$ may not be linear in $m$.
    In \cref{sec:feasibility_numerics} we demonstrate this on a family of instances along which $\beta$ becomes small, and for which the required number of pulses grows accordingly.
    For arbitrary instances, $\beta$ cannot be computed directly and only estimated from generated samples.
    Hence, part (ii) establishes an asymptotic behavior but does not provide exact sample numbers to be used.
    
    Both parts of \cref{th:feasibility_window} concern robust feasibility.
    For part (ii) this is a strengthening, since robust feasibility implies feasibility, and the sufficiency statement therefore applies directly to the plain case.
    The necessity statement (i) however does not apply to the plain feasibility, as fewer pulses may be necessary.

    The related statements of \cref{th:feasibility_window} for the feasibility of the uninformed \ac{LP} have only been observed numerically in \citep{bassler2025,kum2026} and observed to be consistent with the statement of Wendel's theorem \cite{wendel1962} even though its hypothesis isn't satisfied for the relevant distributions.
    Our progress beyond such numerical observations stems from restricting to robust feasibility, which admits a smoothing argument for the necessity bound, and, for sufficiency, from using $\varepsilon$-nets to establish the asymptotic scaling of the required pulse number instead of its exact tight threshold.

    \subsection{Mixed LP algorithm}
    \label{sec:mixed}
    For ill-conditioned instances with small positivity index $\beta$, we now introduce a modified version of the informed \ac{LP} algorithm (\cref{alg:informed_lp}).
    The core intuition is that we may mix $s_{\mathrm{u}} \in \LandauO(m)$ uniformly sampled pulses, to guarantee feasibility, with $s_{\mathrm{i}} \in \LandauO(m)$ informedly sampled pulses, to provide a low quantum run time.
    The resulting algorithm yields feasibility and quantum run time guarantees on $s \in \LandauO(m)$ sampled pulses without requiring any assumptions on the instance.
    In the following, we will denote the set of informed samples by $S_{\mathrm{i}}$, and the set of uniform samples by $S_{\mathrm{u}}$.

    Key to understanding the mixed approach are two observations.
    Firstly, note that for $\vec{x}$ sampled informedly, $\EE[\vec{x} \vec{x}^\dagger] = \1 + \hat{\gamma} M$, so we may approximate a $\frac{1}{\hat{\gamma}}$ quantum run time solution by taking the observed relative frequencies of the informed pulses.
    That is, we may define the informed evolution times $\vec{\lambda}_{\mathrm{i}} : S_{\mathrm{i}} \to \RR_{\geq 0}$ by
    \begin{equation}
        \label{eq:informed_times}
        \lambda_{\mathrm{i}}(\vec{x}) \coloneqq \frac{f_{\vec{x}}}{\hat{\gamma}}
    \end{equation}
    for each sampled informed pulse, where $f_{\vec{x}}$ denotes the relative sample frequency of $\vec{x}$.
    Then, $\lpnorm[1]{\vec{\lambda}_{\mathrm{i}}} = 1 / \hat{\gamma}$,
    \begin{equation}
        \sum_{\vec{x} \in S_{\mathrm{i}}} \lambda_{\mathrm{i}}(\vec{x}) \vec{x} \vec{x}^{\dagger} = \frac{1}{\hat{\gamma}} \mean{\left(\vec{x} \vec{x}^\dagger\right)} \to \frac{1}{\hat{\gamma}} \1 + M,
    \end{equation}
    and hence
    \begin{equation}
        \sum_{\vec{x} \in S_{\mathrm{i}}} \lambda_{\mathrm{i}}(\vec{x}) \vec{\chi}(\vec{x}) = \frac{1}{\hat{\gamma}} \mean{\vec{\chi}(\vec{x})} \to \vec{m},
    \end{equation}
    as $s_{\mathrm{i}} \to \infty$.
    For finite $s_{\mathrm{i}}$ there will however always be some deviation, i.e.\ the sample mean will not correspond to a feasible solution of the \ac{LP}~\eqref{eq:restricted_lp}, and the residual
    \begin{equation}
        \vec{r} \coloneqq \vec{m} - \frac{1}{\hat{\gamma}} \mean{\vec{\chi}(\vec{x})}
    \end{equation}
    will be nonzero.
    Using concentration inequalities, the magnitude of $\vec{r}$ can be bounded.
    \begin{lemma}
        \label{th:residual}
        For $s_{\mathrm{i}}$ informed samples, and $t \geq \frac{\sqrt{m}}{\hat{\gamma} \sqrt{s_{\mathrm{i}}}}$, the magnitude of the residual is bounded by
        \begin{equation}
            \PP[\lpnorm[2]{\vec{r}} \geq t] \leq \exp \left( - \frac{s_{\mathrm{i}} (\hat{\gamma} t - \sqrt{m / s_{\mathrm{i}}})^2}{2 m} \right).
        \end{equation}
    \end{lemma}

    The second key observation is captured by the following lemma, and states that $s_{\mathrm{u}} \in \LandauO(m)$ \emph{uniform} samples suffice to guarantee the enclosure of a closed $\ell_2$-ball $\bar{B}(\vec{0}, \rho)$ in the convex hull of the samples, where $\rho$ is an instance-independent constant radius.
    Remember that $m$ is the number of interaction terms and $D$ the real dimension of the constraint space given by $M$. 
    Their ratio $c \coloneqq m / D$ is a constant $c\in \Set{1,1/2}$ depending on whether $k=2$ or $k>2$.

    \begin{lemma}
        \label{th:uniform_ball}
        Take $\kappa \coloneqq 15$, $c = m / D$, and 
        \begin{equation}
            \rho_0 \coloneqq \frac{\sqrt{c}}{4 \sqrt{\kappa}}, \quad C \coloneqq 200 \kappa.
        \end{equation}
        Then for every $\delta \in (0, 1)$, generating 
        \begin{equation}
            s_{\mathrm{u}} \geq C (D + \ln(1 / \delta))
        \end{equation}
        uniform samples gives a probability of the convex hull containing the closed $\rho$-ball lower bounded by
        \begin{equation}
            \PP[\bar{B}(\vec{0}, \rho_0) \subseteq \conv \vec{\chi}(S_{\mathrm{u}})] \geq 1 - \delta.
        \end{equation}        
    \end{lemma}
    \noindent
    We provide the proofs of \cref{th:residual,th:uniform_ball} in \cref{sec:mixed_proofs}.

    The mixed LP algorithm then utilizes \cref{th:uniform_ball} by taking $s_{\mathrm{i}}$ large enough, s.t.\ $\lpnorm[2]{\vec{r}}$ is small and can be corrected cheaply from $\bar{B}(\vec{0}, \rho_0)$.
    We outline the approach in \cref{alg:mixed}.
    \Cref{th:mixed_algo} establishes that, independent of the specific instance, $\LandauO(m)$ samples suffice for a feasible solution with quantum run time \emph{close} to $\frac{1}{\hat{\gamma}}$ with high probability.

    \begin{algorithm}
        \caption{Mixed LP algorithm}\label{alg:mixed}
        \DontPrintSemicolon
        \SetKwInOut{Input}{Input}\SetKwInOut{Output}{Output}
        \BlankLine
        \Input{Hollow Hermitian matrices $A$ and $B$; numbers $s_{\mathrm{i}}, s_{\mathrm{u}} > 0$ of informed and uninformed pulses to sample.}
        \Output{Evolution times $\vec{\lambda}$ fulfilling \cref{eq:decomposition_constraints}.}
        \BlankLine
        $M \gets$ effective matrix from $A, B$\;
        $\hat{\gamma} \gets$ ray binary search (\cref{alg:ray_binary_search}) on $M$\;
        $X \gets \1 + \RDF_k^{\elementwise -1}(\hat{\gamma} M)$\;
        $S_{\mathrm{i}} \gets s_{\mathrm{i}}$ pulses sampled and rounded on $X$\;
        $S_{\mathrm{u}} \gets s_{\mathrm{u}}$ pulses sampled uniformly from $\urootset_k^n$\;
        $\vec{\lambda} \gets$ solve \eqref{eq:restricted_lp} over $S_{\mathrm{i}} \cup S_{\mathrm{u}}$\;
        \Return $\vec{\lambda}$;
    \end{algorithm}

    \begin{theorem}
        \label{th:mixed_algo}
        For any $\delta, \delta', \varepsilon \in (0, 1)$, let 
        \begin{equation}
            s_{\mathrm{u}} \geq C (D + \ln(1 / \delta)),
        \end{equation}
        and 
        \begin{equation}
            s_{\mathrm{i}} \geq \frac{m}{\rho_0^2 \varepsilon^2} \left(1 + \sqrt{2 \ln(1 / \delta')}\right)^2.
        \end{equation}
        Then, with probability at least $1 - \delta - \delta'$, the mixed \ac{LP} algorithm (\cref{alg:mixed}) returns a feasible solution with quantum run time $\lpnorm[1]{\vec{\lambda}} \leq (1 + \varepsilon) / \hat{\gamma}$.
    \end{theorem}
    \begin{proof}
        Take $\vec{\lambda}_{\mathrm{i}}$ as in \cref{eq:informed_times}.
        By \cref{th:residual}, since 
        \begin{equation}
            \frac{\rho_0 \varepsilon}{\hat{\gamma}} \geq \frac{\sqrt{m}}{\hat{\gamma}} \cdot \frac{\rho_0 \varepsilon}{\sqrt{m} (1 + \sqrt{2 \ln(1 / \delta')})} \geq \frac{\sqrt{m}}{\hat{\gamma} \sqrt{s_{\mathrm{i}}}},
        \end{equation}
        the probability of a large residual is upper bounded by
        \begin{equation}
            \begin{split}
                &\PP[\lpnorm[2]{\vec{r}} \geq \rho_0 \varepsilon / \hat{\gamma}] \\
                &\quad\leq \exp\left( - \frac{s_{\mathrm{i}} (\rho_0 \varepsilon - \sqrt{m / s_{\mathrm{i}}})^2}{2m} \right) 
                \leq \delta'.
            \end{split}
        \end{equation}
        By \cref{th:uniform_ball}, with probability at least $1 - \delta$, $(\rho_0 / \lpnorm[2]{\vec{r}})\vec{r}$ can be realized from the uniform pulses for unit cost, i.e.\ there is $\vec{\lambda}_{\mathrm{u}}$, s.t.\
        \begin{equation}
            \sum_{\vec{x} \in S_{\mathrm{u}}} \lambda_{\mathrm{u}}(\vec{x}) \vec{\chi}(\vec{x}) = \vec{r},
        \end{equation}
        and $\lpnorm[1]{\vec{\lambda}_{\mathrm{u}}} = \frac{\lpnorm[2]{\vec{r}}}{\rho_0}$.
        Hence with probability at least $1 - \delta - \delta'$, taking $\vec{\lambda} \coloneqq \vec{\lambda}_{\mathrm{i}} + \vec{\lambda}_{\mathrm{u}}$ is a feasible solution the \ac{LP}~\eqref{eq:restricted_lp} and fulfills
        \begin{equation}
            \lpnorma[1]{\vec{\lambda}_{\mathrm{i}} + \vec{\lambda}_{\mathrm{u}}} = \frac{1}{\hat{\gamma}} + \frac{\lpnorm[2]{\vec{r}}}{\rho_0} < \frac{1 + \varepsilon}{\hat{\gamma}}.
        \end{equation}
    \end{proof}

    The mixed \ac{LP} algorithm also opens a route to suppressing the finite pulse time errors discussed in \cref{sec:engineering}.
    We sketch the approach in the following, leaving a detailed analysis for future work.
    \Citet{bassler2025,kum2026}, generalizing \citet{votto2024}, observe that in the first-order Magnus expansion the error term inherits the locality of $H_S$.
    Whenever it moreover lies in the span of the $O_{ij}$, it can be absorbed into \cref{eq:decomposition_constraints} by shifting the target to $\vec{m} - \vec{\Delta}$, with an offset $\vec{\Delta} \in \RR^D$ that is fixed once the pulse set has been drawn.
    Solving for the shifted target then cancels the error exactly at the level of the first-order average Hamiltonian and leaves only the second-order Magnus truncation error.
    The modification keeps the program linear, but every drawn pulse contributes to $\vec{\Delta}$ and must consequently be applied, even those to which the solution assigns zero evolution time.

    Whether the error term lies in that span depends on $k$.
    A pulse of finite duration passes through the partial rotations $U_{u \vec{\theta}}$ with $u \in [0, 1]$, so that the error term integrates $U_{u\vec{\theta}}^\dagger H_S U_{u\vec{\theta}}$ over $u$.
    If the commutation relations \eqref{eq:phase_commutation} hold for all phases, that is if $k = \infty$ is admissible, as for the fermionic conjugations of \cref{sec:motivating_examples}, then these partial rotations are themselves admissible and merely attach a $u$-dependent phase to each $O_{ij}$.
    The error then remains in the span throughout and the shift absorbs it entirely.
    If the relation instead holds only on a proper subset $\Theta_k$, as for the $Z$-type Ising model with $X$-type conjugations or for the qudit shift and clock operators, the partial rotations take $H_S$ out of the span and leave a remainder that no choice of $\vec{\lambda}$ can reach.
    This remainder can, however, be addressed by an additional degree of freedom, not visible to the \ac{LP}.
    $U_{\vec{\theta}}$ depend on $\vec{\theta}$ only modulo $2\pi$ since the pulse generators have integer spectrum.
    The trajectory during a finite pulse, however, does depend on the winding.
    Winding a component as $\theta_i \to \theta_i + 2\pi w_i$ with $w_i \in \ZZ$ therefore leaves both the commutation relations \eqref{eq:phase_commutation} and the \ac{LP} untouched, while changing the remainder.
    For $k = 2$ the two windings $\pm \pi$ contribute with equal magnitude and opposite sign.
    \Citet{bassler2025} exploit this freedom in the rotation directions to cancel the remainder over the sequence.
    For larger $k$ the admissible windings generally contribute with differing magnitudes, so that cancellation instead requires combining several of them with suitable multiplicities, at a corresponding cost in pulses.

    Since only the right-hand side is modified, the technique is compatible with the \ac{LP}~\eqref{eq:restricted_lp} and can be combined with the informed sampling developed here.
    The offset $\vec{\Delta}$ is of first order in $t_{\mathrm{p}} \snorm{H_S}$, and hence small compared to $\vec{m}$ whenever the first-order treatment is justified.
    Its direction, however, is dictated by $H_S$ and by the drawn pulses, so that the informed pulses, which are sampled to cover the direction of $\vec{m}$, account for the bulk of the shifted target but not for the small off-direction component.
    Supplying the latter is instead done by the uniform pulses. 
    By \cref{th:uniform_ball} they realize a correction in an arbitrary direction at a cost of its $\ell_2$-norm over $\rho_0$, so that absorbing $\vec{\Delta}$ into the residual $\vec{r}$ in the proof of \cref{th:mixed_algo} bounds the quantum run time by
    \begin{equation}
        \lpnorm[1]{\vec{\lambda}} \leq \frac{1 + \varepsilon}{\hat{\gamma}} + \frac{\lpnorm[2]{\vec{\Delta}}}{\rho_0}
    \end{equation}
    with probability at least $1 - \delta - \delta'$.
    Robustness against the finite pulse time error therefore increases the run time by an additive term of first order in the pulse duration, retaining the leading $1 / \hat{\gamma}$ term of the informed sampling.
    As $\vec{\Delta}$ in turn depends on the drawn pulses, we leave a detailed analysis and a systematic treatment of the winding freedom for $k > 2$ to future work.

    As a final remark, we want to emphasize that the uniformly sampled pulses serve two distinct purposes.
    For safeguarding against infeasibility they are rarely necessary:
    On most instances the informed LP algorithm (\cref{alg:informed_lp}) with $s = (2 + \eta)D$, and without any uniform samples, is efficient, feasible, and provides low quantum run time solutions (see \cref{sec:feasibility_numerics}), so that the mixed LP algorithm (\cref{alg:mixed}) is only required if feasibility issues are actually encountered.
    As a means of realizing corrections in arbitrary directions, however, they are what renders the mitigation of the finite pulse time error discussed above applicable, and are therefore needed whenever the pulse duration cannot be neglected.
    
    The constants as stated in \cref{th:feasibility_window,th:uniform_ball,th:mixed_algo} are proof artifacts and likely far from optimal.
    Hence, the statements should be read as proofs of the asymptotic scalings and not as optimal sample numbers.
    Moreover, unless this mitigation is applied, the number of sampled pulses does not correspond to the number of pulses that have to be implemented, since only those in the support of $\vec{\lambda}$ enter the sequence.
    In that case, $s$ merely serves as an indicator of the classical overhead.

    \subsection{Quantum run time bounds}
    \label{sec:bounds}
    As established in \cref{sec:informed,sec:mixed}, the quantum run times of the informed LP algorithm (\cref{alg:informed_lp}) and the mixed LP algorithm (\cref{alg:mixed}) depend on the solution $\hat{\gamma}$ of the ray binary search (\cref{alg:ray_binary_search}).
    The following lemma shows that the previously claimed $\gamma_{\mathrm{lo}}$ is in fact valid, establishing a lower bound on $\hat{\gamma}$ in terms of the Frobenius norm, and thus allows us to derive solution guarantees.

    \begin{lemma}
        \label{th:ray_bound}
        \begin{equation}
            \gamma_{\mathrm{lo}} = \frac{L_k}{\fnorm{M}} \sqrt{\frac{n}{n-1}}
        \end{equation}
        is a valid lower seeding for the ray binary search (\cref{alg:ray_binary_search}).
        In particular, the solution $\hat{\gamma}$ of the search (\cref{alg:ray_binary_search}) is lower bounded by $\hat{\gamma} \geq \gamma_{\mathrm{lo}}$.
    \end{lemma}
    \begin{proof}
        We prove the claim by showing that $\1 + \RDF_k^{\elementwise -1}(\gamma_{\mathrm{lo}} M)$ is positive semidefinite and hence in the elliptope.
        First, note that for $n \geq 2$, $\RDF_k^{\elementwise -1}(\gamma_{\mathrm{lo}} M)$ is well-defined, since
        \begin{equation}
            \begin{split}
                \gamma_{\mathrm{lo}} \lpnorm[\infty]{M} &\leq L_k \sqrt{\frac{n}{n-1}} \frac{\lpnorm[\infty]{M}}{\fnorm{M}} \\
                &\leq L_k \sqrt{\frac{n}{2 (n - 1)}} \\
                &\leq L_k,
            \end{split}
        \end{equation}
        so by \cref{th:expectation}, $\RDF_k^{-1}$ exists on each entry of $\gamma_{\mathrm{lo}} M$.
        Moreover, recall from \cref{th:expectation} that
        \begin{equation}
            \abs{\RDF_k(z)} \geq L_k \abs{z},
        \end{equation}
        so
        \begin{equation}
            \abs{\RDF_k^{-1}(z)} \leq \frac{1}{L_k} \abs{z}.
        \end{equation}
        Then, we can bound the Frobenius norm of $\RDF_k^{\elementwise -1}(\gamma_{\mathrm{lo}} M)$ as
        \begin{equation}
            \begin{aligned}
                \fnorm*{\RDF_k^{\elementwise -1}(\gamma_{\mathrm{lo}} M)}
                &= \sqrt{ \sum_{i,j=1}^{n} \abs*{\RDF_k^{-1}(\gamma_{\mathrm{lo}} M_{ij})}^2 } \\
                &\leq \frac{\gamma_{\mathrm{lo}}}{L_k} \sqrt{ \sum_{i,j=1}^{n}  \abs{M_{ij}}^2 } \\
                &\leq \sqrt{\frac{n}{n-1}}.
            \end{aligned}
        \end{equation}

        Denote by $\mu_{\min} < 0$ the smallest eigenvalue, and by $\mu_{i}$ for $i \in [n-1]$ the other eigenvalues of $\RDF_k^{\elementwise -1}(\gamma_{\mathrm{lo}} M)$.
        The trace of $\RDF_k^{\elementwise -1}(\gamma_{\mathrm{lo}} M)$ vanishes, so
        \begin{equation}
            -\mu_{\min} = \sum_{i=1}^{n-1} \mu_i.
        \end{equation}
        Because $\RDF_k^{\elementwise -1}(\gamma_{\mathrm{lo}} M)$ is Hermitian, we can also express its Frobenius norm via its eigenvalues:
        \begin{equation}
            \begin{aligned}
                \fnorm*{\RDF_k^{\elementwise -1}(\gamma_{\mathrm{lo}} M)}
                &= \sqrt{\mu_{\min}^2 + \sum_{i=1}^{n-1} \mu_i^2} \\
                &\geq \sqrt{\mu_{\min}^2 + \frac{1}{n-1} \left( \sum_{i=1}^{n-1} \mu_i \right)^{2}} \\
                &= \sqrt{\frac{n}{n-1}} \abs{\mu_{\min}}.
            \end{aligned}
        \end{equation}
        Thus,
        \begin{equation}
            \abs{\mu_{\min}} \leq 1,
        \end{equation}
        and
        \begin{equation}
            \1 + \RDF_k^{\elementwise -1}(\gamma_{\mathrm{lo}} M) \succeq 0.
        \end{equation}
    \end{proof}

    \begin{corollary}
        \label{th:opt_upper}
        For every instance, $\1 + \hat{\gamma} M \in \CUT_k^n$.
        Consequently, there is a feasible solution to the \ac{LP}~\eqref{eq:general_lp} with $\lpnorm[1]{\vec{\lambda}} = 1 / \hat{\gamma}$ supported on at most $D + 1$ pulses, and the optimum $\vec{\lambda^*}$ satisfies
        \begin{equation}
            \lpnorm[1]{\vec{\lambda^*}} \leq \frac{1}{\hat{\gamma}} \leq \frac{\fnorm{M}}{L_k} \sqrt{\frac{n - 1}{n}}.
        \end{equation}
    \end{corollary}
    \begin{proof}
        The rounding samples from the extreme points of the compact convex set $\CUT_k^n$, so the expectation $\1 + \hat{\gamma} M$ lies in $\CUT_k^n$.
        By Carathéodory's theorem in the $D$-dimensional real space of the constrained entries, it is a convex combination of at most $D + 1$ extreme points.
        The second inequality is \cref{th:ray_bound}.
    \end{proof}

    From \cref{th:ray_bound}, we can then straightforwardly derive an upper bound on the quantum run times.
    \begin{theorem}
        \label{th:qtime_bound}
        If $\vec{\lambda}$ was determined via the informed \ac{LP} algorithm (\cref{alg:informed_lp}), with $s$ and $\delta$ fulfilling \cref{th:feasibility_window}, then with probability $1 - \delta$ the quantum run time can be upper bounded by
        \begin{equation}
            \lpnorm[1]{\vec{\lambda}} 
            \leq \frac{\fnorm{M}}{L_k} \sqrt{\frac{n - 1}{n}}.
        \end{equation}
        Similarly, if $\vec{\lambda}$ resulted from the mixed \ac{LP} algorithm (\cref{alg:mixed}), with $s_{\mathrm{i}}, s_{\mathrm{u}}, \varepsilon, \delta,$ and $\delta'$ fulfilling \cref{th:mixed_algo}, then with probability $1 - \delta - \delta'$ the quantum run time can be upper bounded by
        \begin{equation}
            \lpnorm[1]{\vec{\lambda}}
            \leq \frac{(1 + \varepsilon) \fnorm{M}}{L_k} \sqrt{\frac{n - 1}{n}}.
        \end{equation}
    \end{theorem}
    \begin{proof}
        By \cref{th:feasibility_window,th:mixed_algo}, $\lpnorm[1]{\vec{\lambda}} \leq 1 / \hat{\gamma}$ or $\lpnorm[1]{\vec{\lambda}} \leq (1 + \varepsilon) / \hat{\gamma}$, accordingly.
        In either case, \cref{th:ray_bound} implies the claim directly.
    \end{proof}

    With \cref{th:qtime_bound} we have shown that the quantum run times of our algorithms upper bounded linearly in the Frobenius norm of $M$.
    A lower bound on the optimal quantum run time for the application to qubit systems was derived by \citet{bassler2024}.
    Here, we extend this bound to our more general problem setup.

    \begin{theorem}
        \label{th:opt_bound}
        Let $\vec{\lambda}$ be any feasible solution to the \ac{LP}~\eqref{eq:general_lp}.
        Then, the quantum run time is lower bounded by
        \begin{equation}
            \lpnorm[1]{\vec{\lambda}} \geq \lpnorm[\infty]{M}.
        \end{equation}
    \end{theorem}
    \begin{proof}
        By \cref{eq:decomposition_constraints}, $M_{ij} = \sum_{\vec{\theta} \in \nz(\vec{\lambda})} \lambda(\vec{\theta}) e^{\i (\theta_i - \theta_j)}$ on all $(i, j) \in \nz(A)$.
        Applying the triangle inequality then yields $\abs{M_{ij}} \leq \lpnorm[1]{\vec{\lambda}}$ on all entries, which is the claim.
    \end{proof}

    With the standard norm inequality
    \begin{equation}
        \label{eq:fnorm_to_infnorm}
        \fnorm{M} \leq \lpnorm[\infty]{M} \sqrt{\abs{\nz(M)}},
    \end{equation}
    we can derive a nonconstant approximation ratio from \cref{th:qtime_bound,th:opt_bound}.

    \begin{corollary}
        \label{th:approx_ratio}
        The informed \ac{LP} algorithm (\cref{alg:informed_lp}), with $s$ and $\delta$ fulfilling \cref{th:feasibility_window} and the mixed LP algorithm (\cref{alg:mixed}) with $s_{\mathrm{i}}, s_{\mathrm{u}}, \varepsilon, \delta$, and $\delta'$ fulfilling \cref{th:mixed_algo} are randomized $\LandauO(\sqrt{m})$-approximation algorithms with success probabilities $1 - \delta$ and $1 - \delta - \delta'$, respectively. 
    \end{corollary}
    \begin{proof}
        We denote by $\vec{\lambda}$ the corresponding algorithm solution and by $\vec{\lambda^*}$ the optimal solution to the \ac{LP}~\eqref{eq:general_lp}.
        Then, by applying \cref{th:qtime_bound,th:opt_bound}, and \cref{eq:fnorm_to_infnorm}, we can bound the approximation ratio
        \begin{equation}
            \begin{aligned}
                \frac{\lpnorm[1]{\vec{\lambda}}}{\lpnorm[1]{\vec{\lambda^*}}}
                & \leq \frac{(1 + \varepsilon) \fnorm{M} \sqrt{\frac{n - 1}{n}}}{L_k \lpnorm[\infty]{M}} \\
                & \leq \frac{1 + \varepsilon}{L_k} \sqrt{\abs{\nz(M)} \frac{n-1}{n}} \\
                & \leq \frac{1 + \varepsilon}{L_k} \sqrt{\abs{\nz(A)} \frac{n-1}{n}} \in \LandauO(\sqrt{m}).
            \end{aligned}
        \end{equation}
    \end{proof}

    \section{Application to concrete systems}
    \label{sec:applications}
    Throughout this section we benchmark the informed \ac{LP} algorithm (\cref{alg:informed_lp}) against the uninformed \ac{LP} relaxation of~\citep{bassler2024,bassler2025,kum2026} and for qubit systems against the explicit construction of~\citep{garciadeandoin2026} as well.
    As reference, we additionally report the quantum run time $1 / \hat{\gamma}$ returned by the ray binary search (\cref{alg:ray_binary_search}) and the exact optimum of \eqref{eq:general_lp}, wherever it is computable.
    The guarantees of \cref{th:feasibility_window,th:mixed_algo} are stated in terms of $1 / \hat{\gamma}$, so we also report its value on concrete instances.

    The mixed \ac{LP} algorithm (\cref{alg:mixed}) is not included in the comparison.
    Its performance is governed by the split of the pulse budget into informed and uninformed samples.
    If all pulses are sampled informedly, it reduces to the informed \ac{LP} algorithm, and if all of them are sampled uniformly, to the uninformed \ac{LP} relaxation.
    At a fixed total number of samples, and on an instance for which the informed \ac{LP} algorithm is feasible, a benchmark would therefore display the chosen split instead of a property of the algorithm.
    Choosing $s_{\mathrm{i}}$ and $s_{\mathrm{u}}$ large enough for the guarantees of \cref{th:mixed_algo} to apply, on the other hand, introduces a sample overhead that is not comparable to the pulse numbers used by the other methods.

    The same reasoning applies to the mitigation of the finite pulse time error discussed in \cref{sec:mixed}.
    On the platforms addressed by the applications below one would in general want to suppress this error, but its magnitude depends on the pulse duration and on the drawn pulses, which adds further degrees of freedom that a direct comparison could not control for.
    We instead rely on the analytical statement of \cref{sec:mixed}, by which the additional cost is of first order in the pulse duration, so that for short pulses the run time stays close to the one attained by the informed \ac{LP} algorithm, which we therefore benchmark in its place.

    All quantum run times are reported as $\lpnorm[1]{\vec{\lambda}}$, that is in units of the evolution time under the native system Hamiltonian.
    The solver, the tolerances, the computation of the exact optimum, and the numerical inversion of $\RDF_k$ are collected in \cref{sec:numerics}.

    \subsection{Qubits}
    \label{sec:qubit_application}
    As introduced in \cref{sec:motivating_examples}, our method applies to qubit systems with Pauli conjugations in two ways.
    The $Z$-type Ising model with $X$-type conjugations is captured directly with one site per qubit, so that $n = q$ for a system of $q$ qubits.
    The symplectic encoding of \cref{eq:examples_single_pauli,eq:examples_single_interaction} instead takes $n = 2q + 1$ and reaches any $2$-local interaction of $X$ and $Z$ operators, as well as any single $X$, $Y$, or $Z$ term, at the price of excluding pairs that combine a $Y$ with another nontrivial operator.
    In both cases the commutation relations \eqref{eq:phase_commutation} are satisfied with $k = 2$, so that we cannot engineer complex coefficients and require $M$ to be a real-valued matrix.
    Consequently, $D = m$, and the pulse ratio $s / m$ coincides with the ratio $s / D$ that governs the feasibility window of \cref{th:feasibility_window}.
    The benchmarks reported below use the first encoding on $q = \num{20}$ qubits, so that $n = \num{20}$ and $m = \num{190}$.
    Since all $O_{ij} = Z_i Z_j$ commute in this encoding, the Trotter approximation \eqref{eq:trotter} is exact, i.e., every feasible pulse set realizes the target without Trotter error.

    Out of the three applications we consider, this one is the most studied, and two efficient methods are available for comparison.
    The first is the uninformed \ac{LP} relaxation, which was introduced as an efficient heuristic for Pauli conjugations by \citet{bassler2025}, and carried over to fermionic systems by \citet{kum2026}.
    It draws $s$ pulses uniformly at random and solves the \ac{LP}~\eqref{eq:restricted_lp} on them, and thus differs from the informed \ac{LP} algorithm (\cref{alg:informed_lp}) only in the sampling distribution.
    If $\hat{\gamma}$ were zero, that is if our sampling were centered on the identity, the two approaches would coincide.
    The second method is the explicit construction of \citet{garciadeandoin2026}, which determines a solution directly from a spectral decomposition of $M$ and requires no sampling at all.
    Whenever $M$ has degenerate eigenvalues, however, the spectral decomposition and with it the resulting quantum run time are not unique.
    We therefore evaluate the construction on \num{2000} Haar-random orthonormal bases of each degenerate eigenspace.
    It assumes the ability to implement arbitrary single-qubit gates instead of restricting to Pauli conjugations as we do, and engineers any $2$-local qubit Hamiltonian from an all-to-all connected Ising-type system, a strictly larger class of targets than ours.
    The comparison below is therefore not like for like, and the run times we report for this method should be read with that difference in mind.

    Before discussing the numerics, we evaluate our bounds in this setting.
    Previously, there were two upper bounds on the quantum run time of the optimal solution the \ac{LP}~\eqref{eq:general_lp} in the qubit case.
    A bound in $\lpnorm[\infty]{M}$ was conjectured by \citet{bassler2024}:
    \begin{equation}
        \lpnorm[1]{\vec{\lambda^*}} \leq \lpnorm[\infty]{M} \cdot \begin{cases}
            q, & \textnormal{for odd } q, \\
            q-1, & \textnormal{for even } q.
        \end{cases}
    \end{equation}
    A tighter bound in the Frobenius norm was proven by \citet{garciadeandoin2026a}:
    \begin{equation}
        \label{eq:previously_tight_bound}
        \lpnorm[1]{\vec{\lambda^*}} \leq \sqrt{3/2} \fnorm{M}.
    \end{equation}
    Evaluating \cref{th:ray_bound} for $k=2$ and $n = 2q + 1$ yields that there exists a solution $\vec{\lambda}$ fulfilling
    \begin{equation}
        \lpnorm[1]{\vec{\lambda}} \leq \frac{\pi}{2} \sqrt{\frac{2q}{2q + 1}} \fnorm{M}.
    \end{equation}
    Hence, this bound must in particular hold for the optimal solution $\vec{\lambda^*}$.
    Since $\frac{\pi}{2} > \sqrt{3/2}$, our bound is weaker than the bound \eqref{eq:previously_tight_bound} by a factor of at most $\pi / \sqrt{6} \approx 1.28$.
    Note moreover that the bound \eqref{eq:previously_tight_bound} holds for arbitrary $2$-local qubit Hamiltonians, and thus for a slightly larger class of problems.
    In the opposite direction, \cref{th:opt_bound} bounds the optimal quantum run time from below by $\lpnorm[\infty]{M}$, which in this setting recovers the qubit bound of \citet{bassler2024} that \cref{th:opt_bound} generalizes.

    We first consider the canonical instance shown in \cref{fig:complete_bipartite}, where an Ising-type target on the complete bipartite graph $K_{10,10}$ is to be engineered from an all-to-all connected Ising system on the same \num{20} qubits, with all couplings of both Hamiltonians set to one.
    The horizontal axis starts at $s / m = 2$ because, with $D = m$ in this setting, \cref{th:feasibility_window}~(i) rules out a feasible restricted linear program below that ratio with high probability, and the numerical study of \cref{sec:feasibility_numerics} locates the transition there as well.
    With $m_T = \num{100}$ target edges we have $\fnorm{M} = \sqrt{200}$, so that the guarantee of \cref{th:qtime_bound} evaluates to $\approx \num{21.7}$ (the bound of \citet{garciadeandoin2026a} gives $\approx \num{17.3}$), while the optimum is $\approx \num{10}$.
    The informed \ac{LP} algorithm attains $\approx \num{13}$ already at the smallest pulse ratio shown, in the runs where the restricted linear program is feasible, and settles within a few percent of the optimum from $s / m \approx 3$ onward, with a narrow spread across runs.
    It stays below $1 / \hat{\gamma} \approx \num{15.7}$ over the entire range shown.
    The uninformed \ac{LP} relaxation, in contrast, starts an order of magnitude above the optimum at $s / m = 2$, with a spread covering the full plot range, becomes competitive only beyond $s / m \approx 3$, and remains above the informed \ac{LP} algorithm at $s / m = 4$.
    On this instance, $M$ has an 18-fold degenerate eigenvalue, and depending on the chosen basis of its eigenspace the explicit construction yields between $\approx \num{35.9}$ and $\approx \num{59.0}$, with a median of $\approx \num{44.7}$, that is roughly four to six times the optimum.

    \begin{figure}
        \centering
        \includegraphics{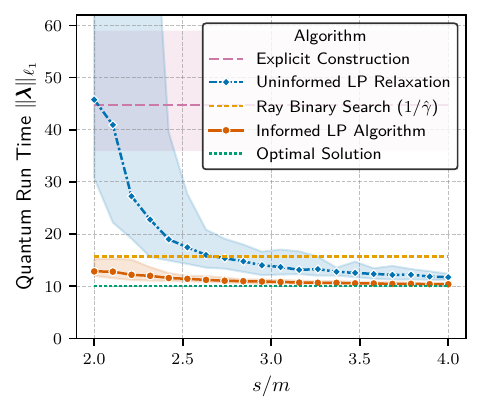}
        \caption[Comparison of algorithm performances on a canonical instance.]{
            Comparison of algorithm performances on a canonical instance.
            The system Hamiltonian is given by an all-to-all connected Ising model on \num{20}~qubits.
            The target is also an Ising-type Hamiltonian with a connectivity graph given by a complete bipartite graph~$K_{10, 10}$.
            All couplings of both Hamiltonians are set to~\num{1}.
            The pulse ratio is varied over \num{20} equally spaced values between \num{2.0} and \num{4.0}.
            For the randomized algorithms, the data is aggregated over \num{50}~independent runs per pulse ratio~$s/m$.
            The markers indicate the median value, the shaded areas range from the minimal to the maximal attained quantum run time.
            Runs with an infeasible restricted linear program are excluded, which affects only $s/m < \num{2.3}$ and, at $s/m = 2$, \num{34}~runs of the informed \ac{LP} algorithm and \num{37}~runs of the uninformed \ac{LP} relaxation.
            For the explicit construction, which does not depend on $s/m$, the line and the shaded area show the median and the range over \num{2000} random orthonormal bases of the degenerate eigenspace of~$M$.
        }
        \label{fig:complete_bipartite}
    \end{figure}

    To probe how the methods respond to the structure of the target, we next vary the target density, as shown in \cref{fig:edge_surpression}.
    Since all couplings are again set to one, $M$ is the adjacency matrix of the target graph, so that $\lpnorm[\infty]{M} = 1$ and \cref{th:opt_bound} places a floor of $\lpnorm[1]{\vec{\lambda^*}} \geq 1$ on every instance of the sweep.
    The optimum is unimodal in the target density:
    It rises from $\approx 2$ on the sparsest targets to $\approx \num{4.3}$ near a density ratio of $\num{0.6}$, decreases again, and collapses to the floor at full density, where the target coincides with the system Hamiltonian and the identity pulse alone is optimal.
    The upper bound of \cref{th:qtime_bound} grows as $\sqrt{m_T}$ over the same sweep and therefore does not reproduce this turnaround.
    The informed \ac{LP} algorithm follows the optimum to within \qty{15}{\percent} across the entire range, including the collapse at full density, and has the narrowest spread of all methods considered.

    \begin{figure}
        \centering
        \includegraphics{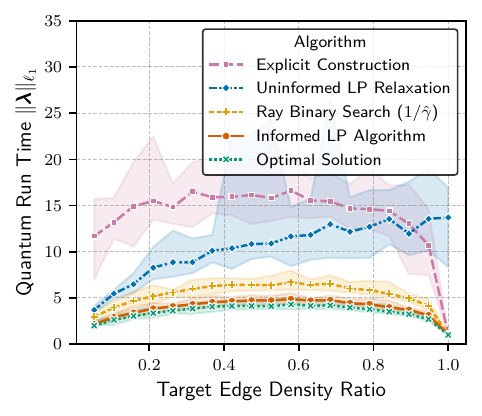}
        \caption[Comparison of algorithm performances on random instances of varying target density.]{
            Comparison of algorithm performances on random instances of varying target density.
            The system Hamiltonian is given by an all-to-all connected Ising model on \num{20}~qubits.
            The target is also an Ising-type Hamiltonian with the connectivity given by a random Erdős--Rényi graph with $m_T$ edges.
            All couplings of both Hamiltonians are set to~\num{1}.
            $m_T$ is varied from~\num{10} to the fully connected~\num{190} edges in a step size of~\num{10}, and the horizontal axis shows the resulting target edge density ratio $m_T / \num{190}$.
            For each target density, \num{20} independent instances were generated.
            The markers show the median quantum run time of each algorithm over the random instances.
            The shaded areas indicate the minimal and maximal quantum run time attained for any instance.
            For the randomized algorithms, the best of \num{20} independent runs with pulse ratio $s/m = 3$ is considered for each instance.
            For the explicit construction, the minimal quantum run time over \num{2000} random orthonormal bases of the degenerate eigenspaces of~$M$ is considered for each instance.
        }
        \label{fig:edge_surpression}
    \end{figure}

    The two established methods behave qualitatively differently from each other.
    The uninformed \ac{LP} relaxation increases monotonically with the target density, from $\approx \num{3.7}$ on the sparsest targets to $\approx \num{13.7}$ at full density, where it does not collapse.
    It thus requires roughly fourteen times the optimal quantum run time on what is the easiest instance of the entire sweep.
    The single optimal pulse is the identity, which is drawn by the uninformed sampling with vanishing probability $2^{- q + 1}$.
    The explicit construction, in contrast, starts at $\approx \num{11.7}$ on the sparsest targets, roughly six times the optimum, rises to a plateau between \num{15} and \num{16.5}, and declines only toward the densest targets.
    It is overtaken by the uninformed \ac{LP} relaxation only in the last few percent of the sweep.
    It does recover the trivial solution at full density, but its spread is the widest of all methods on sparse targets and comparable to that of the uninformed \ac{LP} relaxation elsewhere.
    Both the informed \ac{LP} algorithm and $1 / \hat{\gamma}$ reproduce the unimodal shape of the optimum across the sweep, whereas the uninformed \ac{LP} relaxation does not and the explicit construction follows it only qualitatively, at four to six times the optimum.

    The two benchmarks show that on qubit systems the informed \ac{LP} algorithm reaches near-optimal quantum run times with a number of pulses linear in the number of system interaction terms, and that it stays below both established methods at every point of both sweeps.
    The worst-case approximation ratio of $\LandauO(\sqrt{m})$ established in \cref{sec:bounds} overestimates the quantum run times for our example system.

    \subsection{Qudits}
    \label{sec:qudit_application}
    As introduced in \cref{sec:motivating_examples}, our method applies to $q$ qudits with $d$ levels through the interaction terms $O_{ij} = P_i P_j^\dagger$ on $n = 2q + 1$ sites, where \cref{eq:phase_commutation} holds for every $k$ dividing $d$, so that $D = 2m$ for $k \geq 3$.
    In contrast to the fermionic case, the admissible phase sets are thus restricted by the hardware, with $k = d$ giving the finest one.

    Hamiltonian engineering for qudits by conjugation with Weyl operators was introduced by \citet{alvarez-ahedo2025}.
    They show that a decomposition with nonnegative evolution times exists among all $d^{2q}$ Weyl conjugations, by extending an earlier qubit construction~\citep{garcia-de-andoin2024}, but give no procedure that efficiently returns such a decomposition or minimizes the quantum run time.
    Robust Hamiltonian engineering for strongly interacting qudits was moreover demonstrated on spin-$1$ nitrogen-vacancy centers by \citet{zhou2024}.

    In our framework, the admissible terms are the clock couplings $Z_a Z_b^\dagger$, the shift couplings $X_a X_b^\dagger$, and the mixed couplings $X_a Z_b^\dagger$ between qudits $a \not = b$, the terms $X_a Z_a^\dagger$ on a single qudit, and the fields $X_a$ and $Z_a$, each together with its Hermitian conjugate.
    This excludes, for instance, the spin-$1$ coupling $S^z_a S^z_b$ considered as native Hamiltonian by \citet{alvarez-ahedo2025}, which also contains the terms $Z_a Z_b$ and $Z_a^\dagger Z_b^\dagger$.
    As native Hamiltonian, we take all admissible terms that act on two sites of different qudits with equal strength,
    \begin{equation}
        H_S = - J \sum_{\Set{i, j} \in E} \left( P_i P_j^\dagger + P_j P_i^\dagger \right),
    \end{equation}
    where $E$ contains all pairs of distinct sites $i \not = j \in [2q+1]$, except the $q$ single-site pairs $\Set{a, q + a}$.
    We take the chiral clock model~\citep{huse1981,ostlund1981} on an open chain as target,
    \begin{equation}\label{eq:ChiralClockModel}
        \begin{split}
            H_T &= - J \sum_{a = 1}^{q - 1} \left( \e^{\i \varphi} Z_a Z_{a+1}^\dagger + \e^{- \i \varphi} Z_a^\dagger Z_{a+1} \right) \\
            &\quad - g \sum_{a = 1}^{q} \left( \e^{\i \vartheta} X_a + \e^{- \i \vartheta} X_a^\dagger \right),
        \end{split}
    \end{equation}
    with chiral phases $\varphi$ and $\vartheta$ and transverse field $g$.
    It reduces to the quantum clock model for $\varphi = \vartheta = 0$ and to the transverse field Ising model for $d = 2$.
    The instance thus combines the suppression of interactions, as in \cref{fig:edge_surpression}, with complex target phases.

    \begin{figure*}
        \centering
        \includegraphics{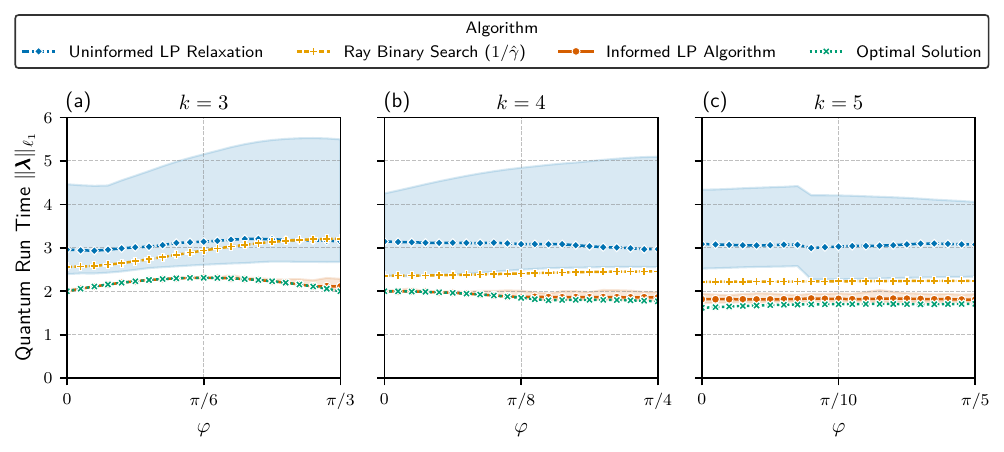}
        \caption[Comparison of algorithm performances on the chiral clock model for varying chiral phase.]{
            Comparison of algorithm performances on the chiral clock model for varying chiral phase.
            The system Hamiltonian contains all admissible two-site terms between different qudits of a register of $q = \num{4}$~qudits, all with strength~$J$.
            The target is the chiral clock model \eqref{eq:ChiralClockModel} on an open chain at vanishing transverse field $g = 0$, so that only the clock couplings $Z_a Z_{a+1}^\dagger$ with chiral phase~$\varphi$ are engineered. 
            All other terms of the system Hamiltonian are suppressed.
            The chiral phase is varied over \num{21} equally spaced values in $[0, \pi / k]$.
            The panels correspond to the phase sets $k = 3$, $4$, and $5$.
            For the randomized algorithms, the data is aggregated over \num{50}~independent runs per value of~$\varphi$ at a pulse ratio of $s / D = 3$.
            The markers show the median value, the shaded areas range from the minimal to the maximal attained quantum run time.
        }
        \label{fig:qudit_chiral_clock_phi}
    \end{figure*}

    \begin{figure*}
        \centering
        \includegraphics{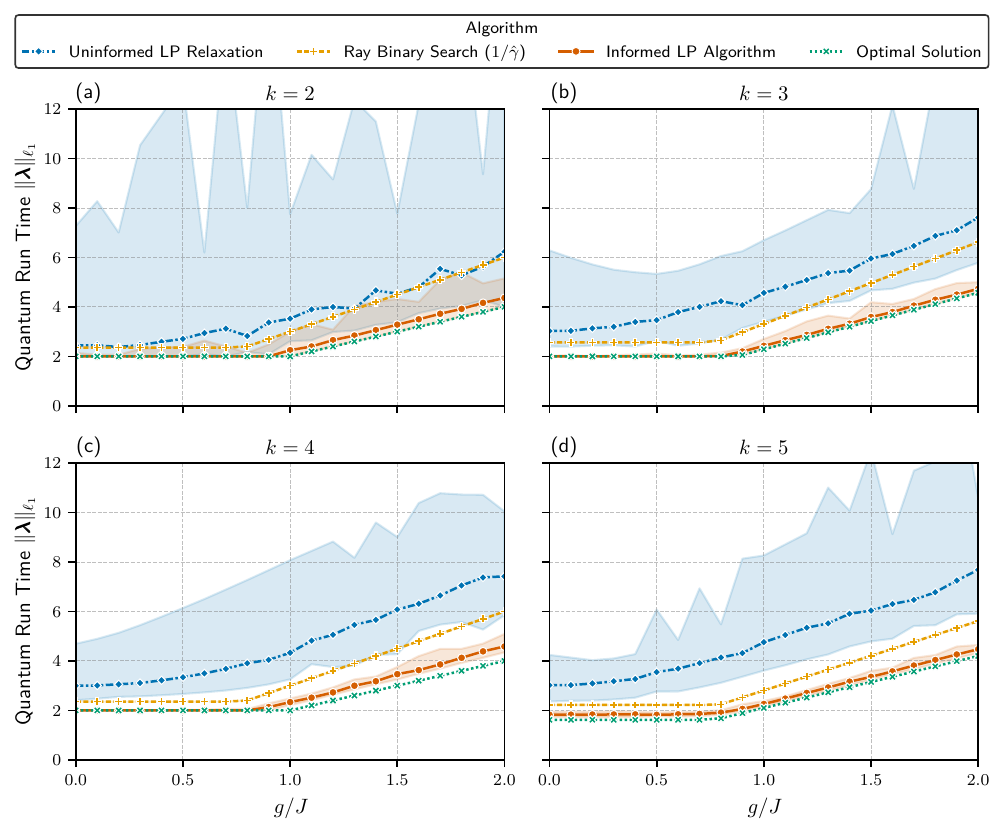}
        \caption[Comparison of algorithm performances on the quantum clock model for varying transverse field.]{
            Comparison of algorithm performances on the quantum clock model for varying transverse field.
            The system Hamiltonian contains all admissible two-site terms between different qudits of a register of $q = \num{4}$~qudits, all with strength~$J$.
            The target is the chiral clock model on an open chain at $\varphi = \vartheta = 0$, that is, the quantum clock model, with the transverse field~$g$ varied over \num{21} equally spaced values of $g / J \in [0, 2]$.
            The panels correspond to the phase sets $k = d = 2$, $3$, $4$, and $5$, where $k = 2$ is the transverse field Ising model.
            For the randomized algorithms, the data is aggregated over \num{50}~independent runs per value of~$g$ at a pulse ratio of $s / D = 3$.
            Runs in which the restricted linear program was infeasible are excluded.
            This occurred only for $k = 2$, in \num{21} of the \num{1050} runs of the informed \ac{LP} algorithm and in \num{13} of the \num{1050} runs of the uninformed \ac{LP} relaxation.
            The markers show the median value, the shaded areas range from the minimal to the maximal attained quantum run time.
        }
        \label{fig:qudit_chiral_clock_g}
    \end{figure*}

    Shifting the pulse phases by suitable elements of $\Theta_k$, or reflecting them, shows that the optimal quantum run time and the distributions of the quantum run times attained by both \ac{LP} methods are periodic in $\varphi$ and $\vartheta$ with period $\frac{2 \pi}{k}$ and invariant under $(\varphi, \vartheta) \to (- \varphi, - \vartheta)$, so that we fix $\vartheta = 0$ and consider $\varphi \in [0, \frac{\pi}{k}]$.
    To our knowledge, \cref{th:opt_bound,th:qtime_bound} provide the first bounds on the quantum run time for qudits.

    We benchmark the methods on $q = \num{4}$ qudits, so that $m = \num{32}$, with $D = \num{64}$ for $k \geq 3$ and $D = m = \num{32}$ for $k = 2$, at the pulse ratio $s / D = 3$.
    At this size, the exact optimum of the \ac{LP}~\eqref{eq:general_lp} follows from enumerating all $k^{2q} \leq 5^8 \approx \num{3.9e5}$ pulses for $k \leq 5$.
    With all system couplings of equal strength, $\lpnorm[\infty]{M} = \max \Set{1, g / J}$ and $\fnorm{M} = \sqrt{2 (q - 1) + 2 q g^2 / J^2}$.
    \Cref{th:opt_bound} therefore places a floor of $\max \Set{1, g / J}$ on the optimal quantum run time.
    At $g = 0$, the guarantee of \cref{th:qtime_bound} evaluates to $\approx \num{3.6}$ for $k = 2$ and $k = 4$, to $\approx \num{3.4}$ for $k = 5$, and to $\approx \num{20.0}$ for $k = 3$, as the constant $L_3 \approx \num{0.12}$ is considerably smaller than $L_4 = 2 / \pi$ and $L_5 \approx \num{0.69}$.

    We first vary the chiral phase at $g = 0$, where only the chiral clock couplings are to be engineered and all other terms of $H_S$ are to be suppressed.
    We then vary the transverse field at $\varphi = 0$, that is, for the quantum clock model with $k = d$, which for $d = 2$ is the transverse field Ising model.
    Our results are shown in \cref{fig:qudit_chiral_clock_phi,fig:qudit_chiral_clock_g}.

    \begin{table*}
        \centering
        \sisetup{range-phrase = --, range-units = single}
        \caption[Summary of the algorithm performances on the qudit benchmarks.]{
            Summary of the algorithm performances on the qudit benchmarks of \cref{fig:qudit_chiral_clock_phi,fig:qudit_chiral_clock_g}.
            Each row aggregates the \num{21}~values of the swept parameter for one phase set~$k$.
            The third column gives the range of the optimal quantum run time $\lpnorm[1]{\vec{\lambda^*}}$ over the sweep.
            All further entries are ratios to the optimal quantum run time at the same parameter value.
            The ``Median'' columns give the range of the ratio of the median over the \num{50}~runs, and the ``Worst'' columns the largest ratio attained by any run, both over the entire sweep.
            The column ``Opt.'' counts the parameter values at which the median of the informed \ac{LP} algorithm equals the optimum.
            The column $1 / \hat{\gamma}$ gives the range over all runs, and the column ``Guarantee'' the range of the guarantee of \cref{th:qtime_bound}.
        }
        \label{tab:qudit_summary}
        \begin{tabular}{llcccccccc}
            \toprule
            & & & \multicolumn{3}{c}{Informed \ac{LP}} & & & \multicolumn{2}{c}{Uninformed \ac{LP}} \\
            \cmidrule(lr){4-6} \cmidrule(lr){9-10}
            Sweep & $k$ & $\lpnorm[1]{\vec{\lambda^*}}$ & Median & Worst & Opt. & $1 / \hat{\gamma}$ & Guarantee & Median & Worst \\
            \midrule
            $\varphi$ & $3$ & \numrange{2.00}{2.31} & \numrange{1.00}{1.06} & \num{1.14} & \num{18} & \numrange{1.21}{1.60} & \numrange{8.68}{10.02} & \numrange{1.34}{1.57} & \num{2.75} \\
            & $4$ & \numrange{1.77}{2.00} & \numrange{1.00}{1.05} & \num{1.13} & \num{9} & \numrange{1.18}{1.39} & \numrange{1.81}{2.05} & \numrange{1.57}{1.72} & \num{2.88} \\
            & $5$ & \numrange{1.62}{1.71} & \numrange{1.06}{1.12} & \num{1.20} & \num{0} & \numrange{1.31}{1.37} & \numrange{1.97}{2.08} & \numrange{1.76}{1.91} & \num{2.68} \\
            \midrule
            $g$ & $2$ & \numrange{2.00}{4.00} & \numrange{1.00}{1.13} & \num{1.55} & \num{10} & \numrange{1.18}{1.50} & \numrange{1.81}{2.77} & \numrange{1.19}{1.77} & \num{8.99} \\
            & $3$ & \numrange{2.00}{4.57} & \numrange{1.00}{1.07} & \num{1.24} & \num{8} & \numrange{1.28}{1.45} & \numrange{10.02}{14.05} & \numrange{1.51}{2.12} & \num{3.32} \\
            & $4$ & \numrange{2.00}{4.00} & \numrange{1.00}{1.17} & \num{1.32} & \num{9} & \numrange{1.18}{1.50} & \numrange{1.81}{2.77} & \numrange{1.50}{2.19} & \num{4.04} \\
            & $5$ & \numrange{1.62}{4.20} & \numrange{1.06}{1.15} & \num{1.21} & \num{0} & \numrange{1.34}{1.37} & \numrange{2.01}{2.72} & \numrange{1.79}{2.46} & \num{4.31} \\
            \bottomrule
        \end{tabular}
    \end{table*}

    \Cref{tab:qudit_summary} summarizes both sweeps.
    Over the chiral phase sweep of \cref{fig:qudit_chiral_clock_phi}, the optimum stays within a factor of \num{2.3} of the floor of \cref{th:opt_bound} for all three phase sets.
    For $k = 3$, it rises from $\varphi = 0$ to a maximum at $\varphi = \pi / 6$ and returns to its initial value at $\varphi = \pi / 3$, whereas it decreases over the sweep for $k = 4$ and increases slightly for $k = 5$.
    The informed \ac{LP} algorithm attains the optimum in the median for $\varphi \leq 3 \pi / 10$ at $k = 3$ and for $\varphi \lesssim \pi / 9$ at $k = 4$, and even its worst run stays within \qty{21}{\percent} of the optimum across the entire sweep.
    The quantity $1 / \hat{\gamma}$ lies above every run of the informed \ac{LP} algorithm.
    The uninformed \ac{LP} relaxation is nearly insensitive to both $\varphi$ and $k$, with medians between $\approx \num{2.9}$ and $\approx \num{3.2}$.
    Its spread is by far the widest, and for $k = 4$ and $k = 5$ its worst runs exceed the guarantee of \cref{th:qtime_bound} for the informed \ac{LP} algorithm on every instance of the sweep.
    In particular, the uninformed \ac{LP} relaxation does not profit from the finer phase set $k = 5$, which lowers the optimum at $\varphi = 0$ from \num{2} to $\approx \num{1.62}$ and the median of the informed \ac{LP} algorithm from \num{2} to $\approx \num{1.82}$.

    Over the transverse field sweep of \cref{fig:qudit_chiral_clock_g}, the optimum is constant for small fields and grows approximately linearly in $g / J$ beyond a crossover slightly below $g = J$, again staying within a factor of \num{2.3} of the floor of \cref{th:opt_bound}.
    The finer phase set $k = 5$ yields the shortest optimal quantum run time only at small fields, while for $g \gtrsim J$ the phase sets $k = 2$ and $k = 4$ do.
    The informed \ac{LP} algorithm again attains the optimum in the median for $g \lesssim \num{0.8} J$ at $k = 2$, $3$, and $4$, and for $k \geq 3$ it stays below $1 / \hat{\gamma}$ in every run.
    The uninformed \ac{LP} relaxation again has the widest spread, and for $k = 2$ its worst run requires almost nine times the optimal quantum run time.

    The small constant $L_3$ renders the guarantee of \cref{th:qtime_bound} for $k = 3$ considerably looser than for the other phase sets.
    On the qudit benchmark the informed \ac{LP} algorithm attains quantum run times within \qty{17}{\percent} of the optimum in the median, across both sweeps and all phase sets, whereas the uninformed \ac{LP} relaxation requires up to \num{2.5} times the optimum.
    Only the informed \ac{LP} algorithm follows the reduction of the optimal quantum run time that the finer phase set $k = 5$ allows at small fields.

    \subsection{Fermions}
    \label{sec:fermion_application}
    As introduced in \cref{sec:motivating_examples}, our method applies to fermionic systems on $n$ modes with $O_{ij} = c_i^\dagger c_j$, conjugated by the exponentiated number operators $U_{\vec{\theta}} = \prod_{i=1}^{n} \e^{-\i \theta_i c_i^\dagger c_i}$.
    In contrast to the qubit and the qudit cases, the commutation relations \eqref{eq:phase_commutation} are satisfied for every $k \in \NN_{\geq 2} \cup \Set{\infty}$. 
    Finite $k$ correspond to discrete sets of $k$ equally spaced on-site phases, $k = \infty$ to continuously tunable local potentials.
    For $k \geq 3$, complex entries of $M$ can be engineered, so that $D = 2m$ and the ratio $s / D$ that governs the feasibility window of \cref{th:feasibility_window} equals $s / (2m)$.

    Hamiltonian engineering with these conjugations was investigated by \citet{kum2026}, who engineer arbitrary number operator free Hamiltonians and adapt the uninformed \ac{LP} relaxation of \citet{bassler2025} to this setting.
    Within our framework, only tunneling terms of the form $c_i^\dagger c_j$ can be engineered.
    Hamiltonians consisting only of such terms describe free fermions and can be simulated efficiently on classical computers~\citep{terhal2002}.
    The relevance of engineering them stems from the terms they may be combined with.
    Consider the Hubbard model~\citep{hubbard1963} on spinful modes $c_{i\sigma}$ with $\sigma \in \Set{\uparrow, \downarrow}$, where the system Hamiltonian
    \begin{equation}
        H_S = H_{\mathrm{hop}} + U \sum_{i} c_{i\uparrow}^\dagger c_{i\uparrow} c_{i\downarrow}^\dagger c_{i\downarrow}
    \end{equation}
    consists of a tunneling part $H_{\mathrm{hop}}$ with uniform hopping amplitude $J$ and an on-site interaction of strength $U$.
    Every number operator $c_{i\sigma}^\dagger c_{i\sigma}$ commutes with every $U_{\vec{\theta}}$, so that the interaction term is invariant under all conjugations.
    Any feasible $\vec{\lambda}$ for the tunneling part therefore yields
    \begin{equation}
        \begin{split}
            &\sum_{\vec{\theta} \in \nz(\vec{\lambda})} \lambda(\vec{\theta}) U_{\vec{\theta}}^\dagger H_S U_{\vec{\theta}} \\
            &\quad = H_{T} + \lpnorm[1]{\vec{\lambda}} \, U \sum_{i} c_{i\uparrow}^\dagger c_{i\uparrow} c_{i\downarrow}^\dagger c_{i\downarrow},
        \end{split}
    \end{equation}
    where $H_T$ is the engineered tunneling Hamiltonian.
    For a target with uniform hopping amplitude $\lpnorm[\infty]{M} J$, the engineered model hence has the interaction ratio
    \begin{equation}
        \label{eq:hubbard_ratio}
        \frac{U_T}{J_T} = \frac{\lpnorm[1]{\vec{\lambda}}}{\lpnorm[\infty]{M}} \frac{U}{J},
    \end{equation}
    which is invariant under a rescaling of $M$.
    By \cref{th:opt_bound}, the prefactor is at least one, so that at a fixed native interaction strength the ratio can only be raised, and the smallest ratio that is accessible for a given target geometry is set by the attained quantum run time.
    In this application, $\lpnorm[1]{\vec{\lambda}}$ thus additionally determines which regime of the interacting model can be reached.
    This use of Hamiltonian engineering to reach $U/J$ ratios beyond those realized natively was established by \citet{kum2026} and is the primary application of their work.

    The spinful tunneling part consists of two decoupled copies of the lattice, one per spin species.
    Restricting the pulses of a feasible spinful solution to one species yields a feasible spinless solution with the same quantum run time, and applying the phases of a spinless solution to both species yields a spinful one.
    The solutions of both problems therefore coincide, and it suffices to benchmark the spinless problem with $n$ equal to the number of lattice sites.

    The only efficient method available for comparison is the uninformed \ac{LP} relaxation in the form of \citet{kum2026}.
    The explicit construction of \citet{garciadeandoin2026} is formulated for qubits, and the exact optimum of the \ac{LP}~\eqref{eq:general_lp} is out of reach, since it requires $k^n$ variables for finite $k$ and infinitely many for $k = \infty$.
    We therefore report the informed \ac{LP} algorithm, the uninformed \ac{LP} relaxation, and $1 / \hat{\gamma}$.

    To our knowledge, no bounds on the optimal quantum run time were previously available in this setting.
    For targets with unit modulus entries on all system and target interactions, as considered below, $\lpnorm[\infty]{M} = 1$, so that \cref{th:opt_bound} yields $\lpnorm[1]{\vec{\lambda^*}} \geq 1$, and \cref{eq:hubbard_ratio} reduces to $U_T/J_T = \lpnorm[1]{\vec{\lambda}} U/J$.
    In the other direction, $\fnorm{M} = \sqrt{2m}$, and the guarantee of \cref{th:qtime_bound} evaluates to
    \begin{equation}
        \lpnorm[1]{\vec{\lambda}} \leq \frac{\sqrt{2m}}{L_k} \sqrt{\frac{n - 1}{n}}.
    \end{equation}

    We benchmark the methods on the Hofstadter model~\citep{harper1955,hofstadter1976}, which describes charged particles hopping on a square lattice in a perpendicular magnetic field.
    The system is an $L \times L$ square lattice with open boundary conditions, $n = L^2$ modes, and all $m = 2L(L-1)$ nearest-neighbor hopping amplitudes set to one.
    The target describes a uniform magnetic flux $\Phi$ per plaquette, encoded in Peierls phases in the Landau gauge, that is $M_{ij} = \e^{\i \Phi y}$ on the horizontal bond between $i$ and $j$ in row $y$, and $M_{ij} = 1$ on all vertical bonds.
    We choose $\Phi$ to be the inverse golden mean in units of flux quantum $\Phi = \frac{\sqrt{5} - 1}{2} \cdot 2\pi$.
    The number of flux quanta per plaquette is thus irrational, and the target admits no magnetic unit cell, that is, no two rows of horizontal bonds carry the same phase.
    For a rational flux $\Phi = 2\pi p/q$, in contrast, the target is periodic in $y$ with period $q$ and consists of a fixed pattern repeated across the lattice.
    The scaling of the quantum run time with $L$ could then reflect this repetition instead of the behavior on targets without such structure \cite{kum2026}. 
    On a finite lattice, the target is approximately periodic whenever $\Phi$ is close to $2\pi p/q$ with $q \leq L$, so that irrationality alone does not rule out this effect.
    The golden mean and its inverse are the irrational numbers least well approximated by rationals, so that the target stays maximally far from periodic for every considered $L$.
    Moreover, the Peierls phases $\Phi y \bmod 2\pi$ are spread nearly uniformly over the unit circle, and none of them lies in $\Theta_k$ for $y \geq 1$, so that the target favors no particular phase set.
    This is in contrast to the rational flux $\Phi = 2\pi/3$ considered by \citet{kum2026}, whose target repeats every three rows and whose Peierls phases all lie in $\Theta_3$.
    Additionally, the Hofstadter model at this flux has been studied extensively, in particular its fractal spectrum and critical eigenstates~\citep{ostlund1983,tang1986,hiramoto1989}, typically via Fibonacci approximants~\citep{thouless1983}.
    
    Although the free model can be simulated efficiently on classical computers~\citep{terhal2002}, one may still wish to implement it on a quantum simulator, possibly combined with the on-site interaction discussed above.
    Such an implementation requires large lattices, which makes the scaling of the quantum run time with $L$ a relevant metric.
    Since the lattice has open boundary conditions, any two gauges of the same flux are related by a conjugation with some $U_{\vec{\theta}}$.
    Shifting all pulses of a solution by $\vec{\theta}$ preserves its quantum run time, and the sampling distributions of both \ac{LP} methods transform covariantly under this shift.
    For $k = \infty$, the optimal quantum run time and the distributions of the quantum run times attained by all considered methods are therefore gauge invariant.
    For finite $k$, this holds only for gauge transformations with phases in $\Theta_k$, so that the reported quantum run times refer to the Landau gauge with a fixed origin of $y$, and a different gauge, or merely a shifted origin, may yield different values.
    With $m$ system terms, the guarantee of \cref{th:qtime_bound} grows linearly in $L$ and evaluates at $L = \num{20}$ to $\approx \num{338}$, \num{61}, and \num{50} for $k = 3$, $4$, and $\infty$, respectively.
    The pulse ratio is fixed to $s / m = 6$, that is $s / D = 3$, which matches the ratio used in \cref{fig:edge_surpression}.
    Our results are shown in \cref{fig:hofstadter}.

    \begin{figure*}
        \centering
        \includegraphics{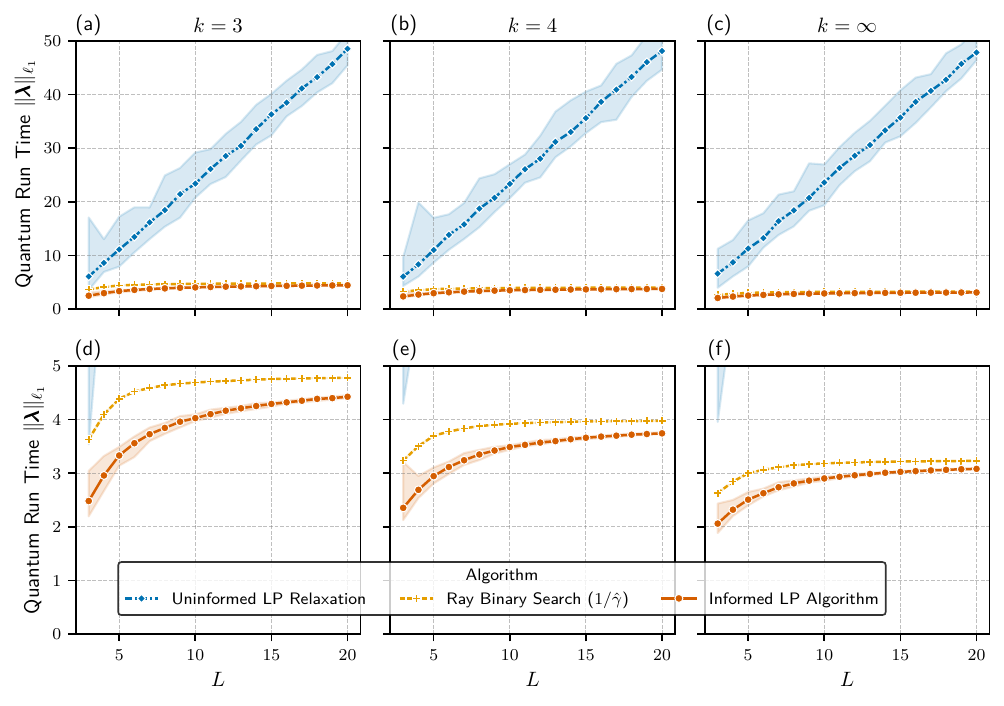}
        \caption[Comparison of algorithm performances on the Hofstadter model on varying sizes.]{
            Comparison of algorithm performances on the Hofstadter model on varying sizes.
            The system is given by an $L \times L$ square lattice with open boundary conditions on $n = L^2$ fermionic modes, with all hopping terms in the lattice set to~\num{1}.
            The target is to engineer a Hofstadter gauge field in the Landau gauge with a magnetic flux of $\Phi = \frac{\sqrt{5} - 1}{2} \cdot 2 \pi$ per plaquette.
            The linear lattice size $L$ is varied from~\num{3} to~\num{20}.
            The columns correspond to the phase sets $k = 3$, $k = 4$, and $k = \infty$, and panels~(d)--(f) show the same data as panels~(a)--(c) on a reduced vertical range.
            For the randomized algorithms, the data is aggregated over \num{50}~independent runs per lattice size at a pulse ratio of $s/m = 6$, that is $s / D = 3$.
            The markers show the median value, the shaded areas range from the minimal to the maximal attained quantum run time.
            No exact optimum is shown, since it is not computable at these system sizes.
        }
        \label{fig:hofstadter}
    \end{figure*}

    The uninformed \ac{LP} relaxation grows approximately linearly in $L$ and reaches $\approx \num{48}$ at $L = \num{20}$, with essentially no dependence on $k$.
    Its spread is widest on the smallest lattices.
    The informed \ac{LP} algorithm, in contrast, converges with increasing $L$ to $\approx \num{4.4}$, \num{3.7}, and \num{3.1} for $k = 3$, $4$, and $\infty$, respectively, with a narrow spread across runs.
    On the largest lattices, it thus requires between roughly one eleventh and one fifteenth of the quantum run time of the uninformed \ac{LP} relaxation, and its run time is effectively independent of the system size, while the guarantee of \cref{th:qtime_bound} grows as $\LandauO(L) = \LandauO(\sqrt{m})$.
    The quantity $1 / \hat{\gamma}$ converges as well, to $\approx \num{4.8}$, \num{4.0}, and \num{3.2}, and stays above the informed \ac{LP} algorithm throughout, with a gap that narrows as $L$ grows.

    Coarser phase sets require longer quantum run times, an ordering that the constants $L_k$ in the guarantee of \cref{th:qtime_bound} share.
    By \cref{eq:hubbard_ratio}, passing from $k = 3$ to continuously tunable local potentials lowers the smallest accessible interaction ratio $U_T / J_T$ by roughly \qty{30}{\percent}.

    On a lattice model with complex targets the informed \ac{LP} algorithm attains quantum run times that do not grow with the system size, using a number of pulses linear in $D$, whereas the quantum run time of the uninformed \ac{LP} relaxation grows linearly in $L$.
    
    \section{Conclusion and outlook}
    \label{sec:conclusion}
    We have introduced a unifying framework for Hamiltonian engineering that captures qubit, qudit, and fermionic systems through a single phase commutation relation between the interaction terms and the control pulses.
    Under this relation, the optimal quantum run time is determined by the point at which a ray originating from~$\1$ in the problem-specific direction~$M$ leaves the complex $k$-cut polytope~$\CUT_k^n$.
    We have proved the associated decision problem \NP-hard for all $k$ and \NP-complete for finite $k$ (\cref{th:ray_hardness}), extending the known case $k = 2$.
    We therefore relaxed the problem to the elliptope~$\elliptope_n$ and applied Krivine rounding, inverting the systematic distortion of the randomized rounding before sampling, which yields pulses that are \emph{informed} of the system and target Hamiltonians.

    Based on these pulses we derived two algorithms.
    The informed \ac{LP} algorithm (\cref{alg:informed_lp}) restricts a linear program to the sampled pulses, which is robustly feasible with nonnegligible probability only for roughly $2 D_{\mathrm{eff}}$ or more pulses (\cref{th:feasibility_window}).
    The sufficient number of pulses depends on the instance through the positivity index~$\beta$ and is in general not linear in $m$.
    On typical instances, however, we observed numerically that feasibility sets in sharply at the necessary number of pulses.
    The mixed \ac{LP} algorithm (\cref{alg:mixed}) additionally samples uniform pulses and, on any instance, returns a feasible solution with quantum run time at most $(1 + \varepsilon) / \hat{\gamma}$ from $\LandauO(m)$ pulses with high probability (\cref{th:mixed_algo}), while allowing the finite pulse time error to be suppressed at an additive cost of first order in the pulse duration.
    For both algorithms, an upper bound on the quantum run time linear in $\fnorm{M}$ (\cref{th:qtime_bound}) and a lower bound on the optimal quantum run time in $\lpnorm[\infty]{M}$ (\cref{th:opt_bound}) yield a worst-case approximation ratio of $\LandauO(\sqrt{m})$.
    The lower bound generalizes a result previously restricted to qubits, the upper bound has the same form as the best known qubit bound~\citep{garciadeandoin2026a} and extends it to every phase set $k$, and to our knowledge both are the first such bounds for qudits and fermions.

    In numerical benchmarks, the informed \ac{LP} algorithm performed far better than the established worst-case guarantees, using a number of pulses linear in the number of interaction terms.
    It attained near-optimal quantum run times wherever the optimum was computable, within \qty{15}{\percent} on random qubit targets and within \qty{17}{\percent} in the median on the qudit benchmark, and outperformed the uninformed \ac{LP} relaxation of \citet{bassler2025} as well as, on qubits, the explicit construction of \citet{garciadeandoin2026}.
    On the fermionic Hofstadter model, its quantum run time did not grow with the lattice size, whereas that of the uninformed \ac{LP} relaxation grew linearly in the linear lattice size. 
    By \cref{eq:hubbard_ratio}, a short quantum run time directly lowers the smallest interaction ratio accessible in an engineered Hubbard model.

    Several questions remain open.
    The gap between the worst-case guarantees and the observed performance, both in the approximation ratio and in the number of pulses required by \cref{th:feasibility_window}~(ii), may be closed via instance-dependent guarantees, for example by bounding the positivity index~$\beta$.
    Ray-tracing techniques that account for a re-entry of the bent ray into~$\elliptope_n$ could tighten the solutions, and a sharper bound on the constant $L_3$, which numerically appears to be close to $1/2$, compared to the proven $\approx 0.12$, would tighten the guarantees for $k = 3$.
    A detailed analysis of the finite pulse time error, including the winding freedom for $k > 2$, and of the Trotter error of noncommuting solutions with equal quantum run time is left to future work.
    Finally, extending the framework to qubit interactions involving $Y$ operators, to the full Weyl--Heisenberg group for qudits, or to higher-order terms in fermionic systems would broaden the class of engineerable Hamiltonians.

    \section*{Acknowledgments}
    We are grateful to Matthias Zipper, Mirko Arienzo, and Finn Müller for fruitful discussions on the project.
    Friese and Kliesch are funded by the Hamburg Quantum Computing project, which is co-financed by the ERDF of the European Union and the Fonds of the Hamburg Ministry of Science, Research, Equalities and Districts (BWFGB);
    Kum is
    funded by the TouQan project within the QuantERA~II Programme that has received funding from the EU's H2020 research and innovation programme under Grant Agreement No.\ 101017733,
    and from the Deutsche Forschungsgemeinschaft (DFG, German Research Foundation) under the grant number 532779266; 
    Harrow is funded by US NSF grant PHY-2325080 and a Fulbright Scholar Grant; and
    Kliesch is also funded by the Fujitsu Germany GmbH and Dataport as part of the endowed professorship ``Quantum Inspired and Quantum Optimization.''

    Google Gemini 3.1 Pro was used to explore the properties of the rounding distortion function, and to assist with the proofs in \cref{sec:expectation}.
    Anthropic Claude Opus 5 was used to explore approaches to the proofs in \cref{sec:hardness,sec:feasibility_proofs,sec:mixed_proofs}.
    Anthropic Claude Opus 5.5 was used to improve the usability of the code repository.
    Anthropic Claude Opus 5.5 and Fable 5.1 were used for general copyediting and improving the clarity of the paper.
    All AI-generated code and all technical and literature claims were thoroughly verified by the authors.
    The authors take full responsibility for the content of this submission.

    \section*{Code availability}
    The source code used to generate all numerical results and figures in this paper is publicly available under the MIT license~\citep{friese2026}.

    \bibliographystyle{quantum-natbib}
    \bibliography{cphe.bib}

    \clearpage
    \onecolumn
    \appendix
    \crefalias{section}{appendix} 
    \part*{Appendices}
    \section{Hardness of the cut polytope membership problem}
    \label{sec:hardness}
    In this section, we prove \cref{th:membership_hardness}.
    That is, we show \NP-hardness of the $\CUT_k^n$ membership problem (\cref{prob:cut_membership}) for every fixed $k \in \NN_{\geq 3} \cup \Set{\infty}$.
    Moreover, we derive \NP-completeness for finite $k$.
    Recall that \cref{prob:cut_membership} is known to be \NP-complete for $k=2$.

    \begin{proof}[Proof of \cref{th:membership_hardness} (finite $k$)]
        Let $X \in \Herm_n^{\vec{1}}(\QQ(\i))$ be an instance of the $\CUT_2^n$ membership problem.
        For fixed $k \in \NN_{\geq 3}$, we construct an instance $Y \in \Herm_n(\QQ(\uroot_k, \i))$ of the $\CUT_k^n$ membership problem, such that $X \in \CUT_2^n$ if and only if $Y \in \CUT_k^n$.
        First, verify that $X$ is real and $\lpnorm[\infty]{X} \leq 1$, as otherwise $X \not \in \CUT_2^n$, and we may construct an arbitrary invalid $Y \not \in \CUT_k^n$.
        Thus, w.l.o.g., assume in the following that $X \in \Sym_n(\QQ)$ and $\lpnorm[\infty]{X} \leq 1$.
        Denote by $J$ the $n \times n$ all-ones matrix and by $\vec{1}$ the $n$-element all-ones vector.
        Moreover, let $\vec{u}$ denote the first row of $X$, i.e.\ $u_{j} = X_{1j}$, and use the short forms $c_k \coloneqq \cos(2 \pi / k)$, and $s_k \coloneqq \sin(2 \pi / k)$.
        Then we construct
        \begin{equation}
            Y \coloneqq \frac{1 + c_k}{2} J + \frac{1 - c_k}{2} X + \i \frac{s_k}{2} (\vec{1} \vec{u}^T - \vec{u} \vec{1}^T).
        \end{equation}
        $Y$ is Hermitian by construction.
        Moreover, it consists of a constant number of algebraic operations on $\QQ(\uroot_k, i)$, so it is polynomial-time computable.

        For completeness, let $X \in \CUT_2^n$.
        Then, there is $\tilde{\vec{\lambda}} : \urootset_2^n \to [0, 1]$, $\lpnorm[1]{\tilde{\vec{\lambda}}} = 1$, s.t.\
        \begin{equation}
            X = \sum_{\vec{x} \in \urootset_2^n} \tilde{\lambda}(\vec{x}) \vec{x} \vec{x}^\dagger.
        \end{equation}
        W.l.o.g.\ $x_1 = 1$ for every $\vec{x} \in \nz(\tilde{\vec{\lambda}})$, since $\vec{x} \vec{x}^\dagger$ is phase invariant.
        Construct $\tilde{\vec{\lambda}}' : \urootset_k^n \to [0, 1]$ as follows:
        Take $\tilde{\lambda}'(\vec{y}) = 0$, if $y_i \not \in \Set{1, \uroot_k}$ for some $i \in [n]$.
        Otherwise, there is a unique $\vec{x}(\vec{y}) \in \urootset_2^n$, s.t.\
        \begin{equation}
            \label{eq:identification_binary_conv_face}
            y_i = \uroot_k^{(1 - x(\vec{y})_i) / 2}
            = \begin{cases}
                1, & \textnormal{if } x(\vec{y})_i = 1, \\
                \uroot_k, & \textnormal{if } x(\vec{y})_i = -1.
            \end{cases}
        \end{equation}
        Take $\tilde{\lambda}'(\vec{y}) = \tilde{\lambda}(\vec{x}(\vec{y}))$.
        Clearly, $\lpnorm[1]{\tilde{\vec{\lambda}}'} = \lpnorm[1]{\tilde{\vec{\lambda}}} = 1$.
        Moreover,
        \begin{equation}
            \left(\sum_{\vec{y} \in \urootset_k^n} \tilde{\lambda}'(\vec{y}) \vec{y} \vec{y}^\dagger \right)_{ij}
            = \sum_{\vec{x} \in \urootset_2^n} \tilde{\lambda}(\vec{x}) \uroot_k^{- (x_i - x_j) / 2},
        \end{equation}
        and
        \begin{equation}
            \label{eq:uroot_diff_expansion}
            \begin{split}
                \uroot_k^{-(x_i - x_j) / 2}
                &= \frac{(1 - x_i) (1 + x_j)}{4} \uroot_k + \frac{(1 + x_i) (1 - x_j)}{4} \uroot_k^{-1} + \frac{(1 + x_i)(1 + x_j) + (1 - x_i)(1 - x_j)}{4} \\
                &= \frac{1}{2} \left( 1 + c_k + (1 - c_k) x_i x_j + \i s_k (x_j - x_i) \right) \\
                &= \frac{1 + c_k}{2} J_{ij} + \frac{1 - c_k}{2} (\vec{x} \vec{x}^\dagger)_{ij} + \i \frac{s_k}{2} ((\vec{x} \vec{x}^\dagger)_{1j} - (\vec{x} \vec{x}^\dagger)_{1i})
            \end{split}
        \end{equation}
        Thus,
        \begin{equation}
            \left(\sum_{\vec{y} \in \urootset_k^n} \tilde{\lambda}'(\vec{y}) \vec{y} \vec{y}^\dagger \right)_{ij} = Y_{ij},
        \end{equation}
        and $Y \in \CUT_k^n$.

        For soundness, we will use the following fact:
        For every $i \in [n]$,
        \begin{equation}
            \begin{split}
                Y_{i1} 
                &= \frac{1 + c_k}{2} + \frac{1 - c_k}{2} X_{i1} + \i \frac{s_k}{2} (X_{11} - X_{1i}) \\
                &= \frac{1 + X_{i1} + \uroot_k (1 - X_{i1})}{2} \\
                &= 1 + t_i (\uroot_k - 1),
            \end{split}
        \end{equation}
        with 
        \begin{equation}
            t_i \coloneqq (1 - X_{i1}) / 2 \in [0, 1],
        \end{equation}
        so $Y_{i1} \in \conv \Set{1, \uroot_k}$, which is a face of $\conv \urootset_k$.
        Thus, if $Y$ is a convex combination over $\urootset_k^n$ with weights $\tilde{\vec{\lambda}}$, each vector with nonzero weight $\vec{y} \in \nz(\tilde{\vec{\lambda}})$ must fulfill $y_{i} \conj{y_{1}} \in \Set{1, \uroot_k}$.
        
        Now let $Y \in \CUT_k^n$, i.e.\ there is a $\tilde{\vec{\lambda}} : \urootset_k^n \to [0, 1]$, $\lpnorm[1]{\tilde{\vec{\lambda}}} = 1$, s.t.\
        \begin{equation}
            Y = \sum_{\vec{y} \in \urootset_k^n} \tilde{\lambda}(\vec{y}) \vec{y} \vec{y}^\dagger.
        \end{equation}
        Then, again w.l.o.g., $y_1 = 1$ for each $\vec{y} \in \nz(\tilde{\vec{\lambda}})$ by the phase invariance of $\vec{y} \vec{y}^\dagger$.
        We construct $\tilde{\vec{\lambda}}' : \urootset_2^n \to [0, 1]$ by taking $\tilde{\lambda}'(\vec{x}(\vec{y})) = \tilde{\lambda}(\vec{y})$, with the unique identification from \cref{eq:identification_binary_conv_face}.
        By the previous claim, $y_i \in \Set{1, \uroot_k}$ for each $i \in [n]$ and each $\vec{y} \in \nz(\tilde{\vec{\lambda}})$, so $\lpnorm[1]{\tilde{\vec{\lambda}}'} = \lpnorm[1]{\tilde{\vec{\lambda}}} = 1$.
        Finally, by the expansion in \cref{eq:uroot_diff_expansion},
        \begin{equation}
            \Re\left((\vec{y} \vec{y}^\dagger)_{ij}\right) = \Re\left(\uroot_k^{- (x_i - x_j) / 2}\right) = \frac{1 + c_k}{2} J_{ij} + \frac{1 - c_k}{2} (\vec{x} \vec{x}^\dagger)_{ij},
        \end{equation}
        so
        \begin{equation}
            \left( \sum_{\vec{x} \in \urootset_2^n} \tilde{\lambda}'(\vec{x}) \vec{x} \vec{x}^\dagger \right)_{ij} = \frac{2}{1 - c_k} \Re \left( \sum_{\vec{y} \in \urootset_k^n} \tilde{\lambda}(\vec{y}) \vec{y} \vec{y}^\dagger  \right)_{ij} - \frac{1 + c_k}{1 - c_k} J_{ij} = X_{ij},
        \end{equation}
        and $X \in \CUT_2^n$.
        This concludes the many-one reduction from the $\CUT_2^n$ to the $\CUT_k^n$ membership problem for finite $k$.

        If $Y \in \CUT_k^n$, then the linear feasibility system $\sum_{\vec{y} \in \urootset_{k}^{n}} \tilde{\lambda}(\vec{y}) \vec{y} \vec{y}^\dagger = Y$, $\tilde{\vec{\lambda}} \geq 0$, has a basic feasible solution with $\abs{\nz(\tilde{\vec{\lambda}})} \leq n (n - 1) + 1$, whose weights solve a nonsingular linear subsystem over $\QQ(\uroot_k, \i)$.
        By Cramer's rule \citep{korte2018}, they have polynomially bounded encoding length for fixed $k$, so this sparse decomposition is a certificate that can be verified in polynomial time.
        Hence, for finite $k$, \cref{prob:cut_membership} is in \NP\ and \NP-complete.
    \end{proof}

    The strong optimization problem over $\CUT_{\infty}^n$ is known to be \NP-hard \citep{zhang2006}.
    We show that this hardness also translates to the weak optimization problem, which we will reduce to the weak membership problem.

    In the following, let distances on $\Herm_n^{\vec{1}}$ always be Frobenius distances.
    For $\varepsilon > 0$, set
    \begin{equation}
        \begin{split}
            S(\CUT_{\infty}^{n}, \varepsilon)
            &\coloneqq \Set*{ X \in \Herm_n^{\vec{1}} \given \dist_F(X, \CUT_{\infty}^{n}) \leq \varepsilon }, \\
            S(\CUT_{\infty}^{n}, - \varepsilon)
            &\coloneqq \Set*{ X \in \Herm_n^{\vec{1}} \given B_{\Herm_n^{\vec{1}}}(X, \varepsilon) \subseteq \CUT_{\infty}^{n} },
        \end{split}
    \end{equation}
    where $B_{\Herm_n^{\vec{1}}}$ is the Frobenius ball, \emph{taken within $\Herm_n^{\vec{1}}$}.
    We state the two weak problems, which we will use as intermediate steps in our reduction.

    \begin{problem}[$\CUT_{\infty}^n$ weak optimization]
        \label{prob:weak_cut_opt}
        \instance A matrix $C \in \Herm_{n}(\QQ(\i))$ and a rational $\varepsilon \in \QQ_{> 0}$.
        \task Find $X \in \Herm_{n}^{\vec{1}}(\QQ(\i))$ with $X \in S(\CUT_{\infty}^{n}, \varepsilon)$ and $\innerp{C}{X} \geq \innerp{C}{Y} - \varepsilon$ for all $Y \in S(\CUT_{\infty}^n, - \varepsilon)$.
    \end{problem}

    \begin{problem}[$\CUT_{\infty}^n$ weak membership]
        \label{prob:weak_cut_membership}
        \instance A matrix $X \in \Herm_{n}^{\vec{1}}(\QQ(\i))$ and a rational $\varepsilon \in \QQ_{> 0}$.
        \task Assert that $X \in S(\CUT_{\infty}^{n}, \varepsilon)$, or assert that $X \not \in S(\CUT_{\infty}^{n}, - \varepsilon)$.
    \end{problem}

    To show hardness of \cref{prob:weak_cut_opt}, we additionally need a bound on the inradius of $\CUT_{\infty}^{n}$.
    Luckily, such a bound naturally follows as corollary from the analysis of our algorithms.
    \begin{corollary}
        \label{th:inradius}
        $\CUT_{\infty}^{n}$ has Frobenius inradius in $\Herm_n^{\vec{1}}$ and around $\1$ of at least $r \coloneqq 1/2$.
    \end{corollary}
    \begin{proof}
        Let $X = \1 + M \in \Herm_{n}^{\vec{1}}$ with $\fnorm{X - \1} \leq 1/2$.
        Then,
        \begin{equation}
            \fnorm{M} = \fnorm{X - \1} \leq \frac{1}{2} < \frac{\pi}{4} < L_{\infty} \sqrt{\frac{n}{n-1}},
        \end{equation}
        so $1 \leq \gamma_{\mathrm{lo}}$, and the argument of \cref{th:ray_bound}, which applies verbatim to every $\gamma \leq \gamma_{\mathrm{lo}}$, gives $\1 + \RDF_{\infty}^{\elementwise -1}(M) \in \elliptope_n$, and $X = \RDF_{\infty}^{\elementwise}(\1 + \RDF_{\infty}^{\elementwise -1}(M)) \in \CUT_{\infty}^{n}$ by \cref{th:opt_upper}.
    \end{proof}

    \begin{lemma}
        \label{th:weak_opt_hardness}
        \Cref{prob:weak_cut_opt} is \NP-hard.
    \end{lemma}
    \begin{proof}
        We follow the reduction from \citep[Prop.~3.5]{zhang2006} (proof due to Tom Luo), which we restate for completeness, and then establish the gap, which is not asserted there.

        The reduction is from the following \NP-complete matrix partition problem:
        Given a matrix $G = [G_1, \dots, G_N] \in \ZZ^{M \times N}$, decide whether a subset $I \subseteq [N]$ exists, such that $\sum_{i \in I} G_i = \frac{1}{2} \sum_{i=1}^{N} G_i$.
        Its hardness follows from the fact that for $M = 1$, it recovers the famous \NP-complete partition problem, see e.g.\ \citep[p.~223]{garey1979}.
        Take $n = 2N + 1$, 
        \begin{equation}
            A \coloneqq \begin{pmatrix}
                - \vec{1}_{N} & \1_{N} & \1_{N} \\
                - \frac{1}{2} G \vec{1}_{N} & G & 0_{M \times N}
            \end{pmatrix} \in \QQ^{(M + N) \times n},
        \end{equation}
        and $Q = A^T A$.
        Then, $G$ has a matrix partition if and only if there is a $\vec{x} \in \urootset_{\infty}^{n}$ s.t.\ $\vec{x}^\dagger Q \vec{x} = 0$ (claim shown in \citep{zhang2006}).
        Notice that $\vec{x}^\dagger Q \vec{x} = \innerp{Q}{\vec{x} \vec{x} ^\dagger}$, so the claim can be stated equivalently over $\CUT_{\infty}^{n}$.
        
        We now show that if $G$ does not admit a matrix partition, then the optimization objective is lower bounded by
        \begin{equation}
            \label{eq:gap}
            \innerp{Q}{X} \geq \frac{1}{36 N \Gamma^2},
        \end{equation}
        for all $X \in \CUT_{\infty}^n$, where $\Gamma \coloneqq \max_{i} \lpnorm[2]{G_i}$.
        
        Observe first that, setting $\epsilon_i \coloneqq 1$ for $i \in I$ and $\epsilon_i \coloneqq -1$ otherwise, the subsets $I \subseteq [N]$ with $\sum_{i \in I} G_i = \frac{1}{2} \sum_{i=1}^N G_i$ correspond bijectively to the sign vectors $\vec{\epsilon} \in \Set{\pm 1}^N$ with $\sum_{i=1}^N \epsilon_i G_i = \vec{0}$. 
        Hence, $\sum_{i=1}^N \epsilon_i G_i \neq \vec{0}$ for every $\vec{\epsilon} \in \Set{\pm 1}^N$.

        Since $\innerp{Q}{\cdot}$ is linear and $\CUT_\infty^n$ is the convex hull of the matrices $\vec{x} \vec{x}^\dagger$ with $\vec{x} \in \urootset_\infty^n$, its minimum over $\CUT_\infty^n$ is attained at such a matrix, and $\innerp{Q}{\vec{x}\vec{x}^\dagger} = \vec{x}^\dagger A^T A \vec{x} = \lpnorm[2]{A\vec{x}}^2$ as $A$ is real.
        It therefore suffices to show $\lpnorm[2]{A \vec{x}}^2 \geq \frac{1}{36 N \Gamma^2}$ for all $\vec{x} \in \urootset_\infty^n$.

        Fix such a $\vec{x}$. 
        As $\lpnorm[2]{A \vec{x}}$ is invariant under $\vec{x} \mapsto \conj{x_1} \vec{x}$, we may assume $x_1 = 1$. 
        Writing $\vec{g} \coloneqq \sum_{i=1}^N G_i$ and splitting off the two row blocks of $A$,
        \begin{equation}
            \label{eq:split}
            \lpnorm[2]{A\vec{x}}^2
            = \underbrace{\sum_{j=2}^{N+1} \abs{x_j + x_{N+j} - 1}^2}_{\eqqcolon \beta}
            + \underbrace{\lpnorm[2]{\sum_{i=1}^{N} G_i x_{i+1} - \frac{1}{2} \vec{g}}^2}_{\eqqcolon \delta}.
        \end{equation}
        For $i \in [N]$ write $y_i \coloneqq x_{i+1}$ and $y'_i \coloneqq x_{N+i+1}$.
        Let $\phi_i \coloneqq \arg y_i$ and $t_i \coloneqq \cos (\phi_i) - \frac{1}{2} \in [-\frac{3}{2}, \frac{1}{2}]$ for $i \in [N]$, and set $T \coloneqq \sum_{i=1}^{N} \abs{t_i}$.

        First, we claim $\beta \geq \frac{4}{9} \sum_{j=1}^N t_j^2$.
        For each $j$ the reverse triangle inequality gives
        \begin{equation}
            \abs{y_j + y'_j - 1} = \abs{y'_j - (1 - y_j)} \geq \abs{1 - c_j},
        \end{equation}
        where $c_j \coloneqq \abs{1 - y_j}$. 
        Now $c_j^2 = 2 - 2\cos\phi_j = 1 - 2 t_j \in [0,4]$, so $c_j \leq 2$ and
        \begin{equation}
            (1 - c_j)^2 = \frac{(1-c_j^2)^2}{(1+c_j)^2} = \frac{4 t_j^2}{(1+c_j)^2}
            \geq \frac{4}{9} t_j^2.
        \end{equation}

        Now, we claim $\delta \geq \bigl( \frac{\sqrt{3}}{2} - \sqrt{3}\, \Gamma T \bigr)^2$ whenever the right-hand side is nonnegative.
        Bounding $\delta$ below by the squared norm of the imaginary part, $\delta \geq \lpnorm[2]{\sum_{i=1}^{N} G_i \sin (\phi_{i})}^2$.
        Put $h(t) \coloneqq \sqrt{1 - (\frac{1}{2} + t)^2} \geq 0$, so that $\abs{\sin (\phi_i)} = h(t_i)$, and choose $\epsilon_i \in \Set{\pm 1}$ with $\sin (\phi_i) = \epsilon_i h(t_i)$ (either sign if $h(t_i) = 0$). 
        From $h(t)^2 - \frac{3}{4} = -t(1+t)$ and $\abs{1+t} \leq \frac{3}{2}$ on $[-\frac{3}{2},\frac{1}{2}]$ we get
        \begin{equation}
            \abs*{ h(t) - \frac{\sqrt{3}}{2} } = \frac{\abs{h(t)^2 - \frac{3}{4}}}{h(t) + \frac{\sqrt{3}}{2}} \leq \frac{2}{\sqrt{3}} \abs{t} \abs{1+t} \leq \sqrt{3} \abs{t}.
        \end{equation}
        Hence, with $\vec{v} \coloneqq \sum_{i=1}^N \epsilon_i G_i \in \ZZ^M$ and
        $\vec{w} \coloneqq \sum_{i=1}^N \epsilon_i \left( h(t_i) - \frac{\sqrt{3}}{2} \right) G_i$,
        \begin{equation}
            \sum_{i=1}^N G_i \sin \phi_i = \frac{\sqrt{3}}{2} \vec{v} + \vec{w},
        \end{equation}
        and
        \begin{equation}
            \lpnorm[2]{\vec{w}} \leq \sqrt{3} \, \Gamma \sum_{i=1}^N \abs{t_i} = \sqrt{3}\, \Gamma T.
        \end{equation}
        By assumption, we have $\vec{v} \neq \vec{0}$, and $\vec{v}$ is integral, so $\lpnorm[2]{\vec{v}} \geq 1$ and the claim follows from the triangle inequality.

        Now, if $T \leq \frac{1}{4\Gamma}$, then the second claim gives 
        \begin{equation}
            \delta \geq \left( \frac{\sqrt{3}}{2} - \frac{\sqrt{3}}{4} \right)^2 = \frac{3}{16}.
        \end{equation}
        If $T > \frac{1}{4\Gamma}$, then Cauchy--Schwarz yields
        \begin{equation}
            \sum_{i=1}^N t_i^2 \geq T^2 / N > \frac{1}{16 N \Gamma^2},
        \end{equation}
        so the first claim gives 
        \begin{equation}
            \beta \geq \frac{4}{9} \cdot \frac{1}{16 N \Gamma^2} = \frac{1}{36 N \Gamma^2}.
        \end{equation}
        As $N \geq 1$ and $\Gamma \geq 1$ we have
        \begin{equation}
            \frac{1}{36 N \Gamma^2} \leq \frac{1}{36} < \frac{3}{16}, 
        \end{equation}
        so in both cases \cref{eq:split} gives $\lpnorm[2]{A \vec{x}}^2 \geq \frac{1}{36 N \Gamma^2}$.

        Now, to decide an instance $G \in \ZZ^{M \times N}$ of the matrix partition problem, we set $Q$ as before.
        If $G = 0$, then $I = \emptyset$ is a matrix partition and we answer yes.
        Otherwise, some $G_i \not = 0$, so $\Gamma \geq 1$ by integrality.
        We define
        \begin{equation}
            \varepsilon \coloneqq \frac{1}{144 N \Gamma^2 (n\kappa_Q + 1)},
        \end{equation}
        with $\kappa_Q \coloneqq \sum_{i, j} \abs{Q_{ij}} \in \QQ_{\geq 0}$, $\kappa_Q \geq \fnorm{Q}$, and note that $\Gamma^2$, $Q$ and $\varepsilon$ are rational and computable from $G$ in polynomial time, and that their encoding lengths are polynomial in that of $G$.
        We call the weak optimization oracle for $\CUT_{\infty}^{n}$ (\cref{prob:weak_cut_opt}) on the objective $C \coloneqq - Q$ with accuracy $\varepsilon$.
        This is admissible, as $\varepsilon \leq 1 / 144 < 1 / 2 = r$, with $r$ the inradius of $\CUT_{\infty}^{n}$ (\cref{th:inradius}), so $S(\CUT_{\infty}^{n}, - \varepsilon) \not = \emptyset$.
        Let $\hat{X}$ be the returned matrix and put $\nu \coloneqq \innerp{Q}{\hat{X}}$.
        We claim $G$ admits a matrix partition if and only if $\nu < \frac{1}{72 N \Gamma^2}$.

        We use that for $X' \in S(\CUT_\infty^n, \varepsilon)$ there is a $Y' \in \CUT_\infty^n$ with $\fnorm{X' - Y'} \leq \varepsilon$, whence by Cauchy--Schwarz
        \begin{equation}
            \label{eq:outer}
            \abs{ \innerp{Q}{X'} - \innerp{Q}{Y'} } \leq \varepsilon \fnorm{Q}.
        \end{equation}

        For necessity, suppose $G$ admits a matrix partition.
        By the construction, there is $\vec{x} \in \urootset_{\infty}^n$ with $\innerp{Q}{\vec{x} \vec{x}^\dagger} = 0$.
        Then, $Z \coloneqq (1 - \lambda) \vec{x} \vec{x}^\dagger + \lambda \1 \in \CUT_{\infty}^{n}$, with $\lambda = \varepsilon / r = 2 \varepsilon$, fulfills for every $H \in \Herm_{n}^{\vec{0}}$, $\fnorm{H} \leq \varepsilon$:
        \begin{equation}
            Z + H = (1 - \lambda) \vec{x} \vec{x}^\dagger + \lambda \left(\1 + \frac{1}{\lambda} H\right) \in \CUT_{\infty}^{n},
        \end{equation}
        since $\fnorm{H} / \lambda \leq r$.
        That is, $Z \in S(\CUT_{\infty}^{n}, - \varepsilon)$ and $\fnorm{Z - \vec{x} \vec{x}^\dagger} \leq 2 n \varepsilon$, so $\innerp{Q}{Z} \leq 2 n \varepsilon \fnorm{Q}$.
        The guarantee of \cref{prob:weak_cut_opt} then gives $\innerp{C}{\hat{X}} \geq \innerp{C}{Z} - \varepsilon$, that is
        \begin{equation}
            \begin{split}
                \nu \leq \innerp{Q}{Z} + \varepsilon
                &\leq (2 n \fnorm{Q} + 1) \varepsilon \\
                &\leq (2 n \kappa_{Q} + 1) \varepsilon \\
                &< 2 (n \kappa_{Q} + 1) \varepsilon = \frac{1}{72 N \Gamma^2}.
            \end{split}
        \end{equation}

        For sufficiency, suppose $G$ admits no matrix partition.
        As $\hat{X} \in S(\CUT_{\infty}^{n}, \varepsilon)$, \cref{eq:outer} and the shown gap~\eqref{eq:gap} yield a $Y \in \CUT_{\infty}^{n}$ with
        \begin{equation}
            \begin{split}
                \nu \geq \innerp{Q}{Y} - \varepsilon \fnorm{Q}
                &\geq \frac{1}{36 N \Gamma^2} - \varepsilon \fnorm{Q} \\
                &\geq \frac{1}{36 N \Gamma^2} - \varepsilon \kappa_Q \\
                &\geq \frac{3}{4} \cdot \frac{1}{36 N \Gamma^2} \\
                &> \frac{1}{72 N \Gamma^2}.
            \end{split}
        \end{equation}

        Hence, a polynomial-time algorithm for \cref{prob:weak_cut_opt} decides matrix partition, so \cref{prob:weak_cut_opt} is \NP-hard.
    \end{proof}

    Now we can finally show the \NP-hardness of the $\CUT_{\infty}^n$ membership problem.
    \begin{proof}[Proof of \cref{th:membership_hardness} ($k = \infty$)]
        By \cref{th:weak_opt_hardness}, the weak optimization problem \cref{prob:weak_cut_opt} is \NP-hard.
        Note that additionally to the inradius $r$ from \cref{th:inradius}, $\CUT_{\infty}^{n}$ has an outer Frobenius radius in $\Herm_n^{\vec{1}}$ and around $\1$ of $R \coloneqq \sqrt{n (n - 1)}$, since each $X \in \CUT_{\infty}^{n}$ fulfills $\abs{X_{ij}} \leq 1$ on all entries.
        Moreover, $\CUT_{\infty}^{n}$ is centered at $\1$.
        It is well known that for centered convex bodies $(K; n, R, r, a_0)$ that have an inradius $r$ and an outer radius $R$, the weak optimization problem reduces under a Turing reduction to the weak membership problem~\citep[Cor.~4.3.12]{grotschel1993}, with run time polynomial in $n$ and in the encoding lengths of $R$, $r$, $a_0 = \1$, and $\varepsilon$, all of which are polynomial here.
        Thus, the weak membership problem over $\CUT_{\infty}^{n}$ (\cref{prob:weak_cut_membership}) is \NP-hard.
        Finally, the weak membership problem trivially reduces to the (strong) membership problem (\cref{prob:cut_membership}).
    \end{proof}

    \section{Properties of the rounding distortion function}
    \label{sec:expectation}
    This section is concerned with proving fundamental properties of the rounding distortion function $\RDF_k$, as outlined in \cref{th:expectation}.
    We first show the weaker properties (5.\ -- 7.) for $k=2$ on real entries, before extending to general $k$ and complex entries.
    
    \begin{proof}[Proof of 5., 6., and 7.]
        Applying standard trig\-o\-no\-met\-ric identities yields a useful representation for the distortion when $k=2$, where the general expression reduces to Grothendieck's identity~\citep{goemans1995,alon2006}:
        \begin{equation}
            \RDF_2(x) = \frac{2}{\pi} \arcsin(x).
        \end{equation}
        From this, 5.\ and 6.\ become apparent.
        Moreover,
        \begin{equation}
            \abs{\arcsin(x)} \geq \abs{x},
        \end{equation}
        concluding 7.
    \end{proof}

    We now show properties 1.\ and 2.\ by straightforward calculation.
    
    \begin{proof}[Proof of 1.\ and 2.]
        We have
        \begin{equation}
            \RDF_k(0) = C_k \sum_{j=0}^{k-1} \uroot_k^{j} \arccos^2(0) = \frac{\pi^2}{4} C_k \sum_{j=0}^{k-1} \uroot_k^j = 0,
        \end{equation}
        for finite $k$, and
        \begin{equation}
            \RDF_{\infty}(0) = \frac{1}{4 \pi} \int_0^{2 \pi} \e^{\i \vartheta} \arccos^2(0) \, \dd \vartheta = \pi \int_0^{2 \pi} e^{\i \vartheta} \, \dd \vartheta = 0.
        \end{equation}

        Moreover, for finite $k$, when $\theta \in \Theta_k$, then $\e^{\i \theta} = \uroot_k^{p}$ for some $p \in \Set{0, \dots, k-1}$, so
        \begin{equation}
            \begin{split}
                \RDF_k(\e^{\i \theta})
                &= C_k \sum_{j=0}^{k-1} \uroot_k^{j} \arccos^2\left(-\Re \left( \uroot_k^{-j} \uroot_k^p \right)\right) \\
                &= C_k \sum_{j=0}^{k-1} \uroot_k^{j} \arccos^2 \left( - \cos\left( 2 \pi \frac{p-j}{k} \right) \right) \\
                &= C_k \sum_{q=0}^{k-1} \uroot_k^{p} \uroot_k^{-q} \arccos^2\left( -\cos \left( 2 \pi \frac{q}{k} \right) \right) \\
                &= \e^{\i \theta} C_k \sum_{q=0}^{k-1} \uroot_k^{-q} \pi^2 \left(1 - \frac{2q}{k}\right)^2,
            \end{split}
        \end{equation}
        with the substitution $q = p - j$.
        We can further simplify to derive
        \begin{equation}
            \begin{split}
                \sum_{q=0}^{k-1} \uroot_k^{-q} \pi^2 \left(1 - \frac{2q}{k}\right)^2
                &= \pi^2 \sum_{q=0}^{k-1} \uroot_k^{-q} - \frac{4 \pi^2}{k} \sum_{q=0}^{k-1} q \uroot_k^{-q}  + \frac{4 \pi^2}{k^2} \sum_{q=0}^{k-1} q^2 \uroot_k^{-q} \\
                &= - \frac{4 \pi^2}{\uroot_k - 1} + \frac{4 \pi^2}{\uroot_k - 1} - \frac{8 \pi^2 \uroot_k}{k^2 (\uroot_k - 1)^2} \\
                &= \frac{8 \pi^2}{k^2 \left( 2 - \uroot_k - \uroot_k^{-1} \right)} \\
                &= \frac{1}{C_k},
            \end{split}
        \end{equation}
        hence concluding
        \begin{equation}
            \RDF_k(\e^{\i \theta}) = \e^{\i \theta}.
        \end{equation}
        
        For $k = \infty$, we use the substitution $\alpha \coloneqq \vartheta - \theta$ to calculate
        \begin{equation}
            \begin{split}
                \RDF_{\infty}(\e^{\i \theta})
                &= \frac{1}{4 \pi} \int_0^{2\pi} e^{\i \vartheta} \arccos^2(- \cos(\vartheta - \theta)) \, \dd \vartheta \\
                &= \e^{\i \theta} \frac{1}{4 \pi} \int_{-\pi}^{\pi} \e^{\i \alpha} \arccos^2(- \cos(\alpha)) \, \dd \alpha \\
                &= \e^{\i \theta} \frac{1}{4 \pi} \int_{-\pi}^{\pi} \e^{\i \alpha} (\pi - |\alpha|)^2 \, \dd \alpha \\
                &= \e^{\i \theta}.
            \end{split}
        \end{equation}

        Finally, we can derive
        \begin{equation}
            \begin{split}
                \RDF_k(\conj{z})
                &= C_k \sum_{j=0}^{k-1} \uroot_k^j \arccos^2\left( - \Re \left( \uroot_k^{-j} \conj{z} \right) \right) \\
                &= C_k \sum_{j=0}^{k-1} \uroot_k^j \arccos^2\left( - \Re \left( \uroot_k^{j} z \right) \right) \\
                &= C_k \sum_{j=0}^{k-1} \uroot_k^{-j} \arccos^2\left( - \Re \left( \uroot_k^{-j} z \right) \right) \\
                &= \conj{\RDF_k(z)},
            \end{split}
        \end{equation}
        for finite $k$, and
        \begin{equation}
            \begin{split}
                \RDF_{\infty}(\conj{z})
                &= \frac{1}{4 \pi} \int_{0}^{2 \pi} \e^{\i \vartheta} \arccos^2(- r \cos(\vartheta + \phi)) \, \dd \vartheta \\
                &= \frac{1}{4 \pi} \int_{0}^{2 \pi} \e^{\i \vartheta} \arccos^2(- r \cos(- \vartheta - \phi)) \, \dd \vartheta \\
                &= \frac{1}{4 \pi} \int_{0}^{2 \pi} \e^{-\i \vartheta} \arccos^2(- r \cos(\vartheta - \phi)) \, \dd \vartheta \\
                &= \conj{\RDF_{\infty}(z)},
            \end{split}
        \end{equation}
        with $z = r \e^{\i \phi}$.
    \end{proof}

    Now, only properties 3.\ and 4.\ remain to be proven, though they turn out to be the most technical to show.
    For $k = \infty$, we will rely on a useful representation of the distortion function in terms of elliptic integrals.
    In the following, we denote by $K$ and $E$ the complete elliptic integrals of the first and second kind, given by
    \begin{equation}
        K(r) \coloneqq \int_{0}^{\pi / 2} \frac{1}{\sqrt{1 - r^2 \sin^2(\vartheta)}} \, \dd \vartheta, \qquad
        E(r) \coloneqq \int_{0}^{\pi / 2} \sqrt{1 - r^2 \sin^2(\vartheta)} \, \dd \vartheta.
    \end{equation}

    \begin{lemma}
        \label{th:elliptic_rep}
        On $z = r \e^{\i \phi}$ with $0 < r < 1$, we can express $\RDF_{\infty}$ in terms of $E$ and $K$ via
        \begin{equation}
            \RDF_{\infty}(z) = \frac{\e^{\i \phi}}{r} \left( E(r) - (1 - r^2) K(r) \right).
        \end{equation}
    \end{lemma}
    \begin{proof}
        We again make use of the substitution $\alpha \coloneqq \vartheta - \phi$ to derive
        \begin{equation}
            \RDF_{\infty}(z) = \frac{1}{4 \pi} \int_{0}^{2 \pi} \e^{i \vartheta} \arccos^2 (-r \cos(\vartheta - \phi)) \, \dd \vartheta = e^{i \phi} \frac{1}{4 \pi} \int_{- \pi}^{\pi} \e^{\i \alpha} \arccos^2(-r \cos(\alpha)) \, \dd \alpha.
        \end{equation}
        Because $\sin(\alpha) \arccos^2(-r \cos(\alpha))$ is an odd function, the imaginary part of the integral vanishes.
        We evaluate the real part via integration by parts:
        \begin{equation}
            \int_{-\pi}^{\pi} \cos(\alpha) \arccos^2(-r \cos(\alpha)) \, \dd \alpha
            = 4 r \int_{0}^{\pi} \frac{\sin^2(\alpha) \arccos(-r \cos(\alpha))}{\sqrt{1 - r^2 \cos^2(\alpha)}} \, \dd \alpha.
        \end{equation}
        The resulting integral can be split at $\pi / 2$.
        On the upper half, we apply the substitution $\beta \coloneqq \pi - \alpha$ to find
        \begin{equation}
            \int_{\pi / 2}^{\pi} \frac{\sin^2(\alpha) \arccos(-r \cos(\alpha))}{\sqrt{1 - r^2 \cos^2(\alpha)}} \, \dd \alpha
            = \int_{0}^{\pi / 2} \frac{\sin^2(\beta) \arccos(r \cos(\beta))}{\sqrt{1 - r^2 \cos^2(\beta)}} \, \dd \beta.
        \end{equation}
        Because $\arccos(x) + \arccos(-x) = \pi$, recombining yields
        \begin{equation}
            \int_{0}^{\pi} \frac{\sin^2(\alpha) \arccos(-r \cos(\alpha))}{\sqrt{1 - r^2 \cos^2(\alpha)}} \, \dd \alpha
            = \pi \int_{0}^{\pi / 2} \frac{\sin^2(\alpha)}{\sqrt{1 - r^2 \cos^2(\alpha)}} \, \dd \alpha.
        \end{equation}
        Finally, via substituting $\gamma \coloneqq \frac{\pi}{2} - \alpha$ we derive
        \begin{equation}
            \label{eq:inj_proof_R_over_r}
            \begin{split}
                \int_{0}^{\pi / 2} \frac{\sin^2(\alpha)}{\sqrt{1 - r^2 \cos^2(\alpha)}} \, \dd \alpha
                &= \int_{0}^{\pi / 2} \frac{\cos^2(\gamma)}{\sqrt{1 - r^2 \sin^2(\gamma)}} \, \dd \gamma \\
                &= \frac{1}{r^2} \int_{0}^{\pi / 2} \frac{1 - r^2 \sin^2(\gamma) - (1 - r^2)}{\sqrt{1 - r^2 \sin^2(\gamma)}} \, \dd \gamma \\
                &= \frac{1}{r^2} \left( E(r) - (1 - r^2) K(r) \right).
            \end{split}
        \end{equation}
        Combining all results, we find the desired identity
        \begin{equation}
            \RDF_{\infty}(z) = \frac{\e^{\i \phi}}{r} \left( E(r) - (1 - r^2) K(r) \right).
        \end{equation}
    \end{proof}

    We will now show property 3., that is injectivity on the closed disk $\DD$.

    \begin{proof}[Proof of 3.]
        For finite $k$, we prove the claim by showing that the expression
        \begin{equation}
            \Re\left( (\RDF_k(z_1) - \RDF_k(z_2)) \conj{(z_1 - z_2)} \right)
        \end{equation}
        does not vanish for $z_1 \not = z_2 \in \DD$.
        We define the substitutions
        \begin{equation}
            h(t) \coloneqq \arccos^2(-t),
        \end{equation}
        as well as
        \begin{equation}
            u_j \coloneqq \Re \left( \uroot_k^{-j} z_1 \right),
            \quad
            v_j \coloneqq \Re \left( \uroot_k^{-j} z_2 \right).
        \end{equation}
        Then,
        \begin{equation}
            \Re \left( (\RDF_k(z_1) - \RDF_k(z_2)) \conj{(z_1 - z_2)} \right)
            = C_k \sum_{j=0}^{k-1} (h(u_j) - h(v_j)) (u_j - v_j).
        \end{equation}
        Because $h(t)$ is strictly monotonically increasing on $t \in [-1, 1]$,
        \begin{equation}
            (h(u_j) - h(v_j))(u_j - v_j) \geq 0,
        \end{equation}
        with equality if and only if $u_j = v_j$.
        Because $k \geq 3$, the $k$-th roots of unity $\urootset_k$ span $\CC$, so there must be a $j \in \Set{0, \dots, k-1}$ such that $u_j \not = v_j$.
        Hence,
        \begin{equation}
            \Re\left((\RDF_k(z_1) - \RDF_k(z_2))\conj{(z_1 - z_2)}\right) > 0.
        \end{equation}

        For $k = \infty$, we derive from \cref{th:elliptic_rep} the radial magnitude function
        \begin{equation}
            \abs{\RDF_k(z)} = \frac{E(r) - (1 - r^2) K(r)}{r} =: R(r)
        \end{equation}
        on $0 < r < 1$.
        We show the claim, by proving that $R(r)$ is strictly monotonically increasing.
        We remind the reader of the derivatives of the complete elliptic integrals:
        \begin{equation}
            \frac{\dd K(r)}{\dd r} = \frac{E(r) - (1 - r^2) K(r)}{r (1 - r^2)}, \qquad
            \frac{\dd E(r)}{\dd r} = \frac{E(r) - K(r)}{r}.
        \end{equation}
        The derivative of the radial magnitude function is hence given by
        \begin{equation}
            \frac{\dd R(r)}{\dd r} 
            = \frac{K(r) - E(r)}{r^2}.
        \end{equation}
        Because $K(r) > E(r) > 0$ on $0 < r < 1$, this concludes that $R(r)$ is strictly monotonically increasing.
        Moreover,
        \begin{equation}
            \lim_{r \to 0^+} R(r) = 0,
            \quad
            \lim_{r \to 1^-} R(r) = 1,
        \end{equation}
        so $R(r)$ forms a bijection from $(0, 1)$ to $(0, 1)$.
        
        This allows us finally to conclude the injectivity of $\RDF_{\infty}$.
        Let $\RDF_{\infty}(z_1) = \RDF_{\infty}(z_2)$ with $z_{i} = r_{i} \e^{\i \phi_i}$.
        If 
        \begin{equation}
            \RDF_{\infty}(z_1) = \RDF_{\infty}(z_2) = 0,
        \end{equation}
        then we must have $r_1 = r_2 = 0$, so $z_1 = z_2$.
        If 
        \begin{equation}
            \abs{\RDF_{\infty}(z_1)} = \abs{\RDF_{\infty}(z_2)} = 1,
        \end{equation}
        then $r_1 = r_2 = 1$.
        But then by property 1., $\RDF_{\infty}(z_i) = z_i$, so again $z_1 = z_2$.
        Finally, if
        \begin{equation}
            0 < \abs{\RDF_{\infty}(z_1)} = \abs{\RDF_{\infty}(z_2)} < 1,
        \end{equation}
        the bijectivity of $R(r)$ implies that $r_1 = r_2$.
        Dividing by $R(r_1) = R(r_2)$ then gives, that also $\e^{\i \phi_1} = \e^{\i \phi_2}$.
        Hence, $z_1 = z_2$.
    \end{proof}

    We now show property 4., i.e.\ a linear lower bound on the magnitude $\abs{\RDF_{k}(z)}$.
    \begin{proof}[Proof of 4.]
        We first show the claim for the more technical case of $k$ being finite.
        For this, we lower bound the expression
        \begin{equation}
            \Re(\conj{z} \RDF_{k}(z)) 
            = C_k \sum_{j=0}^{k-1} x_j \arccos^2(- x_j),
        \end{equation}
        with 
        \begin{equation}
            x_j \coloneqq \Re\left(\uroot_k^{-j} z\right).
        \end{equation}
        Using the fact that $\arccos(-x) = \frac{\pi}{2} + \arcsin(x)$, we may consider three individual sums
        \begin{equation}
            \sum_{j=0}^{k-1} x_{j} \arccos^2(- x_j)
            = \frac{\pi^2}{4} \sum_{j=0}^{k-1} x_{j} + \pi \sum_{j=0}^{k-1} x_{j} \arcsin(x_j) + \sum_{j=0}^{k-1} x_j \arcsin^2(x_j).
        \end{equation}
        The first sum vanishes, as $k$-th roots of unity sum to zero.
        For the second sum, we use the fact that $x_j \arcsin(x_j) \geq x_j^2$, so
        \begin{equation}
            \begin{aligned}
                \sum_{j=0}^{k-1} x_j \arcsin(x_j)
                \geq \sum_{j=0}^{k-1} x_j^2
                &= \frac{k}{2} \abs{z}^2 + \frac{1}{4} \left( z^2 \sum_{j=0}^{k-1} \uroot_k^{-2j} + (\conj{z})^2 \sum_{j=0}^{k-1} \uroot_k^{2j} \right) \\
                &= \frac{k}{2} \abs{z}^2.
            \end{aligned}
        \end{equation}
        For even $k$, the third sum vanishes, since $x_j \arcsin^2(x_j)$ is an odd function, and 
        \begin{equation}
            x_j = - x_{j + k/2}.
        \end{equation}

        If $k$ is odd, we must bound the negative contribution of the third sum from above.
        We use the Taylor series
        \begin{equation}
            x_j \arcsin^2(x_j) = \sum_{\ell = 1}^{\infty} c_{\ell} x_j^{2 \ell + 1}
        \end{equation}
        with coefficients
        \begin{equation}
            c_{\ell} = \frac{2^{2 \ell - 1} ((\ell - 1)!)^2}{(2 \ell)!} > 0,
        \end{equation}
        which is convergent for all $\abs{x_j} \leq \abs{z} \leq 1$.
        We may consider the sum over each exponent
        \begin{equation}
            x_j^{2\ell+1}
            = \frac{1}{2^{2\ell+1}} \sum_{p=0}^{2\ell+1} \binom{2\ell+1}{p} z^p (\conj{z})^{2\ell+1 - p} \uroot_k^{j (2\ell+1 - 2p)}
        \end{equation}
        separately, giving
        \begin{equation}
            \label{eq:Lk_proof_sum_exponents}
            \sum_{j=0}^{k-1} x_j^{2\ell+1}
            = \frac{1}{2^{2\ell+1}} \sum_{p=0}^{2\ell+1} \binom{2\ell+1}{p} z^{p} \left( \conj{z} \right)^{2\ell + 1 - p} \cdot \sum_{j=0}^{k-1} \uroot_k^{j (2\ell + 1 - 2p)}.
        \end{equation}
        However,
        \begin{equation}
            \sum_{j=0}^{k-1} \uroot_k^{j(2\ell + 1 - 2p)} = \begin{cases}
                k, & \textnormal{if } 2\ell+1 \equiv 2p \mod k, \\
                0, & \textnormal{else},
            \end{cases}
        \end{equation}
        so in particular \cref{eq:Lk_proof_sum_exponents} vanishes for $2\ell + 1 < k$.
        Hence, we may bound
        \begin{equation}
            \begin{split}
                \abs*{\sum_{j=0}^{k-1} x_j \arcsin^2(x_j)}
                &= \abs*{\sum_{j=0}^{k-1} \sum_{\ell = (k - 1)/2}^{\infty} c_{\ell} x_j^{2 \ell + 1} } \\
                &\leq \sum_{j=0}^{k-1} \abs*{\sum_{\ell=(k-1)/2}^{\infty} c_{\ell} x_j^{2\ell+1}} \\
                &\eqqcolon  \sum_{j=0}^{k-1} \abs{S_k(x_j)}
            \end{split}
        \end{equation}

        For odd $k \geq 5$, we can exploit a surplus in the bound of the second sum.
        That is, we define the remainder term
        \begin{equation}
            R(x_j) \coloneqq x_j (\arcsin(x_j) - x_j) \geq 0,
        \end{equation}
        such that the second sum evaluates to
        \begin{equation}
            \pi \sum_{j=0}^{k-1} x_j \arcsin(x_j) = \pi \frac{k}{2} \abs{z}^2 + \pi \sum_{j=0}^{k-1} R(x_j).
        \end{equation}
        Then, if the surplus dominates the bound of the third sum, we derive an equivalent bound as for even $k$.
        That is, we want to show that
        \begin{equation}
            \pi \sum_{j=0}^{k-1} R(x_j) \geq \sum_{j=0}^{k-1} \abs{S_k(x_j)},
        \end{equation}
        which we will do by showing the summand-wise bound
        \begin{equation}
            \pi R(x_j) \geq \abs{S_k(x_j)}.
        \end{equation}
        As $R(x_j)$ is even and $S_k(x_j)$ is odd, both $\pi R(x_j)$ and $S_k(x_j)$ are even functions, and it is sufficient to show the bound on $x_j \in [0, 1]$, where 
        \begin{equation}
            \abs{S_k(x_j)} = S_k(x_j) \geq 0.
        \end{equation}
        Since $k \geq 5$, we may express
        \begin{equation}
            \label{eq:s_k_bound_greater_three}
            S_k(x_j) = x_j \arcsin^2(x_j) - \sum_{\ell=1}^{(k-3)/2} c_{\ell} x_j^{2 \ell + 1} \leq x_j \arcsin^2(x_j) - x_j^3,
        \end{equation}
        allowing to directly compute
        \begin{equation}
            \begin{split}
                \pi R(x_j) - S_k(x_j)
                &\geq \pi x_j (\arcsin(x_j) - x_j) - x_j \arcsin^2(x_j) + x_j^3 \\
                &= x_j (\pi - \arcsin(x_j) - x_j) (\arcsin(x_j) - x_j) \\
                &\geq 0
            \end{split}
        \end{equation}
        on $x_j \in [0, 1]$.

        Only for $k=3$ does \cref{eq:s_k_bound_greater_three} not hold, so we must directly bound
        \begin{equation}
            \sum_{j=0}^{k-1} \abs{S_3(x_j)}
            \leq \sum_{j=0}^{k-1} \sum_{\ell = 1}^{\infty} c_{\ell} \abs*{x_j^{2 \ell + 1}}
            \leq \sum_{\ell = 1}^{\infty} c_{\ell} \sum_{j=0}^{k-1} x_j^2
            = \pi \frac{k}{2} \abs{z}^2 \cdot \frac{\pi}{4}.
        \end{equation}

        Thus, for all $k \in \NN_{\geq 3}$, we may combine the derived bounds to
        \begin{equation}
            \begin{split}
                \sum_{j=0}^{k-1} x_j \arccos^2(-x_j)
                &\geq \pi \frac{k}{2} \abs{z}^2 \begin{cases}
                    1, & \textnormal{if } k > 3, \\
                    1 - \frac{\pi}{4}, & \textnormal{if } k = 3,
                \end{cases} \\
                &= \frac{L_k}{C_k} \abs{z}^2.
            \end{split}
        \end{equation}
        Thus, we may conclude
        \begin{equation}
            \Re(\conj{z} \RDF_k(z)) \geq L_k \abs{z}^2,
        \end{equation}
        and
        \begin{equation}
            \abs{\RDF_k(z)} \geq \frac{\Re(\conj{z} \RDF_k(z))}{\abs{z}} = L_k \abs{z}.
        \end{equation}

        For the case of $k = \infty$, we consider the function
        \begin{equation}
            \frac{R(r)}{r} = \int_{0}^{\pi / 2} \frac{\cos^2(\vartheta)}{\sqrt{1 - r^2 \sin^2(\vartheta)}} \, \dd \vartheta,
        \end{equation}
        which we derived in \cref{eq:inj_proof_R_over_r}.
        It is apparent, that this function is monotonically increasing on $r \in (0, 1)$, so it must assume its infimum at
        \begin{equation}
            \lim_{r \to 0^+} \frac{R(r)}{r} = \int_0^{\pi / 2} \cos^2(\vartheta) \dd \vartheta = \frac{\pi}{4}.
        \end{equation}
        Hence,
        \begin{equation}
            \abs{\RDF_{\infty}(z)} \geq \frac{\pi}{4} \abs{z}
        \end{equation}
        for all $0 < \abs{z} < 1$.
        Because the bound is also satisfied for $\abs{z} = 0$ and $\abs{z} = 1$, this concludes the proof.
    \end{proof}

    Finally, we show 5., which turns out to be a corollary of the other properties.
    \begin{proof}[Proof of 5.]
        $\RDF_k(\DD)$ is compact, hence closed, so it suffices to cover the open disk $D_k \coloneqq \Set{\abs{w} < L_k}$.
        Since $\RDF_k$ is continuous and injective on a compact set, it is a homeomorphism onto its image.
        By invariance of domain $U \coloneqq \RDF_k(\interior \DD)$ is open with $\partial U \subseteq \RDF_k(\partial \DD)$.
        Since $\abs{\RDF_k(z)} \geq L_k$, for $\abs{z} = 1$, $D_k \cap \RDF_k(\partial \DD) = \emptyset$, so $D_k \cap \partial U = \emptyset$.
        Since $D_k$ is connected and $0 = \RDF_k(0) \in D_k \cap U$, this forces $D_k \subseteq U$.
    \end{proof}

    \section{Bounds to the probability of feasibility for the informed LP algorithm} 
    \label{sec:feasibility_proofs}
    In this section we will derive \cref{th:feasibility_window}.
    We consider the i.i.d.\ centered correlation vectors
    \begin{equation}
        \vec{y}^{(\ell)} \coloneqq \vec{\chi}(\vec{x}^{(\ell)}) - \hat{\gamma} \vec{m} \in V \cong \RR^{D_{\mathrm{eff}}}
    \end{equation}
    with distribution $\nu$ and convex hull
    \begin{equation}
        K_s \coloneqq \conv \Set{\vec{y}^{(1)}, \dots, \vec{y}^{(s)}}.
    \end{equation}
    Then, \cref{eq:conjecture_lp} is robustly feasible if and only if
    \begin{equation}
        \vec{0} \in \interior_V K_s.
    \end{equation}
    We call the corresponding event
    \begin{equation}
        A \coloneqq \Set{\vec{0} \in \interior_V K_s}.
    \end{equation}
    Moreover, we use the support function
    \begin{equation}
        \label{eq:support_function_h}
        h(\vec{v}) \coloneqq \max_{\ell} \innerp{\vec{v}}{\vec{y}^{(\ell)}},
    \end{equation}
    so that for every $\rho > 0$, $\bar{B}(0, \rho) \subseteq K_s$ if and only if $h(\tilde{\vec{v}}) \geq \rho$ for all unit $\tilde{\vec{v}} \in V$.

    Replacing the distribution $\nu$ with one that is centrally symmetric about the origin and almost surely yields samples in general position gives the problem introduced by Wendel in \citep{wendel1962}.
    In that case, it is well known that the failure probability is given by
    \begin{equation}
        \PP[\vec{0} \not \in K_s] = 2^{-s + 1} \sum_{i = 0}^{D_{\mathrm{eff}} - 1} \binom{s - 1}{i} \eqqcolon W(s, D_{\mathrm{eff}}).
    \end{equation}
    In fact, over absolutely continuous distributions, $W(s, D_{\mathrm{eff}})$ serves as a lower bound to the failure probability \citep{wagner2001}.
    
    Treating $W(s, D_{\mathrm{eff}})$ as the cumulative distribution function of a Binomial random variable $B \sim \binomial(s - 1, 1/2)$ yields
    \begin{equation}
        W(s, D_{\mathrm{eff}}) = \PP[B \leq D_{\mathrm{eff}} - 1] = 1 - \PP[B \geq D_{\mathrm{eff}}].
    \end{equation}
    Note that $\EE[B] = (s - 1) / 2$.
    By Hoeffding's inequality, for $s \geq 2$,
    \begin{equation}
        \PP[B \geq D_{\mathrm{eff}}] = \PP\left[B - \EE[B] \geq D_{\mathrm{eff}} - \frac{s-1}{2}\right] \leq \exp \left( - \frac{(2 D_{\mathrm{eff}} - s + 1)^2}{2 (s - 1)} \right).
    \end{equation}
    Hence, if $s \leq (2 - \eta) D_{\mathrm{eff}}$,
    \begin{equation}
        \label{eq:wendel_lower_bound}
        W(s, D_{\mathrm{eff}}) \geq 1 - \exp \left( - \frac{\eta^2 D_{\mathrm{eff}}}{4} \right).
    \end{equation}

    \begin{proof}[Proof of \cref{th:feasibility_window} (i).]
        We show the claim by upper bounding $\PP[\vec{0} \in \interior_V K_s].$
        If $\vec{y}^{(\ell)}$ were absolutely continuous, then by \citet{wagner2001}
        \begin{equation}
            \PP[\vec{0} \in K_s] \leq 1 - W(s, D_{\mathrm{eff}}).
        \end{equation}
        Our distribution is not absolutely continuous (discrete for finite $k$).
        However, we will show via smoothing, that the same bound holds in our case, if we replace $K_s$ by $\interior_V K_s$.

        For $\rho > 0$ put
        \begin{equation}
            A_{\rho} \coloneqq \Set*{\bar{B}(0, \rho) \subseteq K_s},
        \end{equation}
        where $\bar{B}$ is the closed ball with respect to $V$.
        $A_\rho$ is measurable, as via the support function $h$ we have
        \begin{equation}
            A_{\rho} = \Set*{\min_{\lpnorm[2]{\tilde{\vec{v}}} = 1} h(\tilde{\vec{v}}) \geq \rho},
        \end{equation}
        and $h$ is continuous in $\vec{y}^{(1)}, \dots, \vec{y}^{(s)}$.
        Moreover, $A = \bigcup_{j \in \NN} A_{1 / j}$ is an increasing union, so $A$ is measurable as well.
        By continuity of measure from below,
        \begin{equation}
            \label{eq:feasibility_neccessity_prob_union}
            \PP[A] = \sup_{\rho > 0} \PP[A_{\rho}].
        \end{equation}

        Fix $\rho > 0$ and $t > 0$, and let $\vec{g}^{(1)}, \dots, \vec{g}^{(s)}$ be i.i.d.\ standard Gaussian vectors on $\RR^{D_{\mathrm{eff}}}$, independent of all $\vec{y}^{(\ell)}$.
        We define
        \begin{equation}
            \vec{y}^{(\ell)}(t) \coloneqq \vec{y}^{(\ell)} + t \vec{g}^{(\ell)}
        \end{equation}
        and $G \coloneqq \max_{\ell} \lpnorm[2]{\vec{g}^{(\ell)}}$.
        Then, $\vec{y}^{(\ell)}(t)$ are i.i.d.\ with distribution ${\nu*\mathcal{N}(0, t^2 \1)}$, which has a density.
        Hence, the bound of \citet{wagner2001} applies and
        \begin{equation}
            \PP[\vec{0} \in \conv \Set{\vec{y}^{(1)}(t), \dots, \vec{y}^{(s)}(t)} ]
            \leq 1 - W(s, D_{\mathrm{eff}}).
        \end{equation}
        Let $h_{t}(\vec{v}) \coloneqq \max_{\ell} \innerp{\vec{v}}{\vec{y}^{(\ell)}(t)}$.
        Then, for all unit vectors $\tilde{\vec{v}}$,
        \begin{equation}
            h_{t}(\tilde{\vec{v}}) \geq h(\tilde{\vec{v}}) - t G.
        \end{equation}
        If $\bar{B}(0, \rho) \subseteq K_s$, then $h(\tilde{\vec{v}}) \geq \rho$, and thus either $t G \geq \rho$ or $h_{t}(\tilde{\vec{v}}) > 0$ for all unit~$\tilde{\vec{v}}$.
        Consequently,
        \begin{equation}
            A_{\rho} \subseteq \Set{ \vec{0} \in \conv \Set{\vec{y}^{(1)}(t), \dots, \vec{y}^{(s)}(t)} } \cup \Set{t G \geq \rho},
        \end{equation}
        and
        \begin{equation}
            \PP[A_{\rho}] \leq 1 - W(s, D_{\mathrm{eff}}) + \PP[G \geq \rho / t].
        \end{equation}
        Since the distribution of $G$ does not depend on $t$, and $G < \infty$ almost surely,
        \begin{equation}
            \PP[G \geq \rho / t] \to 0
        \end{equation}
        as $t \to 0$ at fixed $\rho$.
        Hence,
        \begin{equation}
            \PP[A_{\rho}] \leq 1 - W(s, D_{\mathrm{eff}})
        \end{equation}
        for every $\rho > 0$, and \cref{eq:feasibility_neccessity_prob_union} together with \cref{eq:wendel_lower_bound} gives the claim.
    \end{proof}

    To show part (ii) of \cref{th:feasibility_window}, we first prove that $\beta > 0$ on all instances.

    \begin{lemma}
        \label{th:positive_beta}
        Unconditionally, the positivity index fulfills $\beta > 0$.
        Moreover, for $D_{\mathrm{eff}} \geq 1$, $\beta \leq 1 / 2$.
    \end{lemma}
    \begin{proof}
        W.l.o.g.\ assume $D_{\mathrm{eff}} \geq 1$ and set
        \begin{equation}
            g(\vec{v}) \coloneqq \PP[\innerp{\vec{v}}{\vec{y}} > 0]
        \end{equation}
        such that 
        \begin{equation}
            \beta = \inf \Set{ g(\tilde{\vec{v}}) \given \tilde{\vec{v}} \in V, \ \lpnorm[2]{\tilde{\vec{v}}} = 1 }.
        \end{equation}
        Let $(\vec{v}_{j})_{j \in \NN}$ be any sequence in $\RR^{D}$ with $\vec{v}_{j} \to \vec{v}$, and let $\varepsilon > 0$.
        $\lpnorm[2]{\vec{\chi}(\vec{x})} = \sqrt{m}$ deterministically, and $\hat{\gamma}\lpnorm[2]{\vec{m}}=\lpnorm[2]{\EE[\vec{\chi}(\vec{x})]} \leq \sqrt{m}$, so $\lpnorm[2]{\vec{y}} \leq 2\sqrt{m}$.
        Thus,
        \begin{equation}
            \abs{\innerp{\vec{v}_j - \vec{v}}{\vec{y}}} \leq 2 \sqrt{m} \lpnorm[2]{\vec{v}_j - \vec{v}}
        \end{equation}
        almost surely.
        Hence, for every $j$ with $\lpnorm[2]{\vec{v}_j - \vec{v}} \leq \varepsilon / (2 \sqrt{m})$, we have the almost sure inclusion of events
        \begin{equation}
            \Set{\innerp{\vec{v}}{\vec{y}} > 2 \varepsilon} \subseteq \Set{\innerp{\vec{v}_j}{\vec{y}} > \varepsilon}.
        \end{equation}
        Since $\varepsilon > 0$, this gives
        \begin{equation}
            g(\vec{v}_j) \geq \PP[\innerp{\vec{v}_j}{\vec{y}} > \varepsilon] \geq \PP[\innerp{\vec{v}}{\vec{y}} > 2 \varepsilon]
        \end{equation}
        for all large $j$, so
        \begin{equation}
            \liminf_{j \to \infty} g(\vec{v}_j) \geq \PP[\innerp{\vec{v}}{\vec{y}} > 2 \varepsilon]
        \end{equation}
        for every $\varepsilon > 0$.
        The events $\Set{\innerp{\vec{v}}{\vec{y}} > 2 \varepsilon}$ increase to $\Set{\innerp{\vec{v}}{\vec{y}} > 0}$ as $\varepsilon \to 0$, so continuity of measure from below yields
        \begin{equation}
            \liminf_{j \to \infty} g(\vec{v}_j) \geq g(\vec{v}).
        \end{equation}
        Hence, $g$ is lower semicontinuous.

        Let $\tilde{\vec{v}} \in V$ with $\lpnorm[2]{\tilde{\vec{v}}} = 1$.
        Then $\EE[\innerp{\tilde{\vec{v}}}{\vec{y}}] = 0$ and $\EE[\innerp{\tilde{\vec{v}}}{\vec{y}}^2] = \tilde{\vec{v}}^T \Sigma \tilde{\vec{v}} > 0$, because $\tilde{\vec{v}} \in (\ker \Sigma)^{\bot}$.
        If $\PP[\innerp{\tilde{\vec{v}}}{\vec{y}} > 0] = 0$, then $\innerp{\tilde{\vec{v}}}{\vec{y}} \leq 0$ almost surely, which, together with the vanishing expectation, forces $\innerp{\tilde{\vec{v}}}{\vec{y}} = 0$ almost surely, a contradiction.

        A lower semicontinuous function on the compact unit sphere of $V$ attains its infimum, so $\beta = g(\tilde{\vec{v}}^*) > 0$ for some unit vector $\tilde{\vec{v}}^* \in V$, yielding the lower bound.
        For the upper bound, notice that $g(\tilde{\vec{v}}) + g(- \tilde{\vec{v}}) \leq 1$, hence $\beta \leq 1/2$.
    \end{proof}

    The sufficiency bound on $s$ is derived via the well-known epsilon net theorem, which we briefly recall.
    A \emph{range space} $(X, R)$ consists of a set $X$ and a family $R$ of subsets of $X$, called \emph{ranges}.
    A finite set $Y \subseteq X$ is \emph{shattered} by $R$ if every subset of $Y$ is of the form $Y \cap r$ for some $r \in R$.
    The \emph{VC dimension} of $(X, R)$ is the largest cardinality of a finite set shattered by $R$, or infinite if no such maximum exists.
    Given a probability measure $\nu$ on $X$ and $\varepsilon > 0$, a set $N \subseteq X$ is an \emph{$\varepsilon$-net} of $(X, R)$ for $\nu$ if $N \cap r \neq \emptyset$ for every $r \in R$ with $\nu(r) \geq \varepsilon$.
    See \citep{matousek2002} for details on $\varepsilon$-nets and VC dimensions.

    \begin{theorem}[Epsilon net theorem]
        \label{th:eps_net}
        Let $(X, R)$ be a range space of finite VC dimension $D \geq 1$, let $\nu$ be a probability measure on $X$, let $\varepsilon > 0$, and let $s \geq 8 / \varepsilon$ with $2 s \geq D$.
        Then the set $N$ of $s$ random independent draws from $\nu$ fails to be an $\varepsilon$-net of $(X, R)$ for $\nu$ with probability at most
        \begin{equation}
            2 \left( \frac{2 \e s}{D} \right)^{D} 2^{- \varepsilon s / 2}.
        \end{equation}
    \end{theorem}
    \noindent
    The proof is due to \citet{haussler1987} for $\nu$ uniform on a finite set, and to \citet{blumer1989} for arbitrary $\nu$, subject to mild measurability conditions on $R$, which we take for granted for the half-spaces considered below.
    Both state the bound for the ranges with $\nu(r) > \varepsilon$, and with $\Phi_D(2 s)$ in place of $(2 \e s / D)^D$, where $\Phi_D(n) \coloneqq \sum_{i=0}^{D} \binom{n}{i}$.
    For $n \geq D \geq 1$, $\Phi_D(n) \leq (\e n / D)^D$ \citep[Prop.~A2.1]{blumer1989}.
    Moreover, the bound of \citet{blumer1989} holds for $s \geq 2 / \varepsilon'$, so it applies with every $\varepsilon' \in [\varepsilon / 4, \varepsilon)$, and letting $\varepsilon' \to \varepsilon$ yields the stated version for the ranges with $\nu(r) \geq \varepsilon$.

    \begin{proof}[Proof of \cref{th:feasibility_window} (ii).]
        $\beta > 0$ by \cref{th:positive_beta}.
        If $\vec{0} \not \in \interior_V K_s$, then by the supporting hyperplane theorem applied inside $V$, there is a homogeneous open half-space $\mathcal{H} \subseteq V$ s.t.\
        \begin{equation}
            \mathcal{H} \cap \Set{\vec{y}^{(1)}, \dots, \vec{y}^{(s)}} = \emptyset.
        \end{equation}
        Moreover, the set of homogeneous open half-spaces over $V \cong \RR^{D_{\mathrm{eff}}}$ has VC dimension $D_{\mathrm{eff}}$.
        By definition of $\beta$, every such $\mathcal{H}$ has $\nu(\mathcal{H}) \geq \beta$.
        Thus, applying the epsilon net theorem~(\cref{th:eps_net}) with $\varepsilon = \beta$ gives, whenever $s \geq 8 / \beta$ and $2 s \geq D_{\mathrm{eff}}$,
        \begin{equation}
            \PP[\vec{0} \not \in \interior_V K_s] \leq 2 \left( \frac{2 \e s}{D_{\mathrm{eff}}} \right)^{D_{\mathrm{eff}}} 2^{- \beta s / 2}.
        \end{equation}

        Let $\delta \in (0, 1)$ and let
        \begin{equation}
            \label{eq:sufficient_s}
            s \geq 6 \beta^{-1} (D_{\mathrm{eff}} \ln(12 / \beta) + \ln(2 / \delta)).
        \end{equation}
        Then, since $\ln(12 / \beta) \geq \ln(24) > 3$, $s > 18 D_{\mathrm{eff}} / \beta$, so that $s > 8 / \beta$ and $2 s > D_{\mathrm{eff}}$.
        It remains to show, that $\PP[\vec{0} \not \in \interior_V K_s] \leq \delta$, i.e.\
        \begin{equation}
            \label{eq:feasibility_prob_bound}
            D_{\mathrm{eff}} \ln \left( \frac{2 \e s}{D_{\mathrm{eff}}} \right) + \ln (2 / \delta) \leq \frac{\beta s}{2} \ln(2).
        \end{equation}
        By concavity of the logarithm, $\ln(x) \leq \alpha x - \ln(\alpha) - 1$ for all $x, \alpha > 0$.
        Taking $x = 2 \e s / D_{\mathrm{eff}}$ and $\alpha = \beta \ln(2) / (8 \e)$ yields
        \begin{equation}
            D_{\mathrm{eff}} \ln\left( \frac{2 \e s}{D_{\mathrm{eff}}} \right) \leq \frac{\beta s}{4} \ln(2) + D_{\mathrm{eff}} \ln \left(\frac{8}{\ln(2) \beta}\right).
        \end{equation}
        Moreover, by \cref{eq:sufficient_s},
        \begin{equation}
            \frac{\beta s}{4} \ln(2) \geq \frac{6 \ln(2)}{4} \left( D_{\mathrm{eff}} \ln\left(\frac{12}{\beta}\right) + \ln(2 / \delta) \right) \geq D_{\mathrm{eff}} \ln\left( \frac{8}{\ln(2) \beta} \right) + \ln(2 / \delta),
        \end{equation}
        yielding \cref{eq:feasibility_prob_bound}.
    \end{proof}

    \section{Numerical characterization of the feasibility transition} 
    \label{sec:feasibility_numerics}
    \Cref{th:feasibility_window}~(i) rules out robust feasibility below $s = 2 D_{\mathrm{eff}}$, whereas part~(ii) guarantees it only for a considerably larger number of pulses, and only in terms of the positivity index $\beta$.
    In this appendix we locate the transition numerically and illustrate how the $\beta$ dependency can impact the necessary number of pulses.

    A connectivity graph on $n$ vertices is drawn from one of six families: the complete graph, a complete bipartite graph, an Erdős--Rényi graph, a random $3$-regular graph, a periodic square lattice, or a random spanning tree.
    The entries of $M$ on the edges of the graph are drawn from one of several laws.
    For $k = 2$, where $M$ is real, these are Rademacher signs, standard normal values, uniform values on $[-1, 1]$, Student-$t$ values with three degrees of freedom, and constant values.
    For complex $k$, they are unit moduli with uniform phase, unit moduli with phase drawn from $\Theta_k$ for finite $k$, standard complex normal values, or Student-$t$ moduli with three degrees of freedom and uniform phase. 
    On the lattice, we additionally use the Hofstadter phases of \cref{sec:fermion_application}.
    Constant entries on the complete graph make $\1 + M = \vec{1} \vec{1}^{\dagger}$ a vertex of $\CUT_k^n$, which a single pulse realizes.
    Such instances were excluded.
    Finally, $M$ is completed Hermitian and normalized to $\lpnorm[\infty]{M} = 1$.
    This base ensemble covers $n \in \Set{24, 28, 32, 36}$ and the pulse types $k = 2, 3, 4, \infty$, and is complete for every instance with $D \leq 630$. 

    To examine how the transition depends on the saturation ${\varsigma \coloneqq \hat{\gamma} \lpnorm[\infty]{M}}$, we further draw \num{20} instances on the complete graph with $n = 24$ and $k = 4$.
    Each of them is interpolated toward the extreme point $P_{ij} = z_i \conj{z_j}$ of a fixed $\vec{z} \in \urootset_k^n$, which drives up the saturation.

    For all instances, we run the ray binary search (\cref{alg:ray_binary_search}), sample and round pulses on the resulting $X$, and record the smallest $s$ for which the \ac{LP}~\eqref{eq:restricted_lp} admits a solution with $\lpnorm[1]{\vec{\lambda}} \leq 1 / \hat{\gamma}$.
    \Cref{fig:feasibility} shows the resulting success probability against $s / D$, for the base ensemble at every pulse type and for the interpolated instances at $k = 4$.
    On the base ensemble, the transition sits at $s \approx 2 D$ for all four pulse types (see the inset), so that the necessary condition of \cref{th:feasibility_window}~(i) is essentially also sufficient, and the much larger pulse numbers of part~(ii) are not observed.
    For the interpolated instances, the transition instead moves sharply to the right as the saturation grows, reaching $s \approx 7 D$ at $\varsigma \approx 0.8$ and exceeding $s = 12 D$ for $\varsigma \gtrsim 0.9$.

    This behavior is expected.
    As $\varsigma \to 1$, the elliptope point approaches $X = \vec{z} \vec{z}^{\dagger}$ and the rounding distribution concentrates on the single pulse $\vec{z}$.
    The sampled correlation vectors $\vec{\chi}(\vec{x})$ then cluster tightly around their mean $\hat{\gamma} \vec{m}$, so that along most directions only a small fraction of the pulses reach past the target.
    Hence, the positivity index $\beta$ is small, and by \cref{th:feasibility_window}~(ii) the sufficient number of pulses grows as $\beta$ shrinks.
    Such instances are, however, atypical: 
    Across the base ensemble, the saturation stays well below the regime in which the shift sets in.
    We do not claim, that a high saturation is necessary for a small $\beta$.

    For this ensemble, the ray binary search is run to a bracket width of \num{1e-12}, which places the target further along the ray and makes the linear program harder to satisfy.
    Success is decided by a variant of the linear program with a slack variable on each equality constraint, which is feasible by construction and whose optimal residual vanishes precisely when a solution with $\lpnorm[1]{\vec{\lambda}} \leq 1 / \hat{\gamma}$ exists.
    In practice, this residual stays below \num{e-9} on feasible and above \num{e-4} on infeasible pulse sets.
    The smallest successful $s$ is located by bisection over nested prefixes of a single pulse stream, repeated independently at least \num{25} times per instance.
    Of the \num{140} interpolated instances, one at $\varsigma \approx \num{0.8}$ is missing, because several of its linear programs stalled in the simplex method.

    \begin{figure}
        \centering
        \includegraphics{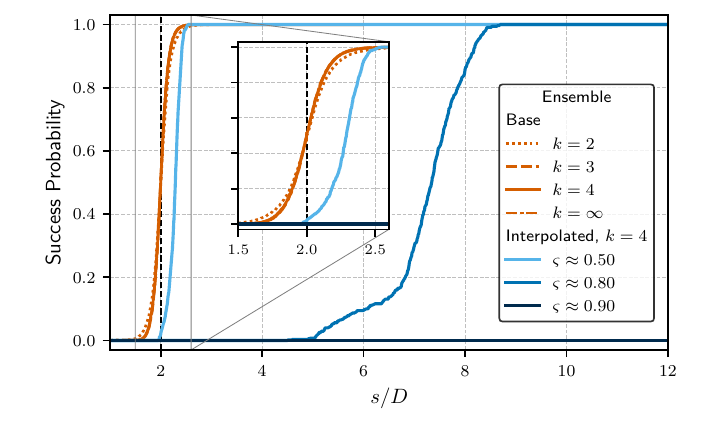}
        \caption[The feasibility transition and its dependence on the saturation.]{
            The feasibility transition and its dependence on the saturation.
            Plotted is the probability that the informed LP algorithm (\cref{alg:informed_lp}) returns a solution with quantum run time $\lpnorm[1]{\vec{\lambda}} \leq 1 / \hat{\gamma}$, against the number of sampled pulses per dimension $s / D$.
            The base ensemble curves collect all instances with $D \leq 630$, one curve per pulse type $k$.
            For these, the saturation $\varsigma = \hat{\gamma} \lpnorm[\infty]{M}$ averages $\varsigma \approx \num{0.3}$.
            The remaining curves are the interpolated instances on the complete graph with $n = 24$ and $k = 4$, at the stated saturation.
            The vertical dashed line marks $s = 2 D$.
            The inset magnifies the range $1.5 \leq s / D \leq 2.6$, marked by the gray box, in which the curves of the base ensemble for the four pulse types lie almost on top of each other.
            For $\varsigma \approx \num{0.9}$, no run succeeded within the sampled range, so the curve is constant at zero.
        }
        \label{fig:feasibility}
    \end{figure}

    \section{Sufficient pulses for the mixed LP algorithm}
    \label{sec:mixed_proofs}
    \Cref{th:mixed_algo} provides guarantees for the feasibility of the mixed LP algorithm (\cref{alg:mixed}) if a sufficient number of pulses $s = s_{\mathrm{i}} + s_{\mathrm{u}} \in \LandauO(m)$ are sampled.
    Its proof relies on \cref{th:residual,th:uniform_ball}.
    \Cref{th:residual} states that if sufficiently many informed pulses are generated, then the scaled sample mean of their action will be close to the effective target.
    That is, $\vec{r} = \vec{m} - \frac{1}{\hat{\gamma}} \mean{\vec{\chi}(\vec{x})}$ will be small with high probability.
    Conversely, \cref{th:uniform_ball} assures that for sufficiently many uniform pulses, a ball with constant radius $\rho_0$ around the origin is likely to be enclosed in the convex hull $\conv \vec{\chi}(S_{\mathrm{u}})$.
    In this section we prove both \cref{th:residual,th:uniform_ball}.

    \begin{proof}[Proof of \cref{th:residual}.]
        All samples are i.i.d, so $\EE\left[\mean{\vec{\chi}(\vec{x})}\right] = \EE[\vec{\chi}(\vec{x})] = \hat{\gamma} \vec{m}$.
        Hence,
        \begin{equation}
            \PP[\lpnorm[2]{\vec{r}} \geq t] = \PP\left[\lpnorma[2]{\mean{\vec{\chi}(\vec{x})} - \EE\left[\mean{\vec{\chi}(\vec{x})}\right]} \geq \hat{\gamma} t\right].
        \end{equation}
        The sample mean is a function of $s_{\mathrm{i}}$ i.i.d.\ random variables and has bounded differences $2 \sqrt{m} / s_{\mathrm{i}}$, since for any two $\vec{x}, \vec{y} \in \urootset_k^n$, $\lpnorm[2]{\vec{\chi}(\vec{x}) - \vec{\chi}(\vec{y})} \leq 2 \sqrt{m}$.
        Moreover, by Jensen's inequality,
        \begin{equation}
            \EE\left[ \lpnorma[2]{\mean{\vec{\chi}(\vec{x})} - \EE\left[\mean{\vec{\chi}(\vec{x})}\right]} \right]
            \leq \sqrt{\EE\left[ \lpnorma[2]{\mean{\vec{\chi}(\vec{x})} - \EE\left[\mean{\vec{\chi}(\vec{x})}\right]}^2 \right]}
            = \sqrt{ \frac{\Tr(\Sigma)}{s_{\mathrm{i}}} }
            \leq \sqrt{\frac{m}{s_{\mathrm{i}}}},
        \end{equation}
        since 
        \begin{equation}
            \Tr(\Sigma) = \EE[\lpnorm[2]{\vec{\chi}(\vec{x})}^2] - \lpnorm[2]{\EE[\vec{\chi}]}^2 = m - \hat{\gamma}^2 \lpnorm[2]{\vec{m}}^2 \leq m.
        \end{equation}
        Thus, for $t \geq \frac{\sqrt{m}}{\hat{\gamma} \sqrt{s_{\mathrm{i}}}}$, McDiarmid's inequality applies and yields
        \begin{equation}
            \PP\left[\lpnorma[2]{\mean{\vec{\chi}(\vec{x})} - \EE\left[\mean{\vec{\chi}(\vec{x})}\right]} \geq \hat{\gamma} t\right]
            \leq \exp \left( - \frac{s_{\mathrm{i}} (\hat{\gamma} t - \sqrt{m / s_{\mathrm{i}}})^2}{2 m} \right).
        \end{equation}
        This shows the claim.
    \end{proof}

    To show \cref{th:uniform_ball}, we use the epsilon net theorem~(\cref{th:eps_net}) together with a VC dimension argument, similar to the proof of \cref{th:feasibility_window} (see \cref{sec:feasibility_proofs}).
    For this purpose, we again study the random projection $\innerp{\vec{v}}{\vec{\chi}(\vec{x})}$, $\vec{v} \in \RR^{D}$, but now for $\vec{x}$ uniform over $\urootset_k^n$.
    Recall from \cref{th:uniform_ball} the definitions $c = m / D$ and $\kappa = 15$.
    We bound the first, second, and fourth moment of $\innerp{\vec{v}}{\vec{\chi}(\vec{x})}$.

    \begin{lemma}
        \label{th:Z_bounded_moments}
        For every unit vector $\tilde{\vec{v}} \in \RR^D$, if $\vec{x}$ is sampled uniformly over $\urootset_k^n$, then $Z \coloneqq \innerp{\tilde{\vec{v}}}{\vec{\chi}(\vec{x})}$ has moments $\EE[Z] = 0$, $\EE[Z^2] = c$, and $\EE[Z^4] \leq \kappa c^2$.
    \end{lemma}
    \begin{proof}
        For an edge $(i, j) \in \nz(A)$, $i < j$, write $z_{\Set{i,j}} \coloneqq x_i \conj{x_j}$, so that $\chi^{\Re}_{\Set{i,j}} = \Re (z_{\Set{i,j}})$ and $\chi^{\Im}_{\Set{i,j}} = \Im (z_{\Set{i,j}})$.
        Conditioned on $x_j$, $z_{\Set{i,j}}$ is a rotation of the uniform $x_i$ by a group element, hence uniform on $\urootset_k$ itself.
        Two edge variables on distinct edges are independent:
        On disjoint edges they are functions of disjoint entries, and on edges sharing a vertex $i$ they are, conditionally on $x_i$, functions of two distinct independent entries and each conditionally uniform, so their joint law does not depend on $x_i$.
        In particular $\EE[\vec{\chi}(\vec{x})] = \vec{0}$, and by linearity $\EE[Z] = 0$.

        For uniform $z \in \urootset_k$ we have $\EE[z] = 0$, and $\EE[z^2] = 0$ whenever $k \in \NN_{\geq 3} \cup \Set{\infty}$.
        With $\abs{z}^2 = 1$, $\Re(z^2) = \Re (z)^2 - \Im (z)^2$, and $\Im(z^2) = 2 \Re (z) \Im (z)$ this gives
        $\EE[\Re (z)^2] = \EE[\Im (z)^2] = 1/2$ and $\EE[\Re (z) \Im (z)] = 0$, whereas for $k = 2$ the variable
        $z = \pm 1$ is real with $\EE[z^2] = 1$.
        Pairwise independence makes all entries of the covariance across distinct edges vanish, so
        \begin{equation}
            \Cov \vec{\chi}(\vec{x}) = \begin{cases}
                \1_D, & \textnormal{if } k = 2, \\
                \frac{1}{2} \1_D, & \textnormal{else},
            \end{cases}
            = c \1_D.
        \end{equation}
        Thus, $\EE[Z^2] = \Var Z = \tilde{\vec{v}}^T \left(\Cov \vec{\chi}(\vec{x})\right) \tilde{\vec{v}} = c$.

        Recall that we label the vector entries in $\RR^D$ with the elements of $\nz(A)$, as given by \cref{eq:correlation_var}.
        Hence, we may associate $\tilde{\vec{v}}$ with a hollow Hermitian matrix $Q \in \Herm_n^{\vec{0}}$.
        For $(i, j) \in \nz(A)$, $i < j$, we take
        \begin{equation}
            Q_{ij} \coloneqq \tilde{v}^{\Re}_{\Set{i,j}} - \i \tilde{v}^{\Im}_{\Set{i,j}}, \qquad
            Q_{ji} \coloneqq \tilde{v}^{\Re}_{\Set{i,j}} + \i \tilde{v}^{\Im}_{\Set{i,j}},
        \end{equation}
        and $Q_{ij} = 0$ otherwise, such that $\fnorm{Q}^2 = 2 \lpnorm[2]{\tilde{\vec{v}}}^2 = 2$, and $Z = \frac{1}{2} \sum_{i \not = j} Q_{ij} x_i \conj{x_j}$.

        To compute $\EE[Z^4]$, we abbreviate $q \coloneqq \sum_{i < j} \abs{Q_{ij}}^4 \in (0, 1]$ and
        \begin{equation}
            S_4 \coloneqq \sum_{a,b,c,d \textnormal{ distinct}} Q_{ab} Q_{bc} Q_{cd} Q_{da} = \Tr(Q^4) - 2 \sum_{i=1}^{n} (Q^2)_{ii}^2 + 2q.
        \end{equation}
        Since $(Q^2)_{ii}^2 \geq \sum_{j=1}^{n} \abs{Q_{ij}}^4$, we have $\sum_{i=1}^{n} (Q^2)_{ii}^2 \geq 2q$, and $\Tr(Q^4) = \fnorm{Q^2}^2 \leq \snorm{Q}^2 \fnorm{Q}^2 \leq \fnorm{Q}^4 = 4$.
        Hence
        \begin{equation}
            S_4 \leq 4 - 4q + 2q = 4 - 2q.
        \end{equation}

        For $k \not = 2$, we may associate $x_i \conj{x_j}$ with an arc $i \to j$.
        For a set of such arcs, we define the \emph{degree difference} vector $\vec{d} \in \ZZ^n$ by
        \begin{equation}
            d_i \coloneqq \#\textnormal{arcs leaving } i - \#\textnormal{arcs entering } i.
        \end{equation}
        For the product of a set of $\ell$ arcs $\Set{i_1 \to j_1, \dots, i_{\ell} \to j_{\ell}}$ with degree difference vector $\vec{d}$ we write 
        \begin{equation}
            \prod_{\tau = 1}^{\ell} x_{i_{\tau}} x_{j_{\tau}}^* = \prod_{i = 1}^{n} x_i^{d_i} \eqqcolon \vec{x}^{\vec{d}}.
        \end{equation}
        Then, by independence of the entries of $\vec{x}$, $\EE[\vec{x}^{\vec{d}}] = 1$ if $k \mid d_i$ for every $i \in [n]$ (meaning $\vec{d} = \vec{0}$ if $k = \infty$), and $\EE[\vec{x}^{\vec{d}}] = 0$ otherwise.
        We write $k \mid \vec{d}$ if the condition holds.

        Now, we may express $Z^4$ as a sum of products of four arcs, yielding the fourth moment
        \begin{equation}
            \EE[Z^4] = \frac{1}{16} \sum_{\substack{i_1 \to j_1 \\ \cdots \\ i_4 \to j_4}} \prod_{\tau = 1}^{4} Q_{i_{\tau} j_{\tau}} \EE\left[ \prod_{\tau = 1}^{4} x_{i_{\tau}} \conj{x_{j_{\tau}}} \right]
        \end{equation}
        By the previous observation, only sets of four arcs with degree vector $\vec{d}$ fulfilling $k \mid \vec{d}$ contribute to this sum.
        Since $\abs{d_i} \leq 4$, $\sum_{i=1}^{n} \abs{d_i} \leq 8$, and $\sum_{i=1}^{n} d_i = 0$, either $\vec{d} = \vec{0}$ or $k \in \Set{3, 4}$ and $\vec{d} = \pm k (\vec{e}_i - \vec{e}_j)$ for some $i \not = j$, and $\vec{e}_i, \vec{e}_j$ canonical basis vectors (three or more nonzero entries would need $\sum_{i=1}^{n} \abs{d_i} \geq 4k > 8$).

        If $\vec{d} = 0$, then the arcs decompose either into two cycles of length $2$, or one cycle of length $4$.
        Two $2$-cycles on distinct edges give $4!$ orderings of the arcs, with each pair of edges arising from exactly two ordered pairs $(\Set{i_1, j_1}, \Set{i_2, j_2})$, two $2$-cycles on the same vertex pair gives $4! / (2! 2!) = 6$ orderings, and a $4$-cycle gives $4!$ orderings with each arc set arising from exactly four ordered quadruples $(a, b, c, d)$.
        Therefore, the total contribution of sets of arcs with $\vec{d} = 0$ is
        \begin{equation}
            \begin{split}
                &12 \sum_{\substack{i_1 < j_1, \ i_2 < j_2 \\ i_1 \not = i_2 \textnormal{ or } j_1 \not = j_2}} \abs{Q_{i_1 j_1}}^2 \abs{Q_{i_2 j_2}}^2 + 6 \sum_{i < j} \abs{Q_{ij}}^4 + 6 \sum_{a, b, c, d \textnormal{ distinct}} Q_{ab} Q_{bc} Q_{cd} Q_{da} \\
                &\quad = 12 (1-q) + 6 q + 6 S_4 \\
                &\quad \leq 36 - 18q.
            \end{split}
        \end{equation}
        Hence, for $k \in \NN_{\geq 5} \cup \Set{\infty}$, we can already bound the fourth moment as
        \begin{equation}
            \EE[Z^4] \leq \frac{36 - 18q}{16} \leq \frac{9}{4} = 9c^2.
        \end{equation}

        If $k = 3$ and $\vec{d} = 3 (\vec{e}_i - \vec{e}_j)$, the arc set is of the form $\Set{i \to j, i \to j, i \to v, v \to j}$ with some $v \not \in \Set{i, j}$, which gives $4! / 2! = 12$ possible orderings.
        Combining with the $12$ orderings from $\vec{d} = -3 (\vec{e}_i - \vec{e}_j)$ gives contribution
        \begin{equation}
            12 \sum_{i,j,v \textnormal{ distinct}} Q_{ij}^2 Q_{iv} Q_{vj} = 12 \sum_{i \not = j} Q_{ij}^2 (Q^2)_{ij}.
        \end{equation}
        By Cauchy--Schwarz,
        \begin{equation}
            \abs*{ \sum_{i \not = j} Q_{ij}^2 (Q^2)_{ij} } \leq \sqrt{\sum_{i \not = j} \abs{Q_{ij}}^4} \sqrt{\sum_{i \not = j} \abs{(Q^2)_{ij}}^2} \leq \sqrt{2q} \fnorm{Q^2} \leq 2 \sqrt{2q}.
        \end{equation}
        Thus, together with the contribution from $\vec{d} = \vec{0}$, the fourth moment is bounded by
        \begin{equation}
            \EE[Z^4] \leq \frac{24 \sqrt{2q} + 36 - 18 q}{16} \leq \frac{13}{4} = 13c^2,
        \end{equation}
        where the second inequality is derived by maximizing over $q \in (0, 1]$ (attained at $q = 8/9$).

        If $k = 4$ and $\vec{d} = 4(\vec{e}_i - \vec{e}_j)$, then $i$ is the tail and $j$ is the head of all four arcs and the contribution is
        \begin{equation}
            \sum_{i \not = j} Q_{ij}^4 \leq \sum_{i \not = j} \abs{Q_{ij}}^4 = 2q.
        \end{equation}
        Hence, together with the contribution from $\vec{d} = \vec{0}$, we bound
        \begin{equation}
            \EE[Z^4] \leq \frac{2q + 36 - 18q}{16} \leq \frac{9}{4} = 9c^2.
        \end{equation}

        The remaining case is $k = 2$.
        We switch from the directed arcs $i \to j$ to undirected edges $\Set{i, j}$, since $\conj{x_i} = x_i$.
        Then, $Z = \sum_{i < j} Q_{ij} x_i x_j$ with $Q_{ij} \in \RR$ and for four edges $\Set{i_1, j_1}, \dots, \Set{i_4, j_4}$, $\EE[\prod_{\tau = 1}^{4} x_{i_\tau} x_{j_\tau}] = 1$ exactly when every vertex has even degree, and $0$ otherwise.
        The even multigraphs on four edges are $\Set{e, e, e, e}$ (one ordering), $\Set{e, e, e', e'}$ ($4! / (2!2!) = 6$ orderings, with each pair $\Set{e, e'}$ arising from exactly two ordered pairs) and 4-cycles ($4!$ orderings, each arising from exactly eight ordered quadruples $(a, b, c, d)$).
        Hence,
        \begin{equation}
            \begin{split}
                \EE[Z^4]
                &= \sum_{i < j} Q_{ij}^4 + 3 \sum_{\substack{i_1 < j_1, \ i_2 < j_2 \\ i_1 \not = i_2 \textnormal{ or } j_1 \not = j_2}} Q_{i_1 j_1}^2 Q_{i_2 j_2}^2 + 3 \sum_{a, b, c, d \textnormal{ distinct}} Q_{ab} Q_{bc} Q_{cd} Q_{da} \\
                &= q + 3 (1-q) + 3 S_4 \\
                &\leq 15 - 8q \leq 15 = 15 c^2.
            \end{split}
        \end{equation}
        This concludes that for every $k$ and every $\tilde{\vec{v}}$, $\EE[Z^4] \leq \kappa c^2 = 15 c^2$.
    \end{proof}

    The bounded moments now allow us to lower bound the probability of $\innerp{\tilde{\vec{v}}}{\vec{\chi}(\vec{x})}$ exceeding $\rho_0 = \frac{\sqrt{c}}{4 \sqrt{\kappa}}$.

    \begin{lemma}
        \label{th:prob_over_rho}
        For every unit vector $\tilde{\vec{v}} \in \RR^D$, if $\vec{x}$ is sampled uniformly over $\urootset_k^n$, then $Z \coloneqq \innerp{\tilde{\vec{v}}}{\vec{\chi}(\vec{x})}$ fulfills
        \begin{equation}
            \PP[Z \geq \rho_0] \geq \frac{4^{-4 / 3}}{\kappa}.
        \end{equation}
    \end{lemma}
    \begin{proof}
        We have the moment bounds $\EE[Z] = 0$, $\EE[Z^2] = c$, and $\EE[Z^4] \leq \kappa c^2$ by \cref{th:Z_bounded_moments}.
        By Hölder with exponents $3/2$ and $3$,
        \begin{equation}
            c = \EE[Z^2] = \EE\left[\abs{Z}^{2 / 3} \abs{Z}^{4 / 3}\right] \leq (\EE[\abs{Z}])^{2 / 3} (\EE[Z^4])^{1 / 3},
        \end{equation}
        so 
        \begin{equation}
            \EE[\abs{Z}] \geq \frac{\EE[Z^2]^{3/2}}{\EE[Z^4]^{1/2}} \geq \sqrt{\frac{c}{\kappa}}.
        \end{equation}
        Since $\EE[Z] = 0$, we moreover have $\EE[Z_+] = \frac{1}{2} \EE[\abs{Z}] \geq \frac{\sqrt{c}}{2 \sqrt{\kappa}}$, where $Z_+ \coloneqq \max(Z, 0)$.

        Now, let $A \coloneqq \Set{Z \geq \rho_0}$ be the desired event with indicator variable $\indicator_A$.
        Then, pointwise $Z_+ \leq Z \indicator_A + \rho_0$, since on $A$ one has $Z_+ = Z \indicator_A$, and otherwise $Z_{+} < \rho_0$.
        Using Hölder with exponents $4$ and $4/3$, we can upper bound
        \begin{equation}
            \EE[Z \indicator_A] \leq \EE[Z^4]^{1 / 4} \EE[\indicator_A]^{3/4} \leq \sqrt{c} \kappa^{1/4} \EE[\indicator_A]^{3/4}.
        \end{equation}
        Hence,
        \begin{equation}
            \begin{split}
                \PP[Z \geq \rho_0]
                = \EE[\indicator_A]
                &\geq \frac{\EE[Z \indicator_A]^{4/3}}{c^{2/3} \kappa^{1/3}}\\
                &\geq \frac{(\EE[Z_+] - \rho_0)^{4/3}}{c^{2/3} \kappa^{1/3}} \\
                &\geq \frac{c^{2/3}}{4^{4/3}c^{2/3} \kappa} = \frac{4^{-4/3}}{\kappa}.
            \end{split}
        \end{equation}
    \end{proof}

    With \cref{th:prob_over_rho}, we can finally state the proof of \cref{th:uniform_ball}.

    \begin{proof}[Proof of \cref{th:uniform_ball}.]
        Recall that we consider the set $S_{\mathrm{u}}$ of $s_{\mathrm{u}}$ uniformly sampled $\vec{x} \in \urootset_k^n$.

        Let the failure event be 
        \begin{equation}
            F = \Set*{\min_{\lpnorm[2]{\tilde{\vec{v}}} = 1} \max_{\vec{x} \in S_{\mathrm{u}}} \innerp{\tilde{\vec{v}}}{\vec{\chi}(\vec{x})} < \rho_0}.
        \end{equation}
        On the complement of $F$ the support function of $\conv \vec{\chi}(S_{\mathrm{u}})$ satisfies
        \begin{equation}
            h(\tilde{\vec{v}}) = \max_{\vec{x} \in S_{\mathrm{u}}} \innerp{\tilde{\vec{v}}}{\vec{\chi}(\vec{x})} \geq \rho_0
        \end{equation}
        for every unit vector $\tilde{\vec{v}}$, and since $\conv \vec{\chi}(S_{\mathrm{u}})$ is compact and convex,
        \begin{equation}
            \conv \vec{\chi}(S_{\mathrm{u}}) = \bigcap_{\tilde{\vec{v}}} \Set{\vec{\chi} \given \innerp{\tilde{\vec{v}}}{\vec{\chi}} \leq h(\tilde{\vec{v}})} \supseteq \bar{B}(\vec{0}, \rho_0).
        \end{equation}
        Hence, we aim to upper bound $\PP[F]$.

        Consider the closed, off-center half-spaces
        \begin{equation}
            \mathcal{H}_{\tilde{\vec{v}}} \coloneqq \Set{\vec{\chi} \in \RR^D \given \innerp{\tilde{\vec{v}}}{\vec{\chi}} \geq \rho_0}.
        \end{equation}
        We claim that the set of such half-spaces has VC dimension at most $D$:
        If the closed, off-center half-spaces shattered $D+1$ points, then general affine half-spaces would shatter those $D+1$ points in addition to the origin, a contradiction as it is well known that they have VC dimension exactly $D + 1$ \citep{matousek2002}.

        By \cref{th:prob_over_rho}, 
        \begin{equation}
            \PP[\vec{\chi}(\vec{x}) \in \mathcal{H}_{\tilde{\vec{v}}}] \geq \frac{4^{-4/3}}{\kappa} \eqqcolon \varepsilon
        \end{equation}
        for every unit vector $\tilde{\vec{v}}$.
        Moreover, if $F$, then there is some $\tilde{\vec{v}}$, s.t.\ $\mathcal{H}_{\tilde{\vec{v}}} \cap \vec{\chi}(S_{\mathrm{u}}) = \emptyset$.
        Thus, applying \cref{th:eps_net} gives, whenever $s_{\mathrm{u}} \geq 8 / \varepsilon$ and $2 s_{\mathrm{u}} \geq D$,
        \begin{equation}
            \PP[F] \leq 2 \left( \frac{2 \e s_{\mathrm{u}}}{D} \right)^{D} 2^{- \varepsilon s_{\mathrm{u}} / 2}.
        \end{equation}

        $s_0 \coloneqq 200 \kappa (D + \ln(1 / \delta))$ satisfies $s_0 > 8 / \varepsilon$ and $2 s_0 > D$, so it is left to show that for $s_{\mathrm{u}} \geq s_0$,
        \begin{equation}
            2 \left( \frac{2 \e s_{\mathrm{u}}}{D} \right)^{D} 2^{- \varepsilon s_{\mathrm{u}} / 2} \leq \delta,
        \end{equation}
        that is
        \begin{equation}
            \frac{\varepsilon s_{\mathrm{u}}}{2} \ln (2) - D \ln \left( \frac{2 \e s_{\mathrm{u}}}{D} \right) - \ln (2) - \ln (1 / \delta) \geq 0.
        \end{equation}
        Differentiating with respect to $s_{\mathrm{u}}$ yields
        \begin{equation}
            \frac{\varepsilon \ln(2)}{2} - \frac{D}{s_{\mathrm{u}}} \geq 0,
        \end{equation}
        for $s_{\mathrm{u}} \geq s_0 > \frac{2 D}{\varepsilon \ln(2)}$, so it suffices to check on $s_0$.
        Using $D \ln(\frac{D + \ln(1 / \delta)}{D}) \leq \ln(1 / \delta)$, we bound
        \begin{equation}
            \begin{split}
                &\frac{\varepsilon s_0}{2} \ln (2) - D \ln \left( \frac{2 \e s_0}{D} \right) - \ln (2) - \ln (1 / \delta) \\
                &\quad \geq \frac{100}{4^{4/3}} \ln(2) (D + \ln(1 / \delta)) - D \ln \left( \frac{400 \kappa \e (D + \ln(1 / \delta))}{D} \right) - \ln(2) - \ln(1 / \delta) \\
                &\quad \geq D \left( \frac{100}{4^{4 / 3}} \ln(2) - \ln(800 \kappa \e) \right) + \ln(1 / \delta) \left( \frac{100}{4^{4/3}} \ln(2) - 2 \right).
            \end{split}
        \end{equation}
        Computing the values yields $\frac{100}{4^{4/3}} \ln(2) > 10.9$ and $\ln (800 \kappa \e) < 10.4$, so both brackets are positive, which completes the proof.
    \end{proof}

    \section{Numerical methods}
    \label{sec:numerics}
    This appendix collects the implementation details of the algorithms of \cref{sec:approx_algo} and the protocol underlying the benchmarks of \cref{sec:applications}.
    The implementation of all algorithms described here is part of the accompanying code~\citep{friese2026}.

    The ray binary search (\cref{alg:ray_binary_search}), and hence both linear program algorithms, requires the evaluation of $\RDF_k^{\elementwise -1}$ on the entries of $\gamma M$.
    By \cref{th:expectation}, $\RDF_k$ is injective on $\DD$ and its image contains the disk of radius $L_k$, so this inverse is well-defined wherever the search evaluates it.
    For $k = 2$ it is available in closed form, since $\RDF_2(x) = \frac{2}{\pi} \arcsin(x)$ and therefore $\RDF_2^{-1}(x) = \sin \left( \frac{\pi}{2} x \right)$.
    The same holds for $k = 4$, where
    \begin{equation}
        \RDF_4(z) = \frac{2}{\pi} \arcsin(\Re (z)) + \i \frac{2}{\pi} \arcsin(\Im (z))
    \end{equation}
    separates into its real and imaginary part, so that
    \begin{equation}
        \RDF_4^{-1}(w) = \sin \left( \frac{\pi}{2} \Re (w) \right) + \i \sin \left( \frac{\pi}{2} \Im (w) \right).
    \end{equation}
    For $k = 3$ and $k = \infty$ no closed form of the inverse is available, and we instead precompute $\RDF_k^{-1}$ itself on a grid and interpolate between the tabulated values.
    The grid consists of \num{1000} equally spaced values of $\Re (w)$ and of $\Im (w)$ over $[-1, 1]$ each, and every entry is obtained by solving $\RDF_k(z) = w$ with a hybrid Powell root finder seeded at $z = w$.
    Grid points with $w \notin \RDF_k(\DD)$, for which no solution exists, are marked as undefined.
    An evaluation is then a bilinear interpolation between the four surrounding grid points, and returns undefined whenever one of them is, in which case the search treats $\gamma M$ as having left the image of $\RDF_k$.
    The tabulation is thereby conservative by at most one grid cell near the boundary of $\RDF_k(\DD)$.

    Only a bounded part of the grid is ever queried.
    Cauchy interlacing applied to the principal submatrix of $M$ on $\Set{i, j}$ yields $\mu_{\min}(M) \leq - \abs{M_{ij}}$ for every pair $i \neq j$, so the initial upper bracket $\gamma_{\mathrm{hi}} = -1 / \mu_{\min}(M)$ of \cref{alg:ray_binary_search} already enforces $\lpnorm[\infty]{\gamma M} \leq 1$ throughout the search.
    Since the grid is computed once per $k$, and every subsequent evaluation is a table lookup, this does not affect the $\LandauO(n^3 \log(1 / \tau))$ cost of the ray binary search, which remains dominated by the positive semidefiniteness test.

    We run the search to a bracket width of $\tau = \num{1e-9}$ and accept $\gamma$ if and only if the smallest eigenvalue of $\1 + \RDF_k^{\elementwise -1}(\gamma M)$ is at least $\tau$.
    The returned elliptope point $X$ is therefore positive definite, so that the factorization $X = L L^{\dagger}$ required by the sampling step of \cref{alg:informed_lp,alg:mixed} can be computed by a Cholesky decomposition.

    The linear programs are solved with the simplex method of HiGHS~\citep{huangfu2018} through CVXPY~\citep{diamond2016}.
    We solve them in the equivalent scale invariant form, which normalizes the pulse weights to a probability vector and maximizes the scale $\alpha$ of the target, so that $\lpnorm[1]{\vec{\lambda}} = 1 / \alpha$.
    A simplex method is used in favor of an interior point method, because only a basic solution attains the minimal support $\nz(\vec{\lambda}) \approx D_{\mathrm{eff}} + 1$, which is the number of pulses that are actually applied.

    The exact optimum of the \ac{LP}~\eqref{eq:general_lp} is obtained by solving the same program over the full pulse set $\Theta_k^n$, which has $k^n$ variables.
    This is possible up to the $q = \num{20}$ qubit systems of \cref{sec:qubit_application} with roughly \num{e6} variables, and for the qudit benchmark of \cref{sec:qudit_application}, where fixing the phase of the dummy index leaves at most $5^8 \approx \num{3.9e5}$ variables, but not for the fermionic benchmark of \cref{fig:hofstadter}.
\end{document}

%% file: definitions.tex
\renewcommand{\vec}[1]{\pmb{#1}}

\newcommand{\uroot}{\omega}
\newcommand{\e}{\ensuremath\mathrm{e}} 
\renewcommand{\i}{\ensuremath\mathrm{i}} 

\DeclareMathOperator{\Tr}{Tr} 
\renewcommand{\Re}{\operatorname{Re}} 
\renewcommand{\Im}{\operatorname{Im}} 

\DeclareMathOperator{\ran}{ran} 
\DeclareMathOperator{\rank}{rank} 

\DeclareMathOperator{\conv}{conv} 
\DeclareMathOperator{\interior}{int} 

\DeclareMathOperator{\nz}{nz}

\newcommand{\1}{\mathds{1}} 

\newcommand{\EE}{\mathbb{E}} 
\newcommand{\PP}{\mathbb{P}} 
\newcommand{\Var}{\mathrm{Var}}

\newcommand{\indicator}{\textbf{I}} 

\DeclareMathOperator{\binomial}{Binomial}

\newcommand{\elementwise}{\circ} 

\newcommand{\dd}{\mathrm{d}} 

\newcommand{\mean}[1]{\overline{#1}}

\newcommand{\conj}[1]{#1^*} 

\newcommand{\CC}{\mathbb{C}}
\newcommand{\RR}{\mathbb{R}}
\newcommand{\KK}{\mathbb{K}}
\newcommand{\QQ}{\mathbb{Q}}
\newcommand{\ZZ}{\mathbb{Z}}
\newcommand{\NN}{\mathbb{N}}

\newcommand{\DD}{\mathbb{D}}

\newcommand{\urootset}{\mathcal{B}} 

\newcommand{\CUT}{\mathrm{CUT}} 
\newcommand{\elliptope}{\mathcal{E}} 

\newcommand{\Sym}{\mathrm{Sym}} 
\newcommand{\Herm}{\mathrm{Herm}} 

\DeclareMathOperator{\RDF}{\mathrm{RDF}}

\let\mod\relax
\DeclareMathOperator{\mod}{mod}

\DeclareMathOperator{\Cov}{\mathrm{Cov}}

\DeclareMathOperator{\dist}{\mathrm{dist}}

\DeclareMathOperator{\LandauO}{\mathrm{O}} 

\newcommand{\class}[1]{{\ensuremath{\mathsf{#1}}}}

\newcommand{\NP}{\class{NP}}

\DeclarePairedDelimiterX{\abs}[1]{\lvert}{\rvert}{%
  \ifblank{#1}{\,\cdot\,}{#1}
}   

\DeclarePairedDelimiterX\norm[1]\lVert\rVert{%
  \ifblank{#1}{\,\cdot\,}{#1}
}   

\newcommand{\lpnorm}[2][p]{\norm{#2}_{\ell_{#1}}}   
\newcommand{\lpnorma}[2][p]{\norm*{#2}_{\ell_{#1}}}   

\newcommand{\pnorm}[2][p]{\norm{#2}_{#1}} 

\DeclarePairedDelimiterX{\iiiNorm}[1]{\lvert}{\rvert}{%
  \delimsize\lvert\delimsize\lvert#1\delimsize\rvert\delimsize\rvert%
}

\DeclarePairedDelimiterXPP\snorm[1]{}\lVert\rVert{_\infty}{\ifblank{#1}{\,\cdot\,}{#1}}   

\DeclarePairedDelimiterXPP\twonorm[1]{}\lVert\rVert{_2}{\ifblank{#1}{\,\cdot\,}{#1}}   

\DeclarePairedDelimiterXPP\trnorm[1]{}\lVert\rVert{_1}{\ifblank{#1}{\,\cdot\,}{#1}}   

\DeclarePairedDelimiterXPP\fnorm[1]{}\lVert\rVert{_2}{\ifblank{#1}{\,\cdot\,}{#1}}   

\DeclarePairedDelimiterXPP\dnorm[1]{}\lVert\rVert{_\diamond}{\ifblank{#1}{\,\cdot\,}{#1}}   

\DeclarePairedDelimiterXPP\cbnorm[1]{}\lVert\rVert{_\mathrm{cb}}{\ifblank{#1}{\,\cdot\,}{#1}}   
\DeclarePairedDelimiterXPP\onenorm[1]{}\lVert\rVert{_{1\rightarrow 1}}{\ifblank{#1}{\,\cdot\,}{#1}}   
\DeclarePairedDelimiterXPP\ddnorm[1]{}\lVert\rVert{_{\diamond\rightarrow \diamond}}{\ifblank{#1}{\,\cdot\,}{#1}}   
\DeclarePairedDelimiterXPP\ssnorm[1]{}\lVert\rVert{_{\infty\rightarrow\infty}}{\ifblank{#1}{\,\cdot\,}{#1}}   

\providecommand\given{}
\newcommand\SetSymbol[1][]{%
  \nonscript\:#1\vert
  \allowbreak
  \nonscript\:
  \mathopen{}}
\DeclarePairedDelimiterX\Set[1]\{\}{%
  \renewcommand\given{\SetSymbol[\delimsize]}
  #1
}

\DeclarePairedDelimiterX\innerp[2]{\langle}{\rangle}{%
  \ifblank{#1}{\,\cdot\,}{#1} , \ifblank{#2}{\,\cdot\,}{#2}%
}

\DeclarePairedDelimiter{\ket}{\vert}{\rangle}

\DeclarePairedDelimiterX\braket[2]{\langle}{\rangle}%
  {#1\kern0.15ex\delimsize\vert\kern0.15ex\mathopen{}#2}

\DeclarePairedDelimiterX\ketbra[2]{\vert}{\vert}%
  {#1\kern0.15ex\delimsize\rangle\delimsize\langle\kern0.15ex\mathopen{}#2}

\DeclarePairedDelimiterX\sandwich[3]{\langle}{\rangle}%
  {#1\,\delimsize\vert\kern0.15ex\mathopen{}#2\kern0.15ex\delimsize\vert\kern0.15ex\mathopen{}#3}

\DeclarePairedDelimiterX\obraket[2]{(}{)}%
  {#1\kern0.15ex\delimsize\vert\kern0.15ex\mathopen{}#2}

\DeclarePairedDelimiterX\oketbra[2]{\vert}{\vert}%
  {#1\kern0.15ex\delimsize)\delimsize(\kern0.15ex\mathopen{}#2}

\DeclarePairedDelimiterX\osandwich[3]{(}{)}%
  {#1\,\delimsize\vert\kern0.15ex\mathopen{}#2\kern0.15ex\delimsize\vert\kern0.15ex\mathopen{}#3}
